\UseRawInputEncoding                   
\documentclass[aos, preprint]{imsart}
\pdfoutput=1
\usepackage{amsthm, amsmath, amssymb, amsfonts, graphicx, caption,subcaption,color}

\usepackage{tikz}
\usetikzlibrary{arrows}
\RequirePackage[square,authoryear]{natbib}  
\RequirePackage{algorithmic}
\RequirePackage[section]{algorithm}
\usepackage{xr-hyper} 

\usepackage[colorlinks,citecolor=blue,urlcolor=blue,filecolor=blue,linkcolor=blue]{hyperref}   
\usepackage{booktabs}
\usepackage{multirow}
\usepackage{pifont}

\usepackage{enumerate}   
\usepackage[toc,page]{appendix}

\usepackage{float}
\usepackage{arydshln}

\newtheorem{proposition}{Proposition}[section]
\newtheorem{lemma}{Lemma}[section]
\newtheorem{theorem}{Theorem}[section]
\newtheorem{corollary}{Corollary}[section]
\newtheorem{remark}{Remark}[section]
\newtheorem{definition}{Definition}[section]

\newtheorem{example}{Example}
\newtheorem*{assumption*}{Assumption}
\newtheorem{assumption}{Assumption}[section]    

\newenvironment{proofof}[1][]{\begin{trivlist}
   \item[\hskip \labelsep {\bfseries Proof of #1.}]}{\end{trivlist}}

\newcommand{\Prob}{\mathbb{P}}
\newcommand{\R}{\mathbb{R}}

\DeclareMathOperator{\Exp}{\mathbb{E}}   
\DeclareMathOperator{\Cov}{Cov}
\DeclareMathOperator{\Var}{Var}

\DeclareMathOperator*{\MEC}{MEC}

\newcommand{\ignore}[1]{}

\def\boxit#1{\vbox{\hrule\hbox{\vrule\kern6pt  \vbox{\kern6pt#1\kern6pt}\kern6pt\vrule}\hrule}}

\def\tcr{\textcolor{red}}

\def\bX{\mathbf{X}}

\def\bXtil{\widetilde{\bX}}

\def\Ytil{\widetilde{Y}}
\def\bv{\mathbf{v}}

\def\bZ{\mathbf{Z}}

\def\bmu{\boldsymbol{\mu}}

\def\bM{\mathbf{M}}

\def\bSigma{\boldsymbol{\Sigma}}

\def\bbeta{\boldsymbol{\beta}}
\def\bbetaS{\bbeta_S}
\def\bbetahat{\widehat{\bbeta}}
\def\bbetahatS{\bbetahat_S}

\def\P{\mathbb{P}}
\def\E{\mathbb{E}}
\def\R{\mathbb{R}}

\def\Ssc{\mathcal{S}}

\def\Dsc{\mathcal{D}}
\def\Jsc{\mathcal{J}}

\def\psitwo{\psi_2}

\def\psitwonorm#1{\left\| #1 \right\|_{\psitwo}}
\def\psitwonormstar#1{\left\| #1 \right\|_{\psitwo}^*}
\def\psione{\psi_1}
\def\psionenorm#1{\left\| #1 \right\|_{\psione}}
\def\psialpha{\psi_{\alpha}}
\def\psialphanorm#1{\left\| #1 \right\|_{\psialpha}}
\def\psibeta{\psi_{\beta}}
\def\psibetanorm#1{\left\| #1 \right\|_{\psibeta}}
\def\psigamma{\psi_{\gamma}}
\def\psigammanorm#1{\left\| #1 \right\|_{\psigamma}}

\def\Ltwonorm#1{\left\| #1 \right\|_2}
\def\Linfnorm#1{\left\| #1 \right\|_{\infty}}

\def\MLtwonorm#1{\left|\left| #1 \right|\right|_2}
\def\MLinfnorm#1{\left|\left| #1 \right|\right|_{\infty}}
\def\Mmaxnorm#1{\left|\left| #1 \right|\right|_{\max}}

\def\bXS{\bX_{S}}
\def\bXSi{\bX_{S,i}}

\def\bmuS{\bmu_{S}}

\def\bXbar{\overline{\bX}}
\def\bXbarS{\bXbar_S}

\def\Ybar{\overline{Y}}
\def\muY{\mu_Y}
\def\bSigmaS{\bSigma_{S}}
\def\bSigmaSY{\bSigma_{S,Y}}

\def\bSigmahat{\widehat{\bSigma}}
\def\bSigmahatS{\bSigmahat_S}

\def\bSigmahatSY{\bSigmahat_{S,Y}}

\def\bSigmatil{\widetilde{\bSigma}}
\def\bSigmatilS{\bSigmatil_S}

\def\bSigmatilSY{\bSigmatil_{S,Y}}

\def\bGamma{\boldsymbol{\Gamma}}

\def\bGammahat{\widehat{\bGamma}}
\def\bGammahatS{\bGammahat_S}
\def\bGammahatSY{\bGammahat_{S,Y}}

\def\bzero{\mathbf{0}}

\def\bXtil{\widetilde{\bX}}
\def\bXtilS{\bXtil_S}
\def\Ytil{\widetilde{Y}}

\def\bpsi{\boldsymbol{\psi}}
\def\bPsi{\boldsymbol{\Psi}}
\def\bpsiS{\bpsi_{S}}
\def\bPsiS{\bPsi_{S}}
\def\bR{\mathbf{R}}
\def\bRnS{\bR_{n,S}}

\def\bTnS{\mathbf{T}_{n,S}}

\def\bZ{\mathbf{Z}}
\def\bZtil{\widetilde{\bZ}}
\def\bZS{\bZ_S}
\def\bZSi{\bZ_{S,i}}
\def\bZtilS{\bZtil_S}

\def\bZbar{\overline{\bZ}}
\def\bZbarS{\bZbar_S}
\def\bnu{\boldsymbol{\nu}}
\def\bnuS{\bnu_S}
\def\bXi{\boldsymbol{\Xi}}
\def\bXiS{\bXi_S}
\def\bXitil{\widetilde{\bXi}}
\def\bXitilS{\bXitil_S}
\def\bUps{\boldsymbol{\Upsilon}}

\def\bUpshatS{\widehat{\bUps}_{S}}

\def\lambdamin{\lambda_{\min}}
\def\lambdamax{\lambda_{\max}}

\def\lambdasupSsc{\lambda_{\sup,\Ssc}}
\def\lambdainfSsc{\lambda_{\inf,\Ssc}}

\def\lambdatil{\widetilde{\lambda}}

\def\lambdatilsupSsc{\lambdatil_{\sup,\Ssc}}
\def\lambdatilinfSsc{\lambdatil_{\inf,\Ssc}}

\def\sigmaY{\sigma_Y}
\def\sigmaXSsc{\sigma_{\bX,\Ssc}}
\def\sigmatil{\widetilde{\sigma}}
\def\sigmatilZSsc{\sigmatil_{\bZ, \Ssc}}

\def\KSsc{K_{\Ssc}}
\def\Ktil{\widetilde{K}}
\def\KtilSsc{\Ktil_{\Ssc}}

\def\Ln{L_{n}}
\def\LnS{L_{n,\Ssc}}
\def\qn{q_{n}}
\def\qnS{q_{n,\Ssc}}
\def\qntil{\widetilde{q}_n}

\def\rn{r_n}
\def\rntil{\widetilde{r}_n}

\def\CSsc{C_{\Ssc}}

\def\bu{\mathbf{u}}
\def\ba{\mathbf{a}}
\def\baS{\ba_S}

\def\AscnS{\mathcal{A}_{n,S}}
\def\BscnS{\mathcal{B}_{n,S}}

\def\ASsc{\mathcal{A}_{\Ssc}}
\def\betahat{\widehat{\beta}}
\def\Rsc{\mathcal{R}}
\def\S{\mathbb{S}}

\def\sigmaxiAS{\sigma_{\xi_{\ASsc}}}
\def\rhoxiAS{\rho_{\xi_{\ASsc}}}

\def\supS{\underset{S \in \Ssc}{\sup}}

\def\convd{\overset{d}{\rightarrow}} 

\numberwithin{equation}{section}

\def\tcr{\textcolor{red}}

\def\tcm{\textcolor{magenta}}

\def\tcr{\textcolor{black}}

\def\tcm{\textcolor{black}}

\def\tcR{\textcolor{red}}

\def\tcR{\textcolor{black}}

\begin{document}


\begin{frontmatter}
\title{Inference for Individual Mediation Effects and Interventional Effects in Sparse High-Dimensional Causal Graphical Models}
\runtitle{Inference in High-Dimensional Causal Graphical Models}
\begin{aug}
  \author[A]{\fnms{Abhishek}  \snm{Chakrabortty$^{\dag}$}\ead[label=e1]{abhishek@stat.tamu.edu}},
  \author[B]{\fnms{Preetam} \snm{Nandy$^{\dag}$}\ead[label=e2]{preetamnandy@gmail.com}}
  \and
  \author[C]{\fnms{Hongzhe} \snm{Li}\corref{Correspondence}\ead[label=e3]{hongzhe@pennmedicine.upenn.edu}\thanksref{t1}}
  \runauthor{A. Chakrabortty, P. Nandy and H. Li}
\thankstext{t1}{ $^{\dag}$These authors contributed equally to this work. Both authors were previously at the University of Pennsylvania during the initial preparation of this work. 
This research was supported in part by the National Institutes of Health grants GM123056 and GM129781.\vspace{0.05in}}

\address[A]{Department of Statistics, Texas A\&M University, \printead{e1}}
\address[B]{LinkedIn Corp., \printead{e2}}
\address[C]{Department of Biostatistics, Epidemiology \& Informatics, University of Pennsylvania, \printead{e3}}
\end{aug}



\begin{abstract}
We consider the problem of identifying intermediate variables (or mediators) that regulate the effect of a treatment on a response variable. While there has been
significant research \tcr{on this classical topic}, little work has been done when the set of potential mediators is high-dimensional.
\tcr{A further complication arises when these 
mediators are interrelated (with unknown dependencies).} In particular, we assume that the causal structure of the treatment, the pre-treatment covariates \tcr{(or confounders)}, the potential mediators and the response is 
a \tcr{(possibly unknown)} directed acyclic graph (DAG). High-dimensional DAG models have previously been used for the estimation of causal effects from observational data. In particular, methods called IDA and joint-IDA have been developed \tcr{for estimating the effects of single interventions and  
multiple simultaneous interventions, respectively.}
In this paper, we propose an IDA-type method, called MIDA, \tcr{for estimating so-called `individual' mediation effects from high-dimensional observational data under our setting.} Although IDA and joint-IDA estimators have been shown to be consistent in certain sparse high-dimensional settings, their asymptotic properties such as convergence in distribution and inferential tools in such settings \tcr{have remained unknown.} In this paper, we prove high-dimensional consistency of MIDA for linear structural equation models with sub-Gaussian errors. More importantly, we derive distributional convergence results for MIDA in similar high-dimensional settings, which are applicable to IDA and joint-IDA estimators as well. \tcr{To the best of our knowledge, these are the first such distributional convergence results facilitating inference for IDA-type estimators.} These results \tcr{are} 
built on our novel theoretical results regarding uniform bounds for linear regression estimators over varying subsets of high-dimensional covariates, which may be of independent interest. Finally, we empirically \tcr{validate} our asymptotic theory \tcr{and demonstrate the usefulness of MIDA} 
in the identification of large mediation effects \tcr{via extensive simulations, and we also} illustrate a practical application of MIDA in genomics with a real dataset.
\end{abstract}

\end{frontmatter}

\vbadness10000 


\section{Introduction}\label{section: introduction}
Although confirmatory causal inference from high-dimensional observational data is impossible due to identifiability issues, this topic has received great attention in the recent past. Intervention experiments are considered to be the gold-standard for making causal inference. However, experimental data cannot always be generated given the considerable ethical concerns, time constraints, and the high costs associated with performing appropriate experiments. Another major problem that can arise in many scientific disciplines is that the sheer number of causal hypotheses is simply too large to test experimentally. Good examples are gene knockout experiments, where potential candidate genes for \tcr{the} knockout experiments typically lie \tcr{in the order of}
thousands. In such a situation, causal predictions from observational data can be extremely useful in prioritizing intervention experiments \citep{MaathuisColomboKalischBuehlmann10, StekhovenEtAl12, LeEtAl17}.

There has been a lot of recent progress in 
estimating causal effects from high-dimensional observational data \tcr{based on a graphical model framework}. Most of these methods assume that the data are generated from an unknown linear structural equation model (LSEM) with independent Gaussian errors, and that the causal relationships among the variables can be represented by a directed cyclic graph (DAG). Under these assumptions, high-dimensional consistency results have been derived for the estimation of the causal graph and causal effects. In particular, \citet{MaathuisKalischBuehlmann09} proposed Interventional calculus when the DAG is Absent (IDA) for estimating the total causal effect of a variable on another variable, and \tcr{they} 
proved a high-dimensional consistency result for \tcr{their IDA estimator. The IDA method has been further
extended to the joint-IDA by \citet{NandyMaathuisRichardson17} for estimating the effects of multiple simultaneous interventions as well,} and a similar high-dimensional consistency result for the joint-IDA estimator has been
 proved therein.

The IDA method estimates a multi-set of causal effects as follows. The first step is to estimate a partially directed graph, called Completed Partially Directed Acyclic Graph (CPDAG), from high-dimensional observational data. This can be done by applying a structure learning algorithm such as the PC algorithm \citep{SpirtesEtAl00, ColomboEtAl14}, greedy equivalence search \tcR{(GES)} \citep{Chickering03} and adaptively restricted greedy equivalence search \tcR{(ARGES)} \citep{NandyHauserMaathuis16}. High-dimensional consistency results for these structure learning algorithms \tcr{have been proved in \citet{KalischBuehlmann07a, ColomboEtAl14,NandyHauserMaathuis16}, e.g. The reason behind estimating a partially directed graph here instead of 
the underlying directed graph is that the true causal DAG is not identifiable from observational data alone without making further stringent assumptions.} A CPDAG
uniquely represents a Markov equivalence class of DAGs that can generate the same joint distribution of the variables. The IDA method estimates a possible causal effect for each DAG in the Markov equivalence class represented by the estimated CPDAG and combines them to produce a multi-set (where each element can have multiple copies) of causal effects. The authors \tcr{also} noted that the listing of all DAGs in the Markov equivalence class from a given CPDAG is typically computationally infeasible for large graphs with thousands of variables, \tcr{and 
provided} computational shortcuts to obtain the multi-set of possible effects without listing all DAGs in the Markov equivalence class of the estimated CPDAG. It is common practice to summarize the multi-set of possible effects by its average or the minimum absolute value.

\tcr{Despite all these recent advances in estimating total causal 
effects and/or learning the underlying causal structure, very little work has been done on the corresponding problem of \emph{causal mediation analysis}
in high-dimensional settings. Such problems, however, are of considerable relevance in the modern `big data' era, with a growing interest across various scientific disciplines in understanding the role of `networks' of multiple intermediate variables (or mediators) in simultaneously regulating the causal effect of a treatment on a response. Inspired by such motivations, we consider here the problem of identifying mediators in settings where the set of potential mediators are: (i) \emph{high dimensional}, and more importantly, (ii) possibly \emph{interrelated with unknown dependencies}. 
Both aspects combined make our setting considerably challenging and distinct compared to most of the existing classical mediation literature. The latter aspect, in particular, creates unique challenges and necessitates revisiting the very \emph{definition} of mediation effects in such settings, compared to existing definitions that typically apply only for conditionally independent mediators.}
\tcr{In this regard, we first propose a novel definition of the mediation effect in the presence of multiple mediators (see Section \ref{subsection: setup}, \tcR{and in particular, Definition \ref{definition: mediation effect} therein}).}

\tcr{Classical mediation analysis has a rich literature. A simpler problem considering only one potential mediator has been well studied within the framework of LSEMs \citep{JuddKenny81, JamesEtAl82, Sobel82, BaronKenny86, MacKinnonEtAl02}. The goal of causal mediation analysis with a single mediator is to understand what portion of the total causal effect of a treatment on a response can be attributed to the potential mediator. In fact, the total effect \tcR{in this case} 
can be decomposed as a sum of the direct effect and the indirect effect, where the indirect effect is the effect of the treatment on the response that goes through the potential mediator. Similarly, in the case of multiple potential mediators, we are interested in understanding what portion of the total effect of the treatment, $X_t$, on the response, $X_p$, can be attributed to a potential mediator, $X_j$. We refer to it as the \emph{individual mediation effect}  \tcR{(Defn. \ref{definition: mediation effect})} with respect to $X_j$. Note
that \tcR{with possibly interrelated mediators,}
the total effect of the treatment on the response may \emph{not} be decomposed as the sum of all individual mediation effects and the direct effect here, unless the potential mediators are conditionally independent of each other given the treatment.} 


\tcr
{The estimation and testing for mediation effects in causal models with conditionally independent mediators have also been considered in both 
classical settings \citep{PreacherHayes08, BocaEtAl14}, as well as in high-dimensional settings \citep{ZhangEtAl16}. For causal models with conditionally dependent mediators, \citet{VanderWeeleVansteelandt14} discussed estimation methods for the total effect of all mediators (or the total indirect effect), while 
\citet{HuangPan16} proposed to estimate the individual effects with respect to a transformed set of conditionally independent variables in high-dimensional settings. In contrast to these existing works, we are interested in separately evaluating the importance of each potential mediator, \tcR{allowing for (unknown) inter-dependencies}. The identification of mediators corresponding to large individual mediation effects can be very useful in a variety of scientific applications, including genomics, where it is often of interest to understand how an influential genotype regulates a phenotype of interest through gene expressions.
}

\paragraph*{Our contributions\tcR{.}}\label{sec:contributions}
In this paper, we propose an IDA-type \tcR{estimation} method, called MIDA \tcR{(see Section \ref{subsection:MIDA-estimator})}, for estimating the causal mediation effect of a treatment variable on a response variable through an intermediate variable (a.k.a.\ mediator) in high-dimensional settings. In particular, we consider a treatment (a.k.a.\ exposer) \tcr{$X_t$, a set of pre-treatment covariates (a.k.a.\ confounders) $\{X_1,\ldots, X_{t-1}\}$,} a response variable $X_p$, and a set of potential mediators \tcr{$\{X_{t+1},\ldots, X_{p-1}\}$} \tcR{that could be high-dimensional}. \tcr{A pre-treatment covariate can be a common cause of the treatment variable and the response, and also a common cause of the potential mediators and the response.} We assume that the causal relationships among the variables in $\mathbf{X} = \{X_1, X_2, \hdots$ $, X_p\}$ can be represented by a DAG \tcR{(possibly unknown)}, where $X_i$ and $X_j$ are connected by a directed edge if and only if $X_i$ is a direct cause of $X_j$.

 As is the case with IDA\tcR{-type estimators}, MIDA relies on the estimation of an underlying CPDAG, and it produces a multi-set of possible mediation effects, which we summarize by taking the average. We prove the \emph{consistency} of MIDA for certain sparse high-dimensional LSEMs with sub-Gaussian errors \tcR{(Theorem \ref{theorem: high-dimensional consistency})}.
 Furthermore, we provide \emph{unified distributional convergence results} for IDA-type estimators in similar high-dimensional settings \tcR{(Theorems \ref{theorem: asymptotic linearity}--\ref{theorem: asymptotic distribution}, Corollaries \ref{corollary: asymptotic normality of IDA}--\ref{corollary: asymptotic distribution of MIDA})}, thus facilitating \emph{inference} for such estimators. These results have been built on a novel \emph{uniform non-asymptotic theory for linear regression over varying subsets} of high-dimensional covariates \tcR{(Theorem \ref{linreg:main_thm1})} which may be of independent interest. \tcr{This is a critical tool in our case for handling the possibly large multi-set of causal effects obtained from the estimated CPDAG which poses the key challenge in our theory for inference. 
 The theory notably also does \emph{not} depend on the nature of the CPDAG estimation procedure as long as it is consistent.} To the best of our knowledge, we propose the first estimation method for mediation effects when the data are generated from an unknown DAG, as well as the first high-dimensional distributional convergence results and inferential tools for IDA-type estimators of \emph{both} interventional effects as well as mediation effects. \tcr{Our contributions in the latter regard thus extend beyond just inference for mediation effects.} 

 \tcr{
 Finally, we note that while we work with the CPDAG here, the underlying causal DAG is identifiable in the following special cases: (i) when \emph{all} error variables in the LSEM are non-Gaussian \citep{ShimizuEtAl06-JMLR, ShimizuEtAl11} and (ii) when \emph{all} error variables in the LSEM are Gaussian with equal error variances \citep{PetersBuhlmann14, ShiLi20}. In these cases, MIDA can \emph{still} be applied with the estimated DAG (instead of the estimated CPDAG) to obtain mediation effects (instead of a multi-set of possible mediation effects). However, these additional assumptions cannot be verified from observational data typically, and hence, a more conservative approach of estimating the CPDAG is recommended. Note also that the equal error variance assumption \emph{cannot} be achieved by normalizing the data to have equal variances for all observed variables since the assumption is on the underlying data generating error variables instead of the observed variables (see Definition \ref{definition: LSEM}).
 }

\par\smallskip
\noindent\tcR{\emph{Organization.}}
The rest of this paper is organized as follows. Section \ref{section: preliminaries} provides some necessary background material. In Section \ref{section: estimation}, we propose the MIDA algorithm for estimating individual mediation effects from observational data. \tcr{In Section \ref{section: high-dimensional consistency}, we prove consistency of MIDA in 
sparse high-dimensional LSEMs with sub-Gaussian errors and also discuss the modifications required to relax the linear sub-Gaussian assumption}. Our non-asymptotic theoretical results on linear regression over varying subsets of high-dimensional covariates are given in Section \ref{section: linear regression theory} which can be read independently. \tcr{Section \ref{section: asymptotic distribution} discusses the distributional convergence results and inferential tools for MIDA and IDA-type estimators.} 
Section \ref{section: simulations} contains simulation results, where we demonstrate the usefulness of MIDA and our asymptotic theory for the identification of \tcr{non-zero} mediation effects. In Section \ref{section: application}, we apply MIDA to a real dataset generated from a collection of yeast segregants, \tcr{and we end with a concluding discussion in Section \ref{section: discussion}.}
All proofs, additional technical materials, and additional numerical results are collected in the \tcr{\hyperref[supp_mat]{Supplement} 
(\tcR{Appendices} \ref{section: proofs}\tcR{--}\ref{sec:FDR_Control}).} 

\section{Preliminaries}\label{section: preliminaries}
We begin with a few basic definitions and notations.

\subsection{Graph Terminology}

We consider graphs $\mathcal H = (\mathbf{X},E)$ with vertex (or node) set $\mathbf{X} = \{X_1,\ldots,X_p\}$ and edge set $E$. There is at most one edge between any pair of vertices and edges may be either directed ($X_i\to X_j$) or undirected ($X_i - X_j$). If $\mathcal H$ contains only (un)directed edges, it is called \emph{(un)directed}. If $\mathcal H$ contains directed and/or undirected edges, it is called \emph{partially directed}.  A pair of nodes $\{X_i,X_j\}$ are \emph{adjacent} if there is an edge between $X_i$ and $X_j$. If $X_i \to X_j$, then $X_i$ is a \emph{parent} of $X_j$. We denote the set of all parents of $X_j$ in $\mathcal{H}$ by $\mathbf{Pa}_{\mathcal{H}}(X_j)$, \tcr{and all adjacent nodes of $X_j$ in $\mathcal{H}$ by $\mathbf{Adj}_{\mathcal{H}}(X_{j})$.} A \emph{path} between $X_i$ and $X_j$ is a sequence of distinct nodes $\{X_i,\ldots,X_j\}$ such that all successive pairs of nodes are adjacent. A \emph{directed path} from $X_i$ to $X_j$ is a path between $X_i$ and $X_j$ where all edges are directed towards $X_j$. A directed path from $X_i$ to $X_j$ together with the edge $X_j \to X_i$ forms a \emph{directed cycle}. A (partially) directed graph that does not contain a directed cycle is called a \emph{(partially) directed acyclic graph} or (P)DAG. 

\subsection{Linear Structural Equation Models (LSEMs)}\label{subsection: LSEM}

\begin{definition}\label{definition: LSEM}
   Let $\mathcal{G}_0 = (\mathbf{X},E)$ be a DAG and let $B_{\mathcal{G}_0}$ be a $p \times p$ matrix such that $(B_{\mathcal{G}_0})_{ij} \neq 0$ if and only if $X_i \in \mathbf{Pa}_{\mathcal{G}_0}(X_j)$. Let $\boldsymbol{\epsilon} = (\epsilon_1,\ldots,\epsilon_p)^{T}$ be a zero mean random vector of jointly independent error variables. Then $\mathbf{X} = (X_1,\dots,X_p)^T$ is said to be generated from a linear structural equation model (LSEM) characterized by the pair $ (B_{\mathcal{G}_0},\boldsymbol{\epsilon})$ if
   \begin{align}\label{eq: LSEM 1}
      (\mathbf{X} - \boldsymbol{\mu} ) \; \leftarrow \;  B_{\mathcal{G}_0}^T(\mathbf{X} - \boldsymbol{\mu}) + \boldsymbol{\epsilon}, \quad \mbox{where}  \;\; \boldsymbol{\mu} \; := \; \E(\mathbf{X}).
   \end{align}
\end{definition}
If $\mathbf{X}$ is generated from an LSEM characterized by the pair $ (B_{\mathcal{G}_0,}\boldsymbol{\mathbf{\epsilon}})$, then
we call $\mathcal{G}_0$ the \emph{causal DAG}. The symbol ``$\leftarrow$" in (\ref{eq: LSEM 1}) emphasizes that the expression should be understood as a generating mechanism rather than as a mere equation. We emphasize that we assume here \tcr{that there are no hidden confounders (see Section \ref{sec:assumption_discussion} for more discussion on this case)}, and hence the joint independence of the error terms. 
In the rest of the paper, we refer to LSEMs without explicitly mentioning the independent error assumption.

\subsection{Markov Equivalence Class of DAGs}\label{subsection: Markov equivalence class of DAGs}

The causal DAG $\mathcal{G}_0$ is \tcR{(typically)} not identifiable from (observational data from) the distribution of $\mathbf{X}$. A DAG encodes conditional independence relationships via the notion of \textit{d-separation} (\cite{Pearl00}, Theorem 1.2.4, page 18). In general, several DAGs can encode the same conditional independence relationships, and such DAGs form a \emph{Markov equivalence class}. Two DAGs belong to the same Markov equivalence class if and only if they have the same skeleton and the same v-structures \citep{VermaPearl90}. A Markov equivalence class of DAGs can be uniquely represented by a \emph{completed partially directed acyclic graph} (CPDAG) \citep{SpirtesEtAl00, Chickering02}, which is a graph that can contain both directed and undirected edges. A CPDAG satisfies the following: $X_i\to X_j$ in the CPDAG if $X_i\to X_j$ in every DAG in the Markov equivalence class, and $X_i - X_j$ in the CPDAG if the Markov equivalence class contains a DAG for which $X_i\to X_j$ as well as a DAG for which $X_i\leftarrow X_j$. CPDAGs can be estimated from observational data using various algorithms \citep{SpirtesEtAl00, Chickering03, TsamardinosEtAl06, NandyHauserMaathuis16}.

\subsection{Problem Setup}
\label{subsection: setup}
\tcr{
We assume that $\mathbf{X} = \{X_1, \ldots, X_p\}$ is generated from an LSEM characterized by the pair $(B_{\mathcal{G}_0,}\boldsymbol{\mathbf{\epsilon}})$ as in \eqref{eq: LSEM 1}, where $\{X_1,\ldots, X_{t-1}\}$ is a set of pre-treatment covariates (a.k.a. \emph{confounders}), $X_t$ is the \emph{treatment} variable, $\{X_{t+1},\ldots, X_{p-1}\}$ is a set of \emph{potential mediators} \tcm{and $X_p$ denotes the \emph{response} variable. Note that we do allow the case $t=1$ here to represent the absence of 
pre-treatment covariates.} We assume that no potential mediator is a direct cause of a variable in $\{X_1, \ldots, X_t\}$, i.e. 
$(B_{\mathcal{G}_0})_{ji} = 0$ for all $i \leq t$, \tcm{for each $j = t+1,\ldots, p-1$}. Further, we assume that the response variable $X_p$ is not a direct cause of any other variable in $\mathbf{X}$, i.e. 
$(B_{\mathcal{G}_0})_{pj} = 0$ for all $j < p$. \tcm{Finally, the \emph{observed data} consists of $n$ independent and identically distributed (i.i.d.) realizations of $\mathbf{X}$, where throughout we allow for a \emph{high-dimensional} setting with $p$ allowed to diverge with the sample size $n$.}
%
}

\par\smallskip
In order to define the total causal effect of a variable $X_i$ on another variable $X_k$, we consider a hypothetical outside \emph{intervention} to the system where we set a variable $X_i$ to some value $x_i$ uniformly over the entire population. This can be denoted by Pearl's do-operator: $do(X_i = x_i)$ \citep{Pearl09}, which corresponds to removing the edges into $X_i$ in $\mathcal{G}_0$ (or equivalently, setting the $i$-th column of $B_{\mathcal{G}_0}$ equal to zero) and replacing $\epsilon_i$ by the constant $x_i$. The post-interventional expectation of $X_k$ is denoted by $\Exp[X_k \mid do(X_i = x_i)]$.

Under the LSEM 
assumption, $\Exp[X_k \hspace{0.05in} \vert \hspace{0.05in} do(X_i = x_i)]$ is a linear function of $x_i$ and the \emph{total causal effect} of $X_i$ on $X_k$ is defined as \citep{MaathuisKalischBuehlmann09}
\[
\theta_{ik} \; := \; \frac{\partial}{\partial x_i} \Exp[X_k \mid do(X_i = x_i)].
\]
\tcr{To provide} a graphical interpretation of $\theta_{ik}$, we define the effect of $X_{i_0}$ to $X_{i_{k+1}}$ through a directed path $\{X_{i_0},X_{i_1},\ldots,X_{i_k},X_{i_{k+1}}\}$ \tcr{as} 
 $\prod_{r=0}^{k}(B_{\mathcal{G}_0})_{i_{r}i_{r+1}}$. Then the total causal effect $\theta_{ik}$ is given by the sum of the effects of $X_i$ to $X_k$ through all directed paths from $X_i$ to $X_k$. This is known as the \emph{path method} for computing the total causal effects in an LSEM \citep{Wright21}.

We denote a \emph{joint-intervention} on $X_i$ and $X_j$ by $do(X_i = x_i,~ X_j = x_j)$. Again, the post-interventional expectation $\Exp[X_k \mid do(X_i = x_i,~ X_j = x_j)]$ is a linear function of $(x_i,x_j)$ and the effect of $X_i$ on $X_k$ in the joint intervention $do(X_i = x_i,~ X_j = x_j)$ is defined as \citep{NandyMaathuisRichardson17}
\[
\theta_{ik}^{(i,j)} \; := \; \frac{\partial}{\partial x_i} \Exp[X_k \mid do(X_i = x_i,~ X_j = x_j)].
\]
Note that $\theta_{ik}^{(i,j)}$ can be interpreted as the total causal effect of $X_i$ on $X_k$ when we set $X_j = x_j$ uniformly over the entire population, that is, the portion of the total effect of $X_i$ on $X_k$ that does not go though $X_j$.

Finally, we define the \emph{individual mediation effect} of a potential mediator $X_j$ \tcr{($j = t+1, \ldots, p-1$)} to be the portion of total effect of the treatment variable $X_t$ on the response $X_p$ that goes through $X_j$.
\begin{definition}\label{definition: mediation effect}
The \tcr{individual mediation effect, $\eta_j$,} with respect to a potential mediator $X_j$ \tcr{($j = t+1, \ldots, p-1$) is defined as}
\[
\eta_j \; := \; \frac{\partial}{\partial x_t} \Exp[X_p \mid do(X_t = x_t)] - \frac{\partial}{\partial x_t} \Exp[X_p \mid do(X_t = x_t,~ X_j = x_j)].
\]
\end{definition}

{\color{black}
\vspace{-0.15in}
\begin{remark}\label{remark: general SEM}
\emph{
\tcr{
The individual mediation effect can be interpreted as the change in the total causal effect of the treatment $X_t$ on the response $X_p$ when the potential mediator $X_j$ is knocked out from the causal graph $\mathcal{G}_0$ by the intervention $do(X_j = x_j)$. Note that Definition \ref{definition: mediation effect} as well as this interpretation of the individual mediation effect holds for a general structural equation model: $X_i \leftarrow f_i(\mathbf{X}_{\mathbf{Pa}_{\mathcal{G}_0}}(X_i), \epsilon_i)$ for $i = 1,\ldots,p$. While under the linearity assumption $\eta_j$ does not depend on the intervention values $x_t$ and $x_j$, in a more general setting $\eta_j(x_t, x_j)$ \emph{can} be a non-trivial function of \tcm{$(x_t,x_j)$}. 
}
}
\end{remark}
}
Under the linearity assumption, the individual mediation effect $\eta_j$ is given by the sum of the effects of $X_t$ to $X_p$ through all directed paths from $X_t$ to $X_p$ that go through $X_j$. It follows from Theorem 3.1 of \citet{NandyMaathuisRichardson17} that $\eta_j$ equals the product of the total causal effect of $X_t$ on $X_j$ and the total causal effect of $X_j$ on $X_p$. \tcr{We formalize this in the proposition below.}
\begin{proposition}\label{proposition: product formula}
Let $\mathbf{X}$ be generated from an LSEM. The individual mediation effect \tcr{$\eta_j$} with respect to a potential mediator $X_j$ is \tcr{then} given by $$\eta_j \; = \; \theta_{tj}\theta_{jp} \quad \tcr{(j = t+1, \ldots, p-1)},$$ where \tcr{for any $(i,k)$,} $\theta_{ik}$ denotes the total causal effect of $X_i$ on $X_k$.
\end{proposition}
\tcr{It is important to note that this `product-type' representation of the mediation effect $\eta_j$ does \emph{not} correspond, in general, to the product of regression coefficients obtained from 
regressing the mediator vs. \tcR{the} treatment, and the response vs. \tcR{the} mediator, as in `marginal' mediation analyses under conditionally independent mediators. We illustrate this further in Examples \ref{example: individual mediation effect}\tcR{--}\ref{example: covariate adjustment} below. \tcR{(Example \ref{example: covariate adjustment} is in Section \ref{subsection: covariate adjustment} and is a continuation of Example \ref{example: individual mediation effect}.)}}

{\color{black}
\vspace{-0.05in}
\begin{example}\label{example: individual mediation effect}
\emph{
\tcr{
\tcm{We consider a simple case with $p = 7$, $t = 2$.}
The DAG  in Figure \ref{fig: example} represents the causal structure among the pre-treatment covariate $X_1$, the treatment variable $X_2$, the potential mediators $\{X_3,\ldots, X_6\}$ and the response variable $X_7$. The edge weights represent the coefficients of the following LSEM: $X_1 \leftarrow \epsilon_1$, $X_2 \leftarrow1.6X_1 + \epsilon_2$, $X_3 \leftarrow 0.7X_2 + 1.4X_4 +  \epsilon_3$, $X_4 \leftarrow 1.4X_1 + \epsilon_4$, $X_5 \leftarrow 1.2X_2 + 0.9X_3 + \epsilon_5$, $X_6 \leftarrow 1.1X_5 + \epsilon_6$, and $X_7 \leftarrow 0.6X_3 + 0.8X_4 + 1.8X_6 + \epsilon_7$. The error variables can be assumed to have any distribution for deriving the total causal effects and the individual mediation effects, since the total causal effects in an LSEM do not depend on the distributions of the error variable. The total causal effects of the treatment variable on the potential mediators, the total causal effects of the potential mediators on the response variable, and the individual mediation effects are given in Table \ref{table: example}. Note that the equality of $\eta_5$ and $\eta_6$ represents the fact that the change in the total causal effect of $X_2$ on $X_7$ for knocking out $X_5$ from the causal graph is the same as the change for knocking out $X_6$.
}
}
\end{example}

\begin{table}[!ht]
\vspace{-0.1in}
\centering
\footnotesize
\caption{Individual mediation effects\tcr{: Illustration for the DAG in Example \ref{example: individual mediation effect}.}}
\label{table: example}
\begin{tabular}{c|c|c|c}
j & $\theta_{2j}$  & $\theta_{j7}$ & $\eta_j$  \\ \hline
$3$ & $0.7$ & $0.6 + 0.9 \times 1.1 \times 1.8$ & $1.6674$ \\ \hline
$4$ & $0$ & $0.8 + 1.3 \times 0.6 + 1.3 \times 0.9 \times 1.1 \times 1.8$ & $0$ \\ \hline
$5$ & $0.7 \times 0.9 + 1.2$ & $1.1 \times 1.8$ & $3.6234$ \\ \hline
$6$ & $0.7 \times 0.9  \times 1.1+ 1.2 \times 1.1$ & $1.8$ & $3.6234$
\end{tabular}
\vspace{-0.1in}
\end{table}
}

\begin{figure}[!t]
  \centering
  \footnotesize
       \begin{tikzpicture}[scale=0.8, transform shape]
        \node        (X2) at (0, 0)                      {$X_2$};
       \node        (X1) at (1.5, 1.5)        {$X_1$};
       \node       (X3) at (1.5, 0)             {$X_3$};
        \node         (X5) at (1.5, -1.5)            {$X_5$};
       \node       (X4)    at (3, 0.75)        {$X_4$};
     \node    (X6) at (3, -0.75)        {$X_6$};
     \node    (X7) at (4.5, 0)            {$X_7$};
        \draw[->, thick] (X1) edge  node[left=2pt, above=0.5pt]{$1.6$}  (X2);
          \draw[->, thick] (X2) edge  node[above=1pt]{$0.7$} (X3);
          \draw[->, thick] (X2) edge   node[left=1pt]{$1.2$} (X5);
         \draw[->, thick] (X3) edge  node[right=1pt]{$0.9$}  (X5);
      \draw[->, thick] (X1) edge  node[right=1pt, above=1pt]{$1.4$}  (X4);
           \draw[->, thick] (X4) edge  node[above=1pt]{$1.3$} (X3);
           \draw[->, thick] (X4) edge  node[above=1pt]{$0.8$} (X7);
           \draw[->, thick] (X3) edge  node[above=1pt]{$0.6$} (X7);
           \draw[->, thick] (X5) edge  node[below=1pt]{$1.1$} (X6);
           \draw[->, thick] (X6) edge  node[below=1pt]{$1.8$} (X7);
   \end{tikzpicture}
     \caption{\tcr{Example of a} weighted DAG representing the data generating process.}
\label{fig: example}
\vspace{-0.15in}
  \end{figure}


\subsection{Notations \tcr{and the `Faithfulness' Assumption}}\label{subsection: notation}

We denote the vector of potential mediators $(X_{t+1},\ldots,$ $X_{p-1})^T$ by $\mathbf{X}'$ and the corresponding subgraph of $\mathcal{G}_{0}$ by $\mathcal{G}_{0}'$ (obtained by deleting the nodes $\mathbf{X}_{\{1,\ldots, t\}}\cup X_p$ and the corresponding edges from $\mathcal{G}_0$). \tcr{Let $B_{\mathcal{G}_0'}$ be the sub-matrix of $B_{\mathcal{G}_0}$ that corresponds to $\mathbf{X}'$ and let $\boldsymbol{\epsilon}' = (\epsilon_{t+1},\ldots,\epsilon_{p-1})^T$.} Further, we denote the CPDAG representing the Markov equivalence class of $\mathcal{G}_{0}'$ by $\mathcal{C}_{0}'$ and the Markov equivalence class by $\MEC(\mathcal{C}_0')$.  We assume that the conditional distribution of $\mathbf{X}'$ given $\{X_1, \ldots, X_t\}$ is \emph{faithful} to $\mathcal{G}_0$. The faithfulness condition states that every independence constraint that holds in the distribution is encoded by $\mathcal{G}_0'$ (see, \tcr{e.g.,} 
Definition 3.8 of \cite{KollerFriedman09}). This assumption is a necessary condition for learning causal structures from observational data \citep{SpirtesEtAl00, Chickering02}, and we do \emph{not} need the faithfulness assumption when the underlying causal structure is known or given.

We will often treat sets as vectors and vice versa, where we consider an arbitrary ordering of the elements in a vector unless specified otherwise. For example, $(i,S,k)$ denotes a vector where the first element is $i$, the last element is $k$, but elements of the set $S$ are ordered arbitrarily in $(i,S,k)$. We denote the covariance matrix of $\mathbf{X}$ by $\Sigma_0$. For any set $S \subseteq \{1,\ldots,p\}$, we denote the corresponding random vector $\{X_r : r \in S\}$, \tcr{i.e. the restriction of $\mathbf{X}$ onto $S$,}  by $\mathbf{X}_S$. Further, we denote $\Cov(\mathbf{X}_{S_1},\mathbf{X}_{S_2})$ by $(\Sigma_0)_{S_1S_2}$. For simplicity, we denote $(\Sigma_0)_{\{i\}\{j\}}$ and $(\Sigma_0)_{\{i\}S}$ by $(\Sigma_0)_{ij}$ and $(\Sigma_0)_{iS}$ respectively.

We denote the $i$-th column of the $k \times k$ identity matrix by $e_{i,k}$. For $i \neq k$ and any set $S \subseteq \{1,\ldots,p\}\setminus \{i,k\}$, we denote the coefficient of $X_i$ in the linear regression of $X_k$ on $\mathbf{X}_{\{i\} \cup S}$ by $\beta_{ik | S}$ or by $\beta_{ik | \mathbf{X}_S}$. For simplicity, we denote $\beta_{ik |\emptyset} $ by $\beta_{ik}$. 
Note that $\beta_{ik | S} $ is \emph{model-free,} i.e. it is well-defined \emph{regardless} of whether or not the conditional expectation $\Exp[X_k\mid \mathbf{X}_{\{i\} \cup S}]$ is a linear function of $\{X_i\} \cup \mathbf{X}_S$, and in general, $\beta_{ik | S} = e_{1,|S|+1}^{T} ((\Sigma_0)_{(i,S)(i,S)})^{-1}(\Sigma_0)_{(i,S)k}$.


\tcr{For any vector $\bv = (v_j)_{j=1}^d\in \R^d$, for any $d \geq 1$, 
$\| \bv \|_r := (\sum_{j=1}^d |v_j|^r)^{1/r}$, for any $r \geq 1$, 
and $\|\bv\|_{\infty} := \max \{|v_j| : j = 1, \hdots, d\}$ denote the $L_r$ and $L_{\infty}$ norms of $\bv$, respectively.
For any matrix $\bM = [M_{ij}]_{i,j=1}^d \in \R^{d \times d}$, 
$\MLtwonorm{\bM} := {\sup}_{\Ltwonorm{\bv} \leq 1} \Ltwonorm{\bM v}$, $\MLinfnorm{\bM} := \underset{1 \leq i \leq d}{\max} \sum_{j=1}^d |M_{ij}|$ and $\Mmaxnorm{\bM} := \underset{1 \leq i,j \leq d}{\max} |M_{ij}|$ respectively denote the spectral norm, the matrix-$L_{\infty}$ norm and the maximum norm of $\bM$. Further, we denote any symmetric positive definite (p.d.) matrix $\bM \in \R^{d \times d}$ as $\bM \succ 0$, and its minimum and maximum eigenvalues (or singular values) as $\lambdamin(\bM) > 0$ and $\lambdamax(\bM) \equiv \MLtwonorm{\bM} > 0$, respectively.}


%


\section{Estimating Individual Mediation Effects}\label{section: estimation}
\tcr{We \tcR{begin with} 
a useful representation of the individual mediation effect\tcR{(s)} via regression coefficients \tcR{(arising from different regressions involving the mediator and/or its parent set).}}

\subsection{Individual Mediation Effects via Covariate Adjustments}\label{subsection: covariate adjustment}
{\color{black}
\tcr{Consider $X_i$ and $X_k$ such that $X_k$ is not a direct cause of $X_i$, that is, $X_k$ is not a parent of $X_i$ in the causal graph. Then the linearity assumption implies that the total causal effect of $X_i$ on $X_k$ can be expressed as the coefficient of $X_i$ in the linear regression of $X_k$ on $X_i \cup \mathbf{X}_S$ for an appropriately chosen set of covariates $\mathbf{X}_S \subseteq \mathbf{X} \setminus \{X_i, X_k\}$ \citep{MaathuisKalischBuehlmann09, NandyMaathuisRichardson17}. An example of such a covariate set is the \tcR{\emph{parent set}} $\mathbf{Pa}_{\mathcal{G}_0}(X_i)$. More generally, any set of variables satisfying Pearl's back-door criterion (see Definition 3.3.1 of \cite{Pearl00}) can be chosen as $\mathbf{X}_S$. Using this we can write \tcm{the individual mediation effects $\eta_j$'s} as a product of regression coefficients as follows.}

\begin{lemma}\label{lemma: covariate adjustment}
\tcr{Let $\mathbf{X}$ be as in Section \ref{subsection: setup}. 
Then, $\forall \; j = t+1, \ldots, p-1$,} 
\begin{align}
\label{eq: product formula}
& \tcr{\eta_j \; = \; \theta_{tj}\theta_{jp} \; = \; \beta_{tj | \mathbf{X}_{\{1,\ldots, t-1\}}}~  \beta_{jp | \mathbf{Pa}_{\mathcal{G}_0'}(X_j) \cup \mathbf{X}_{\{1,\ldots,t\}}}.} 
\end{align}
\end{lemma}
}

{\color{black}
\vspace{-0.1in}
\begin{example}\label{example: covariate adjustment}
\emph{
\tcr{
\tcm{To illustrate Lemma \ref{lemma: covariate adjustment},} we reconsider Example \ref{example: individual mediation effect} and let $\Sigma = \Cov(X)$. Using \eqref{eq: product formula}, the individual mediation effect $\eta_{\tcR{j}}$'s can be computed as in Table \ref{table: covariate adjustment example}. Furthermore, it is easy to verify that a \emph{naive method} that ignores the causal graph \tcm{among the mediators} 
 and computes the individual mediation effects as $\beta_{2j | \{1\}} \times \beta_{j7 | \{1,2\}}$ would be \emph{inaccurate} for $j = 3, 5, 6$.
}
}
\end{example}

\begin{table}[!ht]
\vspace{-0.1in}
\centering
\footnotesize
\caption{Individual mediation effects  via covariate adjustments for the DAG in \tcR{Example \ref{example: individual mediation effect}.}}
\label{table: covariate adjustment example}
\begin{tabular}{c|l}
j & $\eta_j = \theta_{2j}  \times \theta_{j7}$  \\ \hline
$3$ & $\beta_{23 | \{1\}} \times \beta_{37|\{1,2,4\}} = \mathbf{e}_{2,2}^T(\Sigma_{(1,2)(1,2)})^{-1} \Sigma_{(1,2)3} \times \mathbf{e}_{3,4}^T(\Sigma_{(1,2,3,4)(1,2,3,4)})^{-1} \Sigma_{(1,2,3,4)7}$  \\ \hline
$4$ & $\beta_{24 | \{1\}} \times \beta_{47|\{1,2\}} = \mathbf{e}_{2,2}^T(\Sigma_{(1,2)(1,2)})^{-1} \Sigma_{(1,2)4} \times \mathbf{e}_{3,3}^T(\Sigma_{(1,2,4)(1,2,4)})^{-1} \Sigma_{(1,2,4)7}$ \\ \hline
$5$ & $\beta_{25 | \{1\}} \times \beta_{57|\{1,2,3\}} = \mathbf{e}_{2,2}^T(\Sigma_{(1,2)(1,2)})^{-1} \Sigma_{(1,2)5} \times \mathbf{e}_{4,4}^T(\Sigma_{(1,2,3,5)(1,2,3,5)})^{-1} \Sigma_{(1,2,3,5)7}$ \\ \hline
$6$ & $\beta_{26 | \{1\}} \times \beta_{67|\{1,2,3\}} = \mathbf{e}_{2,2}^T(\Sigma_{(1,2)(1,2)})^{-1} \Sigma_{(1,2)6} \times \mathbf{e}_{4,4}^T(\Sigma_{(1,2,5,6)(1,2,5,6)})^{-1} \Sigma_{(1,2,5,6)7}$
\end{tabular}
\vspace{-0.1in}
\end{table}
}

\subsection{The MIDA Estimator}\label{subsection:MIDA-estimator}
{\color{black}
Our goal is to estimate $\eta_j=\theta_{tj}\theta_{jp}$ based on i.i.d. data from the distribution of $\mathbf{X}$, for $j=t+1,\ldots,p-1$. Note that if $\mathbf{Pa}_{\mathcal{G}_0'}(X_j)$ were \emph{known}, then we could estimate $\eta_j$ by plugging in the sample regression coefficients $\hat{\beta}_{tj | \mathbf{X}_{\{1,\ldots, t-1\}}}$ and $\hat{\beta}_{jp | \mathbf{Pa}_{\mathcal{G}_0'}(X_j)\cup \mathbf{X}_{\{1,\ldots,t\}}}$ in \eqref{eq: product formula}.

\tcr{\tcR{When} $\mathcal{G}_0'$ is \emph{unknown}, we need to estimate it from the data. However, a causal DAG is \tcR{(usually)} not identifiable \tcm{from} observational data \tcm{without further assumptions.} But we can estimate the CPDAG representing the corresponding Markov equivalence class (see Section \ref{subsection: Markov equivalence class of DAGs}). In particular, we \tcm{can} estimate the CPDAG $\mathcal{C}_0'$ that represents the Markov equivalence class of $\mathcal{G}_0'$. Consequently, $\theta_{jp}$ and $\eta_j$ are \tcR{also} not identifiable from observational data \tcm{(with unknown $\mathcal{G}_0'$).} Therefore, following the IDA approach of \cite{MaathuisKalischBuehlmann09} \tcm{and} \cite{NandyMaathuisRichardson17}, we \tcm{aim to} estimate the \tcm{following \emph{identifiable version} 
$\eta_j(\mathcal{C}_0')$ of $\eta_j$} defined as}
\begin{align*}
& \tcr{\eta_j(\mathcal{C}_0') \hspace{0.01in}\;\;\; := \; \beta_{tj | \mathbf{X}_{\{1,\ldots, t-1\}}}~\times~\mathrm{aver}(\Theta_{jp}(\mathcal{C}_0')), \quad \mbox{where}}\\
& \tcr{\Theta_{jp}(\mathcal{C}_0') \; := \; \{\beta_{jp | \mathbf{Pa}_{\mathcal{G}}(X_j) \cup \mathbf{X}_{\{1,\ldots,t\}}} : \mathcal{G} \in \MEC(\mathcal{C}_0') \} }
\end{align*}
\tcr{ is a \emph{multi-set} of possible causal effects of $X_j$ on $X_p$, and $\mathrm{aver}(A)$ \tcm{denotes} the \emph{average} of all numbers (respecting \tcR{any} multiple occurrences) in the multi-set $A$. We will empirically verify that $\eta_j(\mathcal{C}_0')$ serves as a reasonable proxy for $\eta_j$ in sparse high-dimensional settings (see Section \ref{section: simulations} \tcm{for more details}).}
}

In order to estimate \tcr{the CPDAG} $\mathcal{C}_0'$, we first remove the effect of $\mathbf{X}_{\{1,\ldots,t\}}$ on each potential mediator $X_j$ by replacing the data that corresponds to $X_j$ by the \emph{residuals} of the regression of $X_j$ on $\mathbf{X}_{\{1,\ldots,t\}}$. Then, we apply \tcR{any suitable} 
structure learning algorithm on th\tcR{is} transformed data for estimating $\mathcal{C}_0'$, \tcr{followed by estimating $\eta_j(\mathcal{C}_0')$.
The steps of our MIDA approach 
are formalized next in Algorithm \ref{algorithm: CPDAG estimation}. For any $i, k$ and $\mathbf{X}_S \subseteq \mathbf{X}$, let $\hat{\beta}_{ki | \mathbf{X}_{S}}$ denote the estimated regression coefficient of $X_i$ in the linear regression of $X_k$ on $\mathbf{X}_{\{i\}\cup S}$.} 
\begin{algorithm}[!h]
   \caption{MIDA}
   \label{algorithm: CPDAG estimation}
   \begin{algorithmic}[1]
   \REQUIRE $n$ i.i.d.\ observations of $\mathbf{X}$ (data)
   \ENSURE Estimates of $\eta_j(\mathcal{C}_0')$ for $j = t+1,\ldots,p-1$
     \STATE for $j=t+1,\ldots,p-1$, obtain the vector of residuals $\mathbf{r}_{j} = (r_{j}^{(1)},\ldots,r_{j}^{(n)})$ from the regression of $X_j$ on $\mathbf{X}_{\{1,\ldots,t\}}$;
   \STATE apply \tcr{\tcR{\emph{any}} suitable} structure learning algorithm (such as (AR)GES or PC) on the data $\{\mathbf{r}_{t+1},\ldots,\mathbf{r}_{p-1}\}$ to obtain an estimate $\hat{\mathcal{C}_0'}$ of the CPDAG $\mathcal{C}_0'$;
    \STATE for each $j = t+1,\ldots,p-1$, obtain a multi-set of possible causal effects $\hat{\Theta}_{jp}(\hat{\mathcal{C}_0'}) := \{\hat{\beta}_{jp | \mathbf{Pa}_{\mathcal{G}}(X_j)\cup \mathbf{X}_{\{1,\ldots,t\}}} : \mathcal{G} \in \MEC(\hat{\mathcal{C}}_0')\}$ based on the original data;\label{line: estimator}
   \RETURN $\hat{\eta}_j(\hat{\mathcal{C}_0'})  := \hat{\beta}_{tj | \mathbf{X}_{\{1,\ldots, t-1\}}}~\times~\mathrm{aver}(\hat{\Theta}_{jp}(\hat{\mathcal{C}_0'}))$ for $j = t+1,\ldots,p-1$. 
     \end{algorithmic}
\end{algorithm}


\vspace{-0.1in}
\begin{lemma}\label{lemma: sub-LSEM} Let $\mathbf{X}$ be as in Section \ref{subsection: setup} and let $\mathbf{X}'$, $\mathcal{G}_0'$, $B_{\mathcal{G}_0'}$ and $\boldsymbol{\epsilon}'$ be as in Section \ref{subsection: notation}. Then, the conditional expectation $\Exp[\mathbf{X}'| \mathbf{X}_{\{1,\ldots, t\}}]$ is linear in $\mathbf{X}_{\{1,\ldots, t\}}$. Further, $\mathbf{X}^{\dagger} := \mathbf{X}' - \Exp[\mathbf{X}'|\mathbf{X}_{\{1,\ldots, t\}}] $ satisfies: $\mathbf{X}^{\dagger} = B_{\mathcal{G}_0'}^T \mathbf{X}^{\dagger} + \epsilon'$, and the distribution of $\mathbf{X}^{\dagger}$ is faithful to $\mathcal{G}_0'$.
\end{lemma}

{\color{black}
\vspace{-0.1in}
\begin{lemma}\label{lemma: partial correlation}
\tcr{Let $\hat{\rho}_{ik|S}$ denote the sample partial correlation between $X_i$ and $X_k$ given $\mathbf{X}_S$
\tcm{and 
$\hat{\rho}_{ik|S}^{\dagger}$} 
 the sample partial correlation between $X_i^{\dagger}$ and $X_k^{\dagger}$ given $\mathbf{X}_S^{\dagger}$ computed based on the residuals $\mathbf{r}_{j}$'s defined in Algorithm \ref{algorithm: CPDAG estimation}. \tcm{Then, for all $i,k \in \{t+1,\ldots,p-1\}$, $i\neq k$ and $S \subseteq \{t+1,\ldots,p-1\}\setminus \{i, k\}$,}}
\[
\tcr{\hat{\rho}_{ik|S}^{\dagger2} \; = \; \hat{\rho}_{ik|S\cup\{1,\ldots,t\}}^2.}
\]
\end{lemma}
\tcr{\tcm{Lemmas \ref{lemma: sub-LSEM}\tcR{--}
\ref{lemma: partial correlation}} \tcm{justify} the estimation of $\mathcal{C}_0'$ from the residuals $\mathbf{r}_{j}$'s in Algorithm \ref{algorithm: CPDAG estimation}. This is because both \tcm{the} (AR)GES and PC \tcm{algorithms} are designed to estimate the CPDAG corresponding to a faithful DAG of an LSEM based on the squared sample partial correlations. In particular, (AR)GES sequentially adds and deletes edges based on the $\log(1 - \hat{\rho}_{ik|S}^2)$ values (see Lemma 5.1 and Section 6 of \citet{NandyHauserMaathuis16}) and the PC algorithm uses the $\log(\frac{1 + |\hat{\rho}_{ik|S}|}{1 - |\hat{\rho}_{ik|S}|})$ values (see Section 2.2.2 of \citet{KalischBuehlmann07a}).}
}

The main difference between $\hat{\Theta}_{jp}(\hat{\mathcal{C}_0'})$ \tcr{above} and the corresponding original IDA estimator of \citet{MaathuisKalischBuehlmann09} is that we \emph{always include} $\mathbf{X}_{\{1,\ldots, t\}}$ in the adjustment set, leveraging the fact that $X_j$ is not a direct cause of the treatment variable $X_t$ and the \tcr{confounders} 
$\mathbf{X}_{\{1,\ldots, t-1\}}$.
Further, note that computing $\MEC(\hat{\mathcal{C}}_0')$ can be computationally infeasible for a large CPDAG $\hat{\mathcal{C}}_0'$ \citep{MaathuisKalischBuehlmann09}. This computation bottleneck can be relieved by directly obtaining the multi-set of parent sets $\mathcal{PA}_{\hat{\mathcal{C}_0'}}(X_{j}) = \{\mathbf{Pa}_{\mathcal{G}}(X_j) : \mathcal{G} \in \MEC(\hat{\mathcal{C}}_0') \}$ from $\hat{\mathcal{C}}_0'$ without computing $\MEC(\hat{\mathcal{C}}_0')$ via Algorithm 3 of \citet{NandyMaathuisRichardson17}.
We note that the output of Algorithm 3 of \citet{NandyMaathuisRichardson17} and $\mathcal{PA}_{\hat{\mathcal{C}}_0}(X_{j})$ may not be the same multi-set, but Theorem 5.1 of \cite{NandyMaathuisRichardson17} guarantees that they are equivalent multi-sets in the sense that they have the same distinct elements and the ratio of the multiplicities of any two elements in the output of Algorithm 3 of \citet{NandyMaathuisRichardson17} equals the ratio of their multiplicities in $\mathcal{PA}_{\hat{\mathcal{C}}_0}(X_{j})$. Therefore,
using the output of Algorithm 3 of \citet{NandyMaathuisRichardson17} instead of $\MEC(\hat{\mathcal{C}}_0')$ makes no difference in obtaining $\mathrm{aver}(\hat{\Theta}_{jp}(\hat{\mathcal{C}_0'}))$. Thus, for simplicity, we can safely pretend that we use $\MEC(\hat{\mathcal{C}}_0')$ for computing $\mathrm{aver}(\hat{\Theta}_{jp}(\hat{\mathcal{C}_0'}))$ in the rest of the paper.

\section{Consistency in High-Dimensional Settings}\label{section: high-dimensional consistency}

We now consider an asymptotic scenario where the sample size $n$ and the number of potential mediators $(p-t-1)$ in $\mathbf{X}$ grows to infinity. (We consider $t \geq 1$ to be fixed.) We prove high-dimensional consistency of the MIDA estimators 
$\hat{\eta}_{\tcR{j}}(\hat{\mathcal{C}}_0')$ 
defined in Algorithm \ref{algorithm: CPDAG estimation}, 
when\tcR{ever} the CPDAG $\hat{\mathcal{C}}_0'$ is estimated consistently. We 
note \tcR{here} that \tcR{such} high-dimensional consistency in the CPDAG estimation holds under the following assumptions, and some additional assumptions (e.g., see \citet{KalischBuehlmann07a, NandyHauserMaathuis16}), when $\hat{\mathcal{C}}_0'$ \tcm{is} estimated using (AR)GES or PC (see Lemmas \ref{lemma: sub-LSEM} and \ref{lemma: partial correlation} and the subsequent discussion).


\begin{assumption}[LSEM with sub-Gaussian error variables]\label{assumption: LSEM} $\mathbf{X}_{}$ is generated from a linear SEM $(B_{\mathcal{G}_{0}},\boldsymbol{\epsilon}_{})$ with sub-Gaussian error variables satisfying $\underset{1\leq i \leq p}{\max} ||\epsilon_{i}||_{\psi_2} \leq C_1$ for some absolute constant $C_1>0$, where $||\cdot||_{\psi_2}$ denotes the sub-Gaussian norm \tcr{given in Definition \ref{def:psialnorm} (in \tcR{Appendix \ref{sec:technicaltools} of} the \hyperref[supp_mat]{Supplement}).} 
\end{assumption}

\begin{assumption}[High-dimensional setting]\label{assumption: high-dim} 
$p = O(n^{a})$ for some $a \geq 0$. 
\end{assumption}

\begin{assumption}[Sparsity condition]\label{assumption: sparsity} 
Let $q_{} := \max_{\tcr{t < j \leq p_{}-1}} |\mathbf{Adj}_{\mathcal{C}_{0}'}(X_{j})|$ \tcr{denote} the maximum degree in $\mathcal{C}_{0}'$.  Then, $q_{} = O(n^{1-b_1})$ for some $0 < b_1 \leq 1$. 
\end{assumption}

\begin{assumption}[Structure learning consistency]\label{assumption: structure learning consistency}
The estimated \tcR{CPDAG} $\hat{\mathcal{C}_0'}$ in Algorithm \ref{algorithm: CPDAG estimation} is a consistent estimator of $\mathcal{C}_0'$, i.e. $\Prob(\hat{\mathcal{C}}_{0}' \neq \mathcal{C}_{0}') \longrightarrow 0$.
\end{assumption}

\begin{assumption}[Bounds on the eigenvalues of covariance matrices]\label{assumption: eigenvalue condition} For any $(q+t+2) \times (q+t+2)$ principal submatrix $\Sigma_{}$ of $\Sigma_{0} = \Cov(\mathbf{X}_{})$, 
$$C_2 \; \leq \; 1/\MLtwonorm{\Sigma_{}^{-1}} \; \leq \; \MLtwonorm{\Sigma_{}^{}} \; \leq \; C_3,$$ for some absolute constants $C_2,C_3>0$, where $\MLtwonorm{\cdot}$ denotes the spectral norm \tcr{(as defined in Section \ref{subsection: notation})} and $q$ is \tcr{as defined in 
Assumption \ref{assumption: sparsity} above.}
\end{assumption}

\begin{theorem}[\tcR{Uniform consistency of MIDA}]\label{theorem: high-dimensional consistency}
Let $\hat{\eta}_{j}(\hat{\mathcal{C}}_{0}') = \hat{\beta}_{tj | \mathbf{X}_{\{1,\ldots, t-1\}}} \times\mathrm{aver}(\hat{\Theta}_{jp}(\hat{\mathcal{C}}_{0}'))$ denote the output of Algorithm
\ref{algorithm: CPDAG estimation}. \tcr{Then, under Assumptions \ref{assumption: LSEM}\tcR{--}\ref{assumption: eigenvalue condition}, we have:}
\begin{enumerate}
\item $\underset{t < j < p_{}}{\max}\left| \mathrm{aver}(\hat{\Theta}_{jp}(\hat{\mathcal{C}}_{0}')) - \mathrm{aver}(\Theta_{jp})\right| ~\overset{\Prob}{\longrightarrow}~ 0$, \; and
\item $ \underset{t < j < p_{}}{\max}\left| \hat{\eta}_{j}(\hat{\mathcal{C}}_{0}') - \eta_j(\mathcal{C}_{0}')\right| ~\overset{\Prob}{\longrightarrow}~ 0.$
\end{enumerate}
%
\end{theorem}
\noindent Note that the high dimensional consistency results \tcr{for} 
\tcr{(joint-)}IDA estimators were proven only for LSEMs with Gaussian errors \citep{MaathuisKalischBuehlmann09, NandyMaathuisRichardson17}. Theorem \ref{theorem: high-dimensional consistency} \tcr{thus} \emph{extends} the existing high-dimensional consistency results \tcr{for} 
IDA estimators to LSEMs with \emph{sub-Gaussian} errors. \tcR{Finally, it is worth noting again that the results above do \emph{not} depend on the nature of the CPDAG estimation \tcr{method} 
as long as it is consistent.}


{\color{black}
\subsection{Discussion on Assumption \ref{assumption: LSEM}}\label{sec:assumption_discussion}
\tcr{We demonstrate the adaptivity of the high-dimensional consistency result \tcm{above for} 
the MIDA estimator, with respect to some of the recent efforts in relaxing the linearity assumption for the IDA-type estimators \citep{NandyMaathuisRichardson17, FrotEtAl19}. We outline \tcm{here} the modifications needed to relax Assumption \ref{assumption: LSEM}, but we refrain from a mere repetition of the existing theoretical analyses 
\tcm{for the sake of} brevity.}

\par\smallskip
\noindent{\tcm{\emph{Relaxing linearity.}}}
\tcr{The linear sub-Gaussian setting is a key requirement for the high-dimensional inference \tcm{results for MIDA} presented in Section \ref{section: asymptotic distribution} \tcm{later}. However, the high-dimensional consistency result \tcm{given here}
can \tcm{indeed} be extended beyond the linear sub-Gaussian setting by combining some existing results \citep{HarrisDrton13, HanLiu17, NandyMaathuisRichardson17, NandyHauserMaathuis16, FrotEtAl19}. In particular, we can \tcm{establish} a high-dimensional consistency \tcm{result} under \tcm{only} the assumption \tcm{that} $\mathbf{X}$ follows a \emph{transelliptical distribution} (Definition 2.2 of \citet{HanLiu17}), \tcm{i.e.} 
$(f_1(X_1),\ldots,f_p(X_p))^T$ has an elliptical distribution \tcm{(e.g.\tcR{,} Gaussian)}
for some increasing (or decreasing) functions $f_1,\ldots,f_p$.} 
\tcr{When the distribution of $(f_1(X_1),\ldots,f_p(X_p))^T$ is multivariate Gaussian, the distribution of $\mathbf{X}$ is called \emph{nonparanormal}. Under the nonparanormal assumption, \citet{HarrisDrton13} proved a
high-dimensional consistency \tcm{result for} the so-called Rank PC algorithm, and \citet{NandyMaathuisRichardson17} extended this result to a modified version of the (joint-) IDA estimator. The high-dimensional consistency of (AR)GES under the nonparanormal distribution, \tcm{as well as} the more general transelliptical distribution, has been proved in \citet{NandyHauserMaathuis16, FrotEtAl19}.}

\tcr{A high-dimensional consistency \tcm{result for} 
  a modified version \tcm{`(M)IDA' of MIDA} under the transelliptical assumption can be obtained by combining the structure learning consistency result of \citet{FrotEtAl19} with the proofs of Theorem 7.1 of \citet{NandyMaathuisRichardson17} and our Theorem \ref{theorem: high-dimensional consistency}. The required modification \tcm{for (M)IDA} is to \tcm{simply} apply Algorithm \ref{algorithm: CPDAG estimation} on a transformed sample rank correlation matrix $\sin((\pi/2)\hat{T})$ instead of the sample covariance matrix $\hat{\Sigma}$ \tcm{itself}, where the sine function is applied element-wise and $\hat{T}$ is the Kendall's rank correlation matrix. To this end, note that the first three steps of Algorithm \ref{algorithm: CPDAG estimation} \tcm{are} equivalent to applying \tcm{the (AR)GES algorithm} on \tcR{the covariance matrix}}
\begin{align}\label{eq: conditional covariance}
&\tcR{\hat{\Sigma}_{\mathbf{X}_{\{t+1,\ldots,p-1\}}\mathbf{X}_{\{t+1,\ldots,p-1\}}} \; - \; 
\hat{\Sigma}_{\mathbf{X}_{\{t+1,\ldots,p-1\}}\mathbf{X}_{\{1,\ldots,t\}}}\hat{\Sigma}_{\mathbf{X}_{\{1,\ldots,t\}}\mathbf{X}_{\{1,\ldots,t\}}}^{-1}\hat{\Sigma}_{\mathbf{X}_{\{1,\ldots,t\}}\mathbf{X}_{\{t+1,\ldots,p-1\}}}.
}
\end{align}

\par\smallskip
\noindent\tcr{\tcm{\emph{\tcR{Handling} the case of hidden confounders.}} Finally, recall that the independence of the error variables in Definition \ref{definition: LSEM} corresponds to the no hidden \tcm{(or unmeasured)} confounder (that is a common cause of more than one observed variable) assumption. This assumption has been relaxed in \citet{FrotEtAl19}. In particular, the authors considered a setting with a few hidden variables that have a direct effect on many of the observed variables and derived a consistent method for estimating the underlying CPDAG corresponding to the observed variables. The main idea is to apply a low-rank plus sparse decomposition \citep{ChandrasekaranEtAl12} of the inverse covariance matrix to remove the influence of the hidden confounders before applying a structure learning method that assumes the absence of hidden confounders. Furthermore,
the authors proposed to use the estimated CPDAG and an estimate of the conditional covariance matrix given the hidden confounders to obtain IDA-type estimators of total causal effects under the linear sub-Gaussian assumption or the transelliptical assumption. The same techniques can be applied to the covariance matrix 
\tcR{in} \eqref{eq: conditional covariance} or the corresponding transformed rank correlation matrix to obtain a consistent MIDA estimator in the presence of a small number of highly influential hidden variables.}
}

\section{Linear Regression over Varying Subsets of High-Dimensional Covariates}\label{section: linear regression theory}

This section considers linear regression over varying subsets of high-dimensional covariates in a general setting, \tcr{and derives uniform bounds and
first order expansions for the resulting estimators.} \tcr{These} results \tcr{are} 
used in Section \ref{section: asymptotic distribution} to derive the asymptotic distributions of the \tcr{estimators} of interventional and mediation effects. 
While derived primarily for establishing the (uniform) asymptotic normality \tcr{and inferential tools for} 
our proposed estimators, these results are applicable far more generally to any setting involving linear regressions 
over varying (non-random) subsets of high-dimensional regressors, and may  be of independent interest. For notational simplicity and clarity of exposition, we therefore derive them under more general and standard notations where $Y$ and $\bX$ denote a generic response and a (high-dimen\tcR{s}ional) covariate vector, respectively. All our results here are non-asymptotic.

\paragraph*{Basic setup and definitions\tcR{.}}\label{susubsec:1:linreg} Let $\Dsc_n := \{ \bZ_i \equiv (Y_i, \bX_i)\}_{i=1}^n$ denote the observed data consisting of $n$ i.i.d. realizations of $\bZ := (Y, \bX)$, where $Y \in \R$, $\bX \in \R^p$ and $p \equiv p_n$ is allowed to diverge with $n$. \tcr{Neither $Y$ nor $\bX$ is needed to be centered (i.e. zero-mean) and/or continuous.} Let $\muY := \E(Y)$, $\bmu := \E(\bX) $ and $\bSigma := \E\{ (\bX - \bmu) (\bX - \bmu)^T\} \equiv \Cov(\bX)$, where we assume $\bSigma \succ \mathbf{0}$.  Further, let $\Ybar := n^{-1} \sum_{i=1}^n Y_i$, $\bXbar := n^{-1} \sum_{i=1}^n \bX_i$ and $\Jsc := \{1, \hdots, p\}$.

\par\smallskip
\tcr{\tcR{Let} $\Omega_{\Jsc}$ denote 
the collection of all possible subsets of $\Jsc$.  For any $S \in \Omega_{\Jsc}$ with $|S| = s \leq p$, and 
any vector $\bv = (v_i)_{i=1}^p \in \R^p$, let $\bv_S \in \R^s$ denote the restriction of $\bv$ onto $S$, i.e. for $S = \{i_1, \hdots, i_s\} \subseteq \Jsc$,
$\bv_S = (v_{i_j})_{j=1}^s$.} 
Let $\bXS$, $\{\bXSi\}_{i=1}^n$, $\bmuS$ and $\bXbarS$ respectively denote the restrictions of $\bX$, $\{\bX_i\}_{i=1}^n$, $\bmu$ and $\bXbar$ onto $S$. 
Let $\bSigmaS := \Cov(\bXS)$ and $\bSigmaSY := \Cov(Y,\bXS)$, and define: 
\begin{align*}
 & \tcr{\;\;\bSigmahatS  \;\; := \; \frac{1}{n} \sum_{i=1}^n (\bXSi - \bXbarS)(\bXSi - \bXbarS)^T, \;\; \bGammahatS  \; := \; (\bXbarS - \bmuS) (\bXbarS - \bmuS)^T,} \nonumber \\
 & \tcr{\bSigmahatSY  \; := \; \frac{1}{n} \sum_{i=1}^n (Y_i - \Ybar)(\bXSi - \bXbarS), \;\; \mbox{and} \;\; \bGammahatSY \; := \; (\Ybar -\muY)(\bXbarS - \bmuS).} \nonumber 
\end{align*}

\paragraph*{Varying subset linear regression estimator(s)\tcR{.}} 
Let $\Ssc \subseteq \Omega_{\Jsc}$ denote any collection of subsets of $\Jsc$ \tcr{with} 
$\max \{s := |S|: S \in \Ssc\} \leq \qn$, for some $\qn \equiv \qnS$ $\leq \min(n,p_n)$, and let $\Ln \equiv \LnS := |\Ssc|$. \tcr{We now
consider 
linear regression(s) of $Y$ on $\bXS$, for all $S \in \Ssc$, via the \emph{ordinary least squares} (OLS) {estimator}, $\bbetahatS$, defined, 
along with its corresponding \emph{target parameter}, $\bbetaS$, as follows.} 
\begin{align}
& \tcr{\bbetaS \; := \; \underset{\bbeta \in \R^s}{\arg \min} \; \E[\{ (Y - \muY) - (\bXS - \bmuS)^T \bbeta\}^2] \;\; \equiv \; \bSigmaS^{-1} \bSigmaSY, \;\; \mbox{and}} \label{eq:targetdef}  \\  
& \tcr{\bbetahatS \; := \;  \underset{\bbeta \in \R^s}{\arg \min} \; \frac{1}{n} \sum_{i=1}^n\{ (Y_i - \Ybar) - (\bXSi - \bXbarS)^T \bbeta\}^2 \;\; \equiv \; \bSigmahatS^{-1} \bSigmahatSY.} \label{eq:olsdef} 
\end{align}
Since we are only interested in the 
\tcr{regression coefficients}, we circumvent the need for any nuisance intercept terms by appropriately centering $Y$ and $\bXS$ in both (\ref{eq:targetdef}) and (\ref{eq:olsdef}). The existence and uniqueness of $\bbetaS$ and $\bbetahatS$ 
in (\ref{eq:targetdef})\tcR{--}(\ref{eq:olsdef}) are both guaranteed for any $S \in \Ssc$ since $\bSigma \succ \bzero$, so that $\bSigmaS$ is invertible, and $\bSigmahatS$ is invertible almost surely (a.s.) as $|S| \equiv s \leq n$. Further, note that throughout the formulations in (\ref{eq:targetdef})--(\ref{eq:olsdef}), we make \emph{no} assumptions on the existence of a true linear model between $Y$ and $\bXS$ for \emph{any} $S$. The target parameter $\bbetaS$ is well-defined \emph{regardless} of any such model assumptions and simply denotes the coefficients in the best (in the $L_2$ sense) linear predictor of $Y$ given $\bXS$. \tcr{Our framework is 
thus completely \emph{model free} in this sense.} 

\paragraph*{Decomposition of $(\bbetahatS - \bbetaS)$\tcR{.}} For notational simplicity, \tcr{define: $\bXtilS := \bXS - \bmuS$ and $\Ytil := Y - \muY$, the centered versions of $\bX$ and $Y$, and let $\bpsiS(\bZ):= \bXtilS (\Ytil - \bXtilS^{T} \bbetaS)$, where we note that $\E\{\bpsiS(\bZ)\}  =  \bzero$ by the definition of $\beta_S$ in \eqref{eq:targetdef}. Using the estimating equations (\ref{eq:targetdef})--(\ref{eq:olsdef}), it is then straightforward to show that $\bbetahatS - \bbetaS$, for any $S \in \Ssc$, satisfies a \emph{deterministic} decomposition:}
\begin{eqnarray}
 &&  \qquad \quad \tcr{\bbetahatS - \bbetaS \; = \; \frac{1}{n} \sum_{i=1}^n \bPsiS(\bZ_i) +  \bTnS + \bRnS}, \quad \mbox{where}  \label{eq:maindecomp} \\
 && \bPsiS(\bZ) \; := \; \bSigmaS^{-1} \bpsiS(\bZ) \; \equiv \;  \bSigmaS^{-1} \bXtilS (\Ytil - \bXtilS^{T} \bbetaS) \;\; \mbox{with} \;\; \tcr{\E\{\bPsiS(\bZ)\}  \tcR{\;  = \;} \bzero},  \nonumber \\
 && \tcr{ \bTnS \tcR{\; := \;}  \frac{1}{n} (\bSigmahatS^{-1} - \bSigmaS^{-1}) \sum_{i=1}^n \bpsiS(\bZ_i)  
 \;\; \mbox{and} \;\;  \bRnS   \tcR{\; := \;} \bSigmahatS^{-1} (\bGammahatS \bbetaS -  \bGammahatSY).} \nonumber
\end{eqnarray}

For a \emph{single} $S$ and under classical asymptotics (i.e. $s$ is fixed and $n \rightarrow \infty$), \tcr{standard results from $M$-estimation theory \citep[e.g.]{VanderVaartWellner96:EmpiricalProcessesBook,VanderVaart98} 
imply that under mild 
conditions,} $\bbetahatS$ is a $\sqrt{n}$-consistent and asymptotically normal (CAN) estimator of $\bbetaS$ and \tcr{admits an \emph{asymptotically linear expansion} (ALE): $\bbetahatS - \bbetaS = n^{-1}   \sum_{i=1}^n \bPsiS(\bZ_i)$ $+ o_{\P}(n^{-1/2})$, with an \emph{influence function} (IF): $\bPsiS(\bZ)$}, so that $\sqrt{n}(\bbetahatS - \bbetaS)$ 
is asymptotically normal with mean $\bzero$ and \tcr{variance:} $\Cov\{\bPsiS(\bZ)\}$. Further, even when $s \equiv |S|$ is allowed to diverge, \tcr{it is also well known (e.g., see 
\citet{Portnoy84, Portnoy85, Portnoy86, Portnoy88}) that under suitable regularity conditions and if $s = o(n)$, $\Ltwonorm{\bbetahatS - \bbetaS} = O_{\P}(\sqrt{s/n})$ 
and $\Ltwonorm{\bbetahatS - \bbetaS - n^{-1}   \sum_{i=1}^n \bPsiS(\bZ_i)}$ $= O_{\P}(s/n)$, so that whenever $s = o(\sqrt{n})$}, 
$\bbetahatS$ is a CAN estimator of $\bbetaS$ and admits an ALE with IF $\bPsiS(\bZ)$. However, these results are all asymptotic \tcr{in nature,} and more importantly, apply only to a single \tcr{set} $S$.

\paragraph*{\tcr{Contributions}\tcR{.}} \tcr{Our main \emph{challenges} lie in the fact}
 that we have a \emph{family} of estimators based on a collection of subsets $\{\bXS\}_{S \in \Ssc}$ of $\bX$, where $|\Ssc|$ itself is possibly large and further, for each $S \in \Ssc$, $\bXS$ may be high-dimensional with $s \leq q_n$ allowed to diverge with $n$. \tcr{Under 
such a 
setting, we aim} 
to provide inferential tools for our family of estimators $\{\bbetahatS\}_{S\in \Ssc}$ and their derived functionals. 
We achieve this by providing (uniform) ALEs for $\{\bbetahatS\}_{S\in \Ssc}$ \tcR{in Theorem \ref{linreg:main_thm1}}, whereby we control the remainder terms $\bTnS$ and $\bRnS$ in (\ref{eq:maindecomp}) \emph{uniformly over $S \in \Ssc$} based on non-asymptotic bounds  \tcr{for $\sup _{S \in \Ssc} \Ltwonorm{\bTnS + \bRnS}$ that establishes} 
their uniform convergence rates \tcR{(Remark \ref{linreg:rem_thm1})}. Note that the potentially diverging sizes of $\Ssc$ and each $S \in \Ssc$  necessitate such non-asymptotic analyses. Lastly, apart from the (second order) error terms $\bTnS$ and $\bRnS$, we also provide uniform convergence rates of the first order term: $n^{-1} \sum_{i=1}^n \bPsiS(\bZ_i)$ in (\ref{eq:maindecomp}) under the $L_2$ norm, thereby establishing the rate of $\sup_{S \in \Ssc} \| \bbetahatS - \bbetaS \|_2$. Further, for \tcR{\emph{linear functionals}} of $\{\bbetaS\}_{S \in \Ssc}$, we also provide results on $\sqrt{n}$-consistency and asymptotic normality 
for the corresponding linear functionals of the estimators $\{\bbetahatS\}_{S\in \Ssc}$ \tcR{(Remark \ref{remark: linear functionals})}. Such results would be useful for establishing our results in Section \ref{section: asymptotic distribution} regarding asymptotic distribution of the IDA based estimators.


\subsection{Uniform ALEs for OLS: Non-Asymptotic Bounds and Uniform Convergence Rates for All Terms in (\ref{eq:maindecomp})}\label{subsec:linreg}

We first state our main assumptions and define a few related quantities that will appear in our results. \tcr{We present our main result in Theorem \ref{linreg:main_thm1} below. Its proof (given in \tcR{Appendix} \ref{proof:main:thm1}) 
also involves two useful supporting lemmas, Lemma \ref{linreg:main_lem1} and Lemma \ref{linreg:main_lem2}, which may be of independent interest. These lemmas are also given in  \tcR{Appendix} \ref{proof:main:thm1}. 
} 

\begin{assumption}[Main assumptions and some definitions]\label{linreg:main_assmpn}
(i) We assume that \tcr{$\Ytil \equiv Y - \muY$} is sub-Gaussian and \tcr{$\bXtilS \equiv \bXS - \bmuS$} is sub-Gaussian uniformly in $S \in \Ssc$, \tcr{and also that
$\bSigmaS$ is well conditioned 
uniformly in $S \in \Ssc$.}
Specifically, for some constants $\sigmaY, \sigmaXSsc \in [0, \infty)$ \tcr{and $\lambdainfSsc, \lambdasupSsc \in (0,\infty)$,}
\begin{eqnarray}
&& \psitwonorm{Y - \muY} \; \leq \; \sigmaY, \;\;\; \underset{S \in \Ssc}{\sup} \; \psitwonormstar{\bXS - \bmuS} \; \leq \; \sigmaXSsc, \quad \mbox{and} \nonumber \\
&&  0 \; < \; \lambdainfSsc \; \leq \; \underset{S \in \Ssc}{\inf} \; \lambdamin(\bSigmaS) \; \leq \; \underset{S \in \Ssc}{\sup} \; \lambdamax(\bSigmaS) \; \leq \; \lambdasupSsc \; < \; \infty, \label{linreg:main_assmpn:eq1}
\end{eqnarray}
where $\psitwonorm{\cdot}$ and $\psitwonormstar{\cdot}$ denote the sub-Gaussian norms \tcr{as in Definitions 
\ref{def:psialnorm}\tcR{--}\ref{def:subgaussvect}, respectively.}
Let us further define the constant $\KSsc := C_1 \sigmaXSsc^2 \frac{\lambdasupSsc}{\lambdainfSsc} > 0$,
where $C_1$ is the same absolute constant as in Lemma \ref{lem:7:matconc} (and also same as the constant given in Theorem 4.7.1 and Exercise 4.7.3 of \citet{Vershynin18}). 
\par\medskip
\noindent 
(ii) Let $\bZS := (Y, \bXS)$, $\bnuS := \E(\bZS)$ and $\bXiS := \Cov(\bZS)$. Then, we also assume that $\bXiS$ is well-conditioned uniformly in $S \in \Ssc$. Specifically, for some constants $\lambdatilinfSsc, \lambdatilsupSsc \in (0, \infty)$,
\begin{equation*}
0 \; < \; \lambdatilinfSsc \; \leq \; \underset{S \in \Ssc}{\inf} \; \lambdamin(\bXiS)  \; \leq \; \underset{S \in \Ssc}{\sup} \; \lambdamax(\bXiS)  \; \leq  \; \lambdatilsupSsc \; < \; \infty. \label{linreg:main_assmpn:eq2}
\end{equation*}
Further, let \tcr{$\sigmatilZSsc := \sigmaY + \sigmaXSsc$} and define the constant $\KtilSsc :=  C_1 \sigmatilZSsc^2 \frac{\lambdatilsupSsc}{\lambdatilinfSsc} > 0$,
where $C_1 > 0$ is the same absolute constant as in part (i) above. 
\end{assumption}


\begin{theorem}[Uniform bounds and convergence rates for all the terms in (\ref{eq:maindecomp})]\label{linreg:main_thm1}
Consider any $\Ssc \subseteq \Omega_{\Jsc}$ with $|\Ssc| := \Ln \equiv \LnS $ and $\sup_{S \in \Ssc} |S| \leq \qn \equiv \qnS \leq \min(n,p_n)$, and suppose Assumption \ref{linreg:main_assmpn} holds. Let $\rn := \qn + \log \Ln$, $\rntil := \rn + 1$ and $\CSsc := \sqrt{2} \sigmaY \lambdainfSsc^{-1/2}$. For any $c > 0$, let $\bar{c} := c+ 1$ and define:
\begin{align*}
& \epsilon_{n,1}(c,\rn) \; := \; \bar{c}\KSsc \left( \sqrt{\frac{\rn}{n}} + \frac{\rn}{n} \right), \;\; \epsilon_{n,2}(c,\rn) \; := \; \bar{c}\KtilSsc \left( \sqrt{\frac{\rntil}{n}} + \frac{\rntil}{n} \right); \\
& \eta_{n,1} (c, \rn) \; := \; 32 \bar{c} \KSsc \frac{\rn}{n} + \frac{\lambdasupSsc}{n}, \;\; \eta_{n,2}(c, \rn) \; := \;  32 \bar{c} \KtilSsc \frac{\rntil}{n} + \frac{\lambdatilsupSsc}{n}; \;\; \mbox{and} \\
& \delta_{n}(c,\rn) \;\; := \; \bar{c} \KSsc^* \left( \sqrt{\frac{\rn}{n}} + \frac{33 \rn}{n} \right) + \frac{2}{n}\frac{\lambdasupSsc}{\lambdainfSsc^2},
\end{align*}
where $(\KSsc, \KtilSsc, \lambdasupSsc, \lambdainfSsc, \lambdatilsupSsc, \lambdatilinfSsc, \sigmaY)$ are as in Assumption \ref{linreg:main_assmpn} and $\KSsc^* := 2 \lambdainfSsc^{-2} \KSsc$. Further, let $c^* > 0$ be any constant that satisfies: 
\begin{equation*}
(c^*+1)\KSsc \left( \sqrt{\frac{\rn}{n}} + \frac{33 \rn}{n} \right)  + \frac{\lambdasupSsc}{n} \; \leq \; \frac{1}{2} \lambdainfSsc.  
\end{equation*}
(i) Then, for any such constant $c^* > 0$, and for any $c > 0$, we have the following bounds. With probability at least $1 - 8 \exp( - c\rn) - 4 \exp(-c^* \rn)$,
{\color{black}
\begin{align*}
\tcr{\underset{S \in \Ssc}{\sup} \Ltwonorm{\bTnS}}  & \tcr{~\leq~ \delta_{n}(c,\rn) \left\{ \epsilon_{n,1}(c,\rn)\CSsc  + \epsilon_{n,2}(c,\rn) \right\} \; \lesssim \; \mathbb{C}_{\Ssc} \bar{c}^2 \frac{\rn}{n}, \;\; \mbox{and}}  \nonumber \\
 \tcr{\underset{S \in \Ssc}{\sup} \Ltwonorm{\bRnS}} & \tcr{~\leq~ \{\delta_{n}(c,\rn) + \lambdainfSsc^{-1} \} \{ \eta_{n,1}(c,\rn)\CSsc  + \eta_{n,2}(c,\rn) \} \; \lesssim \;  \mathbb{C}_{\Ssc} \bar{c} \frac{\rn}{n}.}  \nonumber 
\end{align*}
\tcr{(ii) Further, for any $c > 0$, the first order term $n^{-1} \sum_{i=1}^n \bPsiS(\bZ_i)$ in (\ref{eq:maindecomp}) satisfies the following bound. 
With probability at least $1 - 4 \exp( - c\rn)$,}
\begin{align*}
\tcr{\underset{S \in \Ssc}{\sup} \Ltwonorm{\frac{1}{n} \sum_{i=1}^n \bPsiS(\bZ_i)}}  & \tcr{~\leq~ \lambdainfSsc^{-1} \left\{ \epsilon_{n,1}(c,\rn)\CSsc  + \epsilon_{n,2}(c,\rn) \right\} \; \lesssim \; \mathbb{C}_{\Ssc} \bar{c} \sqrt{\frac{\rn}{n}}.} \nonumber 
\end{align*}
}
Here, $\mathbb{C}_{\Ssc}$ denotes a generic constant (possibly different in each bound) depending only on $\Ssc$, and $`\lesssim'$ denotes inequality upto multiplicative constants.
\end{theorem}

\begin{remark}\label{linreg:rem_thm1}
\emph{
The two bounds in result (i) of Theorem \ref{linreg:main_thm1} also imply, in particular, that with high probability,
\begin{equation*}
\underset{S \in \Ssc}{\sup} \Ltwonorm{(\bbetahatS - \bbetaS) - \frac{1}{n} \sum_{i=1}^n \bPsiS(\bZ_i)} \; \equiv \; \underset{S \in \Ssc}{\sup} \Ltwonorm{\bTnS + \bRnS} \;\; \lesssim \; \frac{\rn}{n},
\end{equation*}
\tcr{thereby establishing (non-asymptotically) that  the uniform (in $S \in \Ssc$) convergence rate (under the $L_2$ norm) of the second order terms in the ALE (\ref{eq:maindecomp}) of $\bbetahatS - \bbetaS$ is $O_{\P}(\rn/n)$.}
}

\emph{
Further, the bound in result (ii) also establishes (non-asymptotically) the uniform (in $S \in \Ssc$) convergence rate (under the $L_2$ norm) of the first order term in the ALE (\ref{eq:maindecomp}) to be $O_{\P}(\sqrt{\rn/n})$. \tcR{Consequently, it} 
establishes that $\sup_{S \in \Ssc} \Ltwonorm{\bbetahatS - \bbetaS} = O_{\P}(\sqrt{\rn/n} + \rn/n).$
}
\end{remark}

\begin{remark}[ALEs and asymptotic normality for linear functionals of $\{ \bbetahatS\}_{S \in \Ssc}$]\label{remark: linear functionals}
\emph{
Let $\ASsc := \{ \baS \in \R^s : S \in \Ssc\}$ \tcr{denote} any collection of (known) vectors with $\sum_{S \in \Ssc}\Ltwonorm{\baS} = O(1)$. Consider the \emph{\tcR{linear} functional} \tcR{of $\{ \bbeta_S\}_{S \in \Ssc}$} given by\tcR{:} $\beta(\ASsc) := \sum_{S \in \Ssc} \baS^T \bbetaS$\tcR{,}
 and its \tcR{corresponding} estimator\tcR{:} $\betahat(\ASsc) := \sum_{S \in \Ssc} \baS^T \bbetahatS$. Then, as a direct consequence of Theorem \ref{linreg:main_thm1}, $\betahat(\ASsc) - \beta(\ASsc)$ satisfies the following ALE\tcR{:}
\begin{align}
& \sqrt{n} \{ \betahat(\ASsc) - \beta(\ASsc) \} \;\; = \; \frac{1}{\sqrt{n}} \sum_{i=1}^n \xi_{\ASsc}(\bZ_i) + \Rsc_{n,\ASsc} \;\; =: \; \sqrt{n} \S_{n,\ASsc} + \Rsc_{n,\ASsc}, \label{eq:linfunctional:ale} \nonumber \\
& \quad \mbox{where} \;\; \xi_{\ASsc}(\bZ) \; := \; \sum_{S \in \Ssc} \baS^T \bPsiS(\bZ) \;\; \mbox{and} \;\; \Rsc_{n,\ASsc} \; := \; \sum_{S \in \Ssc} \baS^T(\bTnS + \bRnS), \nonumber \\
& \quad \mbox{with} \;\; |\Rsc_{n,\ASsc}| \; \leq \; \supS \Ltwonorm{\bTnS + \bRnS} \left( \sum_{S \in \Ssc}\Ltwonorm{\baS} \right) \; = \; O_{\P}\left( \frac{\rn}{\sqrt{n}}\right). \nonumber \qed
\end{align}
Thus, $\sqrt{n} \{ \betahat(\ASsc) - \beta(\ASsc) \}  = \sqrt{n} \S_{n,\ASsc} + o_{\P}(1)$, as long as $\rn = o(\sqrt{n})$ and $\sum_{S \in \Ssc}\Ltwonorm{\baS} = O(1)$. Note that $\S_{n,\ASsc}$ is an average of the centered i.i.d. random variables $\{\xi_{\ASsc}(\bZ_i)\}_{i=1}^n$. Hence,  $\sqrt{n} \S_{n,\ASsc}$, when appropriately scaled to have unit variance, is expected to converge to a $\mathcal{N}(0,1)$ distribution under suitable Lyapunov-type moment conditions on $\xi_{\ASsc}(\bZ)$. We characterize this more explicitly through a stronger non-asymptotic statement \tcr{as follows.} 
}

\vspace{0.1in}
\emph{
Let $\sigmaxiAS^2 := \E[\{\xi_{\ASsc}(\bZ)\}^2] \equiv \Var\{ \xi_{\ASsc}(\bZ)\}$ and \tcr{$\rhoxiAS := [\E\{ | \xi_{\ASsc}(\bZ)|^3 \}]^{1/3}$, so that} $0 < \sigmaxiAS$ $\leq \rhoxiAS$, and assume that $\sigmaxiAS = \Omega(1)$ and $\rhoxiAS = O(1)$, so that $\rhoxiAS/\sigmaxiAS = O(1)$ \tcr{(verifications of these conditions are discussed in the \hyperref[supp_mat]{Supplement}; see \tcR{Appendix} 
\ref{sec:remarkverify}).} 
Finally, let $F_{\S_{n,\ASsc}}(x) := \P(\sqrt{n} \S_{n,\ASsc}/\sigmaxiAS \leq x) $ 
denote the cumulative distribution function (CDF) of $\sqrt{n} \S_{n,\ASsc}/\sigmaxiAS$, and let $\Phi(x) := \P(Z \leq x)$, 
where $Z \sim \mathcal{N}(0,1)$, denote the standard normal CDF, \tcr{for all $x \in \R$}. Then, the Berry-Esseen theorem \citep{Shevtsova11} implies:
\begin{equation*}
\underset{x \in \R}{\sup} \left| F_{\S_{n,\ASsc}}(x) - \Phi(x) \right| \; \leq \; \tcr{0.48\rhoxiAS^3/(\sqrt{n}\sigmaxiAS^3)} \; = \; O(n^{-1/2}). 
\vspace{-.05in}
\end{equation*}
\tcr{As a consequence, we also have:} $\tcr{\sqrt{n} \S_{n,\ASsc}/\sigmaxiAS}  \convd  \mathcal{N}(0,1)$ \tcr{as $n \rightarrow \infty$.} \qed
}

\vspace{0.05in}
\emph{
\tcr{Therefore,} as long as $\rn = o(\sqrt{n})$, $\sum_{S \in \Ssc}\Ltwonorm{\baS} = O(1)$, $\sigmaxiAS = \Omega(1)$ and $\rhoxiAS = O(1)$, we have:  $\sqrt{n} \{ \betahat(\ASsc) - \beta(\ASsc) \}/\sigmaxiAS  = \sqrt{n} \S_{n,\ASsc}/\sigmaxiAS + o_{\P}(1)$ and $\sqrt{n} \S_{n,\ASsc}/\sigmaxiAS \convd  \mathcal{N}(0,1)$. Invoking Slutsky's theorem, we \tcr{finally} have: $\sqrt{n} \{ \betahat(\ASsc) - \beta(\ASsc) \}/\sigmaxiAS \convd \mathcal{N}(0,1)$. Furthermore, for any consistent estimator $\widehat{\sigma}_{\xi_{\ASsc}}$ of $\sigmaxiAS$, it also holds, via another application of Slutsky's theorem, that $\sqrt{n} \{ \betahat(\ASsc) - \beta(\ASsc) \}/ \widehat{\sigma}_{\xi_{\ASsc}} \convd \mathcal{N}(0,1)$ \tcr{as $n \rightarrow \infty$.} \qed
}

\vspace{0.05in}
\emph{\tcr{Lastly, as mentioned above, verification of the moment conditions: $\rhoxiAS = O(1)$ and $\sigmaxiAS = \Omega(1)$ is discussed in \tcR{Appendix} 
\ref{sec:remarkverify} of the \hyperref[supp_mat]{Supplement}, where we provide fairly mild and general sufficient conditions for both to hold.}}

\end{remark}

\tcr{\tcm{As we conclude,} it is worth mentioning that some results on `uniform-in-model' bounds, similar in flavor to those presented in this section, 
were also obtained independently 
in the recent work of \citet{Kuchibhotla18} on post-selection inference in linear regression, although their results are not directly comparable to ours. 
\tcm{Their analysis is targeted towards more general settings, but their results are more involved and less tractable.} Our approach, on the other hand, is simpler and the results are more explicit and ready-to-use for application purposes. Further, our results automatically account for \tcm{any} data-dependent centering of $Y$ and $\bX$, unlike theirs, and the bounds \tcm{in} our main result (i.e. Theorem \ref{linreg:main_thm1}) are also more flexible in the sense that they directly involve the cardinality of $\Ssc$, as opposed to their bounds 
which generally aim at a worst-case analysis with $\Ssc$ assumed to include all subsets of $\bX$ having cardinality bounded by some $k \leq p$. Our bounds are therefore adaptive in $|\Ssc|$ and lead to sharper rates when $|\Ssc|$ is not too large (or at least not growing as fast as the worst case) which is often the case in practice.} 



\section{\tcr{Asymptotic Properties in High-Dimensional Settings: Inference for MIDA}}\label{section: asymptotic distribution}
Using the results from Section \ref{section: linear regression theory}, \tcr{we now develop inferential tools for MIDA (these also apply generally to any IDA-type estimator).} We \tcr{first} prove asymptotic linearity of MIDA 
under Assumptions \ref{assumption: LSEM}--\ref{assumption: eigenvalue condition} \tcR{in Section \ref{subsection:asymptotic-linearity} (Theorem \ref{theorem: asymptotic linearity})}, \tcr{and then} 
derive \tcr{its} asymptotic normality under some additional assumptions \tcR{in Section \ref{subsection:asymptotic-normality} (Corollaries \ref{corollary: asymptotic normality of IDA}--\ref{corollary: asymptotic distribution of MIDA} and Theorem \ref{theorem: asymptotic distribution})}, \tcr{followed by \tcR{construction of} confidence intervals, hypothesis tests \tcR{($p$-values)} etc. for the mediation effects \tcR{in Section \ref{subsection: inference}}.} 
We emphasize that the results derived below do \emph{not} depend on the nature of the CPDAG estimation \tcr{method} 
as long as it's consistent (Assumption \ref{assumption: structure learning consistency}).

\subsection{Asymptotic Linearity}\label{subsection:asymptotic-linearity}

%
%

For $j \in \{t+1,\ldots,p-1\}$, we define
\[
E_{n,jp}(\hat{\mathcal{C}}_{0}', \mathcal{C}_{0}') \; := \; \mathrm{aver}(\hat{\Theta}_{jp}(\hat{\mathcal{C}}_{0}')) - \mathrm{aver}(\hat{\Theta}_{jp}(\mathcal{C}_{0}')), \;\; \tcr{\mbox{where}}
\]
$\hat{\Theta}_{jp}(\mathcal{C}_{0}') := \{\hat{\beta}_{jp | \mathbf{Pa}_{\mathcal{G'}}(X_{j})\cup\mathbf{X}_{\{1,\ldots,t\}}} : \mathcal{G'} \in \MEC(\mathcal{C}_{0}')\}$. Since $E_{n,jp}(\hat{\mathcal{C}}_{0}', \mathcal{C}_{0}')  = 0$ whenever $\hat{\mathcal{C}}_{0}'$ and $\mathcal{C}_{0}'$ are identical, Assumption \ref{assumption: structure learning consistency} implies that for any $\epsilon > 0$ and for \emph{any} non-negative (\tcr{possibly} diverging) sequence $\{a_n\}$ (e.g.\tcR{,} $a_n = n$),
\[
\Prob\left(a_n \left|E_{n,jp}(\hat{\mathcal{C}}_{0}', \mathcal{C}_{0}') \right| > \epsilon \right) \; \leq \; \Prob(\hat{\mathcal{C}}_{0}' \neq \mathcal{C}_{0}') \; \longrightarrow \; 0 \quad \tcr{\mbox{as} \;\; n \rightarrow \infty,} 
\]
so that $E_{n,jp}(\hat{\mathcal{C}}_{0}', \mathcal{C}_{0}') = o_{\P}(a_n^{-1})$ \tcr{for any $t+1 \leq j \leq p-1$.}
This result allows us to use the uniform non-asymptotic theory developed in Section \ref{section: linear regression theory} for linear regression over \emph{non-random} subsets of high-dimensional covariates, since
\begin{align*}
& \mathrm{aver}(\hat{\Theta}_{jp}(\hat{\mathcal{C}}_{0}')) - \mathrm{aver}(\Theta_{jp}(\mathcal{C}_{0}')) \\
& \quad
= \;  E_{n,jp}(\hat{\mathcal{C}}_{0}', \mathcal{C}_{0}') ~+~ \left\{
\mathrm{aver}(\hat{\Theta}_{jp}(\mathcal{C}_{0}')) - \mathrm{aver}(\Theta_{jp}(\mathcal{C}_{0}'))\right\}.
\end{align*}

To present our results, we \tcr{first} define the residual of the linear regression of $X_{i}$ on $S \subseteq \{1,\ldots,p\} \setminus \{i\}$, \tcr{for any $1 \leq i \leq p$,} as
\[
R_{i|S} \; := \; X_{i} - \mu_i  - (\Sigma_{0})_{iS} \left[ (\Sigma_{0})_{SS} \right]^{-1} (\mathbf{X}_{S} - \boldsymbol\mu_S).
\]

Note that $(\Sigma_{0})_{iS} \left[ (\Sigma_{0})_{SS} \right]^{-1}$ is the vector of regression coefficients in the linear regression of $X_{i}$ on $\mathbf{X}_{S}$. For $j \in \{t+1,\ldots,p-1 \}$, we \tcr{now} define
\[
Z_{jp} \; := \; \frac{1}{L_{j}}\sum_{\ell=1}^{L_{j}} \mathbf{e}_{1,|S_{j\ell}|}^T \left((\Sigma_{0})_{S_{j\ell}S_{j\ell}} \right)^{-1}  (\mathbf{X}_{S_{j\ell}} - \boldsymbol\mu_{S_{j\ell}}) R_{p|S_{j\ell}}, \;\; \tcr{\mbox{where}}
\]
$\{\mathbf{X}_{S_{j1}},\ldots,\mathbf{X}_{S_{jL_{j}}}\}$ is the multi-set of vectors $\{ (X_{j}, \mathbf{X}_{\{1,\ldots,t\}}, \mathbf{Pa}_{\mathcal{G'}}(X_{j}))^T : \mathcal{G'} \in \MEC(\mathcal{C}_{0}') \}$ containing $L_{\mathrm{distinct}, j}$ distinct elements. Further, we define
\[
Z_{tj} ~:=~ \mathbf{e}_{1,t}^T \left((\Sigma_{0})_{\{1,\ldots,t\}{\{1,\ldots,t\}}} \right)^{-1}  (\mathbf{X}_{\{1,\ldots,t\}} - \boldsymbol\mu_{\{1,\ldots,t\}}) R_{j|\{1,\ldots,t\}}.
\]

For any random variable (or vector) $Z$, we will denote its $n$ i.i.d.\ copies by $Z^{(1)}, \ldots, Z^{(n)}$. Finally, we define \tcr{$q_{j} := |\mathbf{Adj}_{\mathcal{C}_{0}'}(X_{j})|$} \tcr{for $t+1 \leq j \leq p-1$}. Note that by Assumption \ref{assumption: sparsity}, we have $q_j \leq q = O(n^{1-b_1})~\text{for some $0 < b_1 \leq 1$}.$
\begin{theorem}[\tcR{Asymptotic linearity of $\mathrm{aver}(\hat{\Theta}_{jp}(\hat{\mathcal{C}}_{0}'))$ and $\hat{\eta}_j(\hat{\mathcal{C}}_0')$}]\label{theorem: asymptotic linearity}
Under Assumptions \ref{assumption: LSEM}\tcR{--}\ref{assumption: eigenvalue condition}, we have \tcr{for any $t < j < p$,}
\begin{align}\label{eq: asymptotic decomposition}
& \mathrm{aver}(\hat{\Theta}_{jp}(\hat{\mathcal{C}}_{0}')) - \mathrm{aver}(\Theta_{jp}(\mathcal{C}_{0}')) \nonumber \\
& \qquad \qquad = \; E_{n,jp}(\hat{\mathcal{C}}_{0}', \mathcal{C}_{0}') + \frac{1}{n}\sum_{r = 1}^{n} Z_{jp}^{(r)} + O_{\Prob}\left(\frac{q_{j} + \log(L_{\mathrm{distinct}, j})}{n}\right), \;\;\; \text{and} \nonumber \\
&\hat{\eta}_j(\hat{\mathcal{C}}_0') - \eta_j(\mathcal{C}_0') \;\; = \; \hat{\theta}_{tj}~E_{n,jp}(\hat{\mathcal{C}}_{0}', \mathcal{C}_{0}')  \nonumber \\
& \quad + \; \frac{1}{n}\sum_{r = 1}^{n} \left\{ \theta_{tj}Z_{jp}^{(r)} +  \mathrm{aver}(\Theta_{jp}(\mathcal{C}_{0}')) Z_{tj}^{(r)} \right\} \; + \; \left(\frac{1}{n} \sum_{r = 1}^{n}Z_{jp}^{(r)} \right) \left(\frac{1}{n} \sum_{r = 1}^{n}Z_{tj}^{(r)} \right) \nonumber \\
&  \quad + \; \hat{\theta}_{tj}~O_{\Prob}\left(\frac{q_{j} \; + \; \log(L_{\mathrm{distinct}, j})}{n}\right) \; + \; \mathrm{aver}(\hat{\Theta}_{jp}(\mathcal{C}_{0}'))~O_{\Prob}\left(\frac{1}{n}\right). \nonumber
\end{align}
\tcr{Note that all sums above are zero-mean i.i.d. sums as $\E(Z_{jp}) = \E(Z_{tj}) = 0$.}
\end{theorem}

\subsection{Asymptotic Normality}\label{subsection:asymptotic-normality}
In order to establish the asymptotic normality of the estimator of the total causal effect and \tcr{that} 
of the individual mediation effects, we impose the following stronger sparsity condition.

\begin{assumption}[Sparsity condition]\label{assumption: stronger sparsity}
Let $q_j$ and $L_{\mathrm{distinct, j}}$ be as above, \tcr{for any $t+1 \leq j \leq p-1$.} Then, we assume \tcr{that}
$$
n^{-1/2}\{q_{j}+ \log(L_{\mathrm{distinct}, j})\} \; \longrightarrow \; 0, \;\; \tcr{\mbox{as} \;\; n \rightarrow \infty.}
$$
\end{assumption}

We note that such stronger sparsity assumptions of a similar flavor \tcr{are} 
 frequently adopted in the literature, albeit for different but related problems, for deriving asymptotic normality results and confidence intervals in high-dimensional settings 
 \citep{Portnoy88,VandeGeer14,Javanmard14,Zhang14}. 
 To compare Assumption \ref{assumption: stronger sparsity} with the sparsity condition given
in Assumption \ref{assumption: sparsity}, note that $q = O(n^{1 - b_1})$ for some $1/2 < b_1 \leq 1$ implies that the Assumption \ref{assumption: stronger sparsity} holds, since \tcr{$\forall \; j$,}
\[
q_{j}+ \log(L_{\mathrm{distinct}, j}) \; \leq \; q_j + \log(2^{q_j}) \; \leq \; (1 + \log2)\underset{\tcr{t+1 \leq j \leq p-1}}{\max} q_j \; = \; (1 + \log2)q.
\]
Additionally, we make the following assumption.
\begin{assumption}[\tcr{Non-degenerate} conditional distributions]\label{assumption: conditional variance} The conditional variances satisfy the following lower bounds: for \tcr{any $t < j < p$,} 
\begin{enumerate}
\item $\Var(X_p \mid \mathbf{Adj}_{\mathcal{C}_0'}(X_j) \cup \mathbf{X}_{\{1,\ldots,t\}} \cup \{X_j\}) > v,~ \text{and}$
\item $\Exp\left[\Var(X_j \mid \mathbf{Adj}_{\mathcal{C}_0'}(X_j) \cup \mathbf{X}_{\{1,\ldots,t\}})\right] > v,$ \;\; \tcr{for some constant $v > 0$.}
\end{enumerate}
\end{assumption}

We note that Assumption \ref{assumption: conditional variance} resembles Assumption (F) of \cite{MaathuisKalischBuehlmann09}. Further, note that Assumption \ref{assumption: conditional variance} follows from Assumption \ref{assumption: eigenvalue condition} when the error variables are normally distributed. This is because $\mathbf{X}$ is generated from an LSEM with normally distributed error variables implies that the joint distribution of $\mathbf{X}$ is multivariate Gaussian with covariance matrix $\Sigma_0$. Hence, for any $1 \leq i \leq p$ and $S \subseteq \{1,\ldots, p\} \setminus \{i\}$ such that $|S| \leq q + t + 2$,
\[
\Var(X_i\mid \mathbf{X}_S) \; = \; ( \Sigma_0)_{ii} - (\Sigma_0)_{iS}(\Sigma_0)_{SS}^{-1}(\Sigma_0)_{Si} \; \geq \; \lambda_{\min}((\Sigma_0)_{(S,i)(S,i)}) \; \geq \; C_2,
\]
where the first inequality follows from the interlacing property of eigenvalues of a Hermitian matrix $A$ and eigenvalues of the Schur complement of any principal submatrix of $A$ (see, for example, Corollary 2.3 of \cite{Zhang05}), and the last inequality follows from Assumption \ref{assumption: eigenvalue condition}.

\begin{corollary}[\tcr{Asymptotic normality of the estimator(s) of the total causal effect(s)}]\label{corollary: asymptotic normality of IDA}
Under Assumptions \ref{assumption: LSEM}\tcR{--}\ref{assumption: eigenvalue condition}, \ref{assumption: stronger sparsity} and \ref{assumption: conditional variance}, we have 
\[
\frac{\sqrt{n}\left\{ \mathrm{aver}(\hat{\Theta}_{jp}(\hat{\mathcal{C}}_{0}')) - \mathrm{aver}(\Theta_{jp}(\mathcal{C}_{0}')) \right\}}{\sqrt{\Exp[Z_{jp}^2]}}~ \; \overset{d}{\longrightarrow} ~\mathcal{N}(0,1),  \;\; \tcr{(t < j < p).}
\]
\end{corollary}

\begin{remark}
\emph{Although we state Corollary \ref{corollary: asymptotic normality of IDA} for our particular LSEM setting where $B_{jk} = B_{pj} = 0$ for all $k \in \{1,\ldots,t\}$ and $j \in \{t+1,\ldots,p-1\}$, \tcr{we emphasize that} the same result \tcr{also continues to hold} for the original IDA estimator $\hat{\Theta}_{ik}(\hat{\mathcal{C}}_{0})  = \{\hat{\beta}_{ik | \mathbf{Pa}_{\mathcal{G}}(X_k)} : \mathcal{G} \in \MEC(\mathcal{C}_0) \},$ \tcr{corresponding to the full CPDAG $\mathcal{C}_0$, for any $1 \leq i,k \leq p$,} under the assumptions of Corollary \ref{corollary: asymptotic normality of IDA}.
}
\end{remark}

\begin{corollary}[\tcr{Asymptotic normality of the estimator(s) of the mediation effect(s)}]\label{corollary: asymptotic distribution of MIDA}
Under Assumptions \ref{assumption: LSEM}\tcR{--}\ref{assumption: eigenvalue condition}, \ref{assumption: stronger sparsity} and \ref{assumption: conditional variance}, we have 
\begin{align*}
 T_{n,jp} \; \tcr{:=} \; \frac{\sqrt{n}\left( \hat{\eta}_{j}(\hat{\mathcal{C}}_0') - \eta_{j}(\mathcal{C}_0') \right)}{\sqrt{\Exp[(\theta_{tj}Z_{jp} +  \mathrm{aver}(\Theta_{jp}(\mathcal{C}_{0}'))Z_{tj})^2]}}  ~~ \overset{d}{\longrightarrow} ~\mathcal{N}(0,1),    \;\; \tcr{(t < j < p),}
\end{align*}
provided at least one of $\theta_{tj}$ and $\mathrm{aver}(\Theta_{jp}(\mathcal{C}_{0}'))$ is \tcr{non-zero.}
\end{corollary}

To derive an asymptotic distribution of \tcr{$\hat{\eta}_{j}$} \tcr{even allowing for} 
both $\theta_{tj}$ and $\mathrm{aver}(\Theta_{jp}(\mathcal{C}_{0}'))$ \tcr{to be} 
zero, we consider 
\tcr{a \emph{modification}, $\tilde{T}_{n,jp}$,} of $T_{n,jp}$ \tcr{as follows}.
\begin{align}\label{eq: modified z-score}
\tilde{T}_{n,jp}  \; := \frac{\sqrt{n}\left( \hat{\eta}_{j}(\hat{\mathcal{C}}_0') - \eta_{j}(\mathcal{C}_0') \right)}{\sqrt{\hat{\theta}_{tj}^2\Exp[Z_{jp}^2] + \mathrm{aver}(\hat{\Theta}_{jp}(\mathcal{C}_{0}'))^2 \Exp[Z_{tj}^2] + 2\hat{\theta}_{tj} \mathrm{aver}(\hat{\Theta}_{jp}(\mathcal{C}_{0}')) \Exp[Z_{jp}Z_{tj}]}}.
\end{align}
\tcr{In contrast to $T_{n,jp}$, the denominator of the modified version $\tilde{T}_{n,jp}$ involves only the estimators $\hat{\theta}_{tj}$ and $\mathrm{aver}(\hat{\Theta}_{jp}(\mathcal{C}_{0}'))$ which are expected to be non-zero (a.s.), thereby ensuring that $\tilde{T}_{n,jp}$ is well-defined regardless of whether or not the corresponding true parameters $\theta_{tj}$ and $\mathrm{aver}(\Theta_{jp}(\mathcal{C}_{0}'))$ are both zero.} \tcR{The result below gives its full asymptotic characterization.}

\begin{theorem}[\tcR{Asymptotic limit of $\tilde{T}_{n,jp}$}]\label{theorem: asymptotic distribution}
Let $\tilde{T}_{n,jp}$ be as in \eqref{eq: modified z-score} and let $\rho = \frac{\Exp[Z_{tj}Z_{jp}]}{\sqrt{\Exp[Z_{tj}^2]\Exp[Z_{jp}^2]}}$ be the correlation coefficient between $Z_{tj}$ and $Z_{jp}$. Further, let $ (W_1, W_2)^T$ be a random vector that has a zero-mean bivariate Gaussian distribution with the covariance matrix $ \bigg(\begin{array}{cc} 1 & \rho\\ \rho & 1 \end{array} \bigg)$. Then, under Assumptions \ref{assumption: LSEM} - \ref{assumption: eigenvalue condition}, \ref{assumption: stronger sparsity} and \ref{assumption: conditional variance}, we have \tcr{for any $t < j < p$,}
\begin{align*}
\tilde{T}_{n,jp} ~~ \overset{d}{\longrightarrow} \;
\left\{
\begin{array}{cc}
W(\rho) \; := \; \frac{W_1W_2}{\sqrt{W_1^2 + W_2^2 + 2 \rho W_1 W_2}} & \text{if \; $\theta_{tj}  = \mathrm{aver}(\Theta_{jp}(\mathcal{C}_{0}')) =  0$,} \\
\mathcal{N}(0,1) & \text{otherwise.}
\end{array}
\right.
\end{align*}
\end{theorem}
\tcr{Theorem \ref{theorem: asymptotic distribution} therefore provides a unified result on the asymptotic distribution of $\tilde{T}_{n,jp}$, accounting for all possible cases regarding the true values of $\theta_{tj}$ and $\mathrm{aver}(\Theta_{jp}(\mathcal{C}_{0}'))$. To our knowledge, results of this flavor, that also allow for both the parameters to be zero, are generally rare in the relevant literature on inference for product-type mediation effect parameters. We provide further discussions on Theorem \ref{theorem: asymptotic distribution}, and its implications \tcR{and uses} in inference, in Section \ref{subsection: inference} \tcR{next} (see, in particular, the last two paragraphs and the discussion involving Figure \ref{fig: asymp dist} therein), where we also discuss how it can be used to obtain valid confidence intervals and $p$-values for testing the mediation effects $\eta_j(\mathcal{C}_0')$ via a unified approach,
\emph{without} requiring any knowledge of which case we are under, 
i.e. whether or not $\theta_{tj}$ and $\mathrm{aver}(\Theta_{jp}(\mathcal{C}_{0}'))$ are both truly zero.
}


\subsection{Confidence Interval and Hypothesis Testing}\label{subsection: inference}
To construct a confidence interval \tcr{(CI) for $\eta_j(\mathcal{C}_0')$} based on $\tilde{T}_{n,jp}$, we need to estimate the denominator in the right hand side of \eqref{eq: modified z-score}. To this end, we define \tcr{for $t < j < p$,}
\begin{align*}
\hat{R}_{i|S} &\; := \; X_{i} - \hat{\mu}_i  - (\hat{\Sigma}_{0})_{iS} \left[ (\hat{\Sigma}_{0})_{SS} \right]^{-1} (\mathbf{X}_{S} - \hat{\boldsymbol\mu}_S) \quad\quad \tcr{(1 \leq i \leq p)}, \\
\hat{Z}_{tj} &\; := \; \mathbf{e}_{1,t}^T \left[(\hat{\Sigma}_{0})_{\{1,\ldots,t\}{\{1,\ldots,t\}}} \right]^{-1}  (\mathbf{X}_{\{1,\ldots,t\}} - \hat{\boldsymbol\mu}_{\{1,\ldots,t\}}) \hat{R}_{j|\{1,\ldots,t\}}, \;\;\; \text{and} \\
\hat{Z}_{jp} &\; := \; \frac{1}{\hat{L}_{j}}\sum_{\ell=1}^{\hat{L}_{j}} \mathbf{e}_{1,|\hat{S}_{j\ell}|}^T \left[(\Sigma_{0})_{\hat{S}_{j\ell}\hat{S}_{j\ell}} \right]^{-1}  (\mathbf{X}_{\hat{S}_{j\ell}} - \hat{\boldsymbol\mu}_{\hat{S}_{j\ell}}) \hat{R}_{p|\hat{S}_{j\ell}}, \;\;\; \tcr{\mbox{where}}
\end{align*}
$\hat{\boldsymbol\mu}$ is the sample mean \tcr{and} $\hat{\Sigma}_{0}$ is the sample covariance matrix \tcr{of $\mathbf{X}$}, and \tcr{$\{\mathbf{X}_{\hat{S}_{1}},$ $\ldots,\mathbf{X}_{\hat{S}_{\hat{L}_{j}}}\hspace{-0.05in}\}$} is the multi-set of vectors $\{ (X_{j}, X_{1}, \mathbf{Pa}_{\mathcal{G'}}(X_{j}))^T : \mathcal{G'} \in \MEC(\hat{\mathcal{C}}_{0}')\}$.

\par\smallskip
Finally, we define \tcr{for any $t+1 \leq j \leq p-1$,}
\begin{align}\label{eq: test statistic}
\hat{T}_{n,jp}(\eta_j(\mathcal{C}_{0}')) \; := \; \frac{\sqrt{n}\left( \hat{\eta}_{j}(\hat{\mathcal{C}}_0') - \eta_{j}(\mathcal{C}_0') \right)}{\sqrt{\frac{1}{n} \sum_{r = 1}^{n} \left\{ \mathrm{aver}(\hat{\Theta}_{jp}(\hat{\mathcal{C}_0'}))~\hat{Z}_{tj}^{(r)} + \hat{\theta}_{tj}~ \hat{Z}_{jp}^{(r)} \right\}^2}}.
\end{align}

\tcr{Owing to} the consistency of $\hat{\mathcal{C}}_0'$, \tcr{and that of any} $q \times q$ sub-matrix of $\hat{\Sigma}_{0}$ (for $q$ as in Assumption \ref{assumption: sparsity}), \tcr{it follows that} $\hat{T}_{n,jp}$ and $T_{n,jp}$ have the same asymptotic distribution. 
\tcR{Therefore,} \tcR{using Corollary \ref{corollary: asymptotic distribution of MIDA} (or Theorem \ref{theorem: asymptotic distribution})}, when at least one of $\theta_{tj}$ and $\mathrm{aver}(\Theta_{jp}(\mathcal{C}_{0}'))$ is \tcr{non-zero,} an asymptotically correct $100 \times(1-\alpha)\%$ \tcr{CI} 
for $\eta_j(\mathcal{C}_0')$ is given by
\begin{align}\label{eq: confidence interval}
\hat{\eta}_j(\hat{\mathcal{C}}_0') \pm \frac{\Phi^{-1}(1 - \alpha / 2)}{\sqrt{n}} \sqrt{\frac{1}{n} \sum_{r = 1}^{n} \left\{ \mathrm{aver}(\hat{\Theta}_{jp}(\hat{\mathcal{C}_0'}))~\hat{Z}_{tj}^{(r)} + \hat{\theta}_{tj}~ \hat{Z}_{jp}^{(r)} \right\}^2},
\end{align}
where $\Phi(\cdot)$ \tcr{denotes} the CDF of a standard \tcr{normal} 
distribution.

\tcR{When} both $\theta_{tj}$ and $\mathrm{aver}(\Theta_{jp}(\mathcal{C}_{0}'))$ are zero, 
\tcr{the CI} \tcR{ideally} needs to be computed using the probability distribution of $W(\rho)$ (\tcr{as} defined in Theorem \ref{theorem: asymptotic distribution}). \tcr{However,} Figure \ref{fig: asymp dist} shows that the distribution of $W(\rho)$ is much more concentrated around zero compared to the standard normal distribution. This suggests that a \emph{unified} approach for constructing \tcr{the CI} 
is to use \eqref{eq: confidence interval} for \emph{all} values of $\mathrm{aver}(\Theta_{jp}(\mathcal{C}_{0}'))$ and $\theta_{tj}$. The resulting \tcr{CI} 
\tcr{would be conservative} \tcr{(but nonetheless \emph{still} valid)} when both $\theta_{tj}$ and $\mathrm{aver}(\Theta_{jp}(\mathcal{C}_{0}'))$ are zero, 
\tcr{or are} very close to zero.


\par\smallskip
Again, when at least one of $\theta_{tj}$ and $\mathrm{aver}(\Theta_{jp}(\mathcal{C}_{0}'))$ is \tcr{non-zero}, then \tcr{owing to Theorem \ref{theorem: asymptotic distribution},} the \tcr{$p$-value} for testing the null hypothesis, $H_0: \; \eta_j(\mathcal{C}_{0}') = 0$, against the \tcr{two-sided} alternative, $H_1: \; \eta_j(\mathcal{C}_{0}') \neq 0$, can be \tcr{computed} as:
\[
2\left[ 1 - \Phi\left( |\hat{T}_{n,jp}(\eta_j(\mathcal{C}_{0}'))|\right) \right], \quad \tcr{(t+1 ~\leq~ j ~\leq~ p-1).}
\]
When both $\theta_{tj}$ and $\mathrm{aver}(\Theta_{jp}(\mathcal{C}_{0}'))$ are zero, the \tcr{$p$-value} corresponding to the two-sided test \tcR{ideally} should be $\Prob(|W(\rho)| >  |\hat{T}_{n,jp}(\eta_j(\mathcal{C}_{0}'))|)$. In the \tcr{next} 
section \tcr{(on simulation studies)}, we will use the \tcr{$p$-values} to rank the mediator\tcr{s} according to their individual mediation effect\tcr{s}. For this purpose, \tcR{however,} we propose to use the Gaussian \tcr{$p$-values} in all cases. We justify this choice in the second part of Figure \ref{fig: asymp dist} which shows that the CDF of $|W(\rho)|$ 
\tcr{dominates} the CDF of the absolute value of a standard Gaussian random variable (\tcr{i.e.} 
$2\Phi(x) - 1$, for \tcr{any} $x > 0$). This implies conservative (that is higher than expected), \tcR{but nonetheless \emph{still} valid,} \tcr{$p$-values} when both $\theta_{tj}$ and $\mathrm{aver}(\Theta_{jp}(\mathcal{C}_{0}'))$ are zero, 
\tcr{or are} very close to zero.

\begin{figure}[!ht]
\centering
\includegraphics[width = \textwidth]{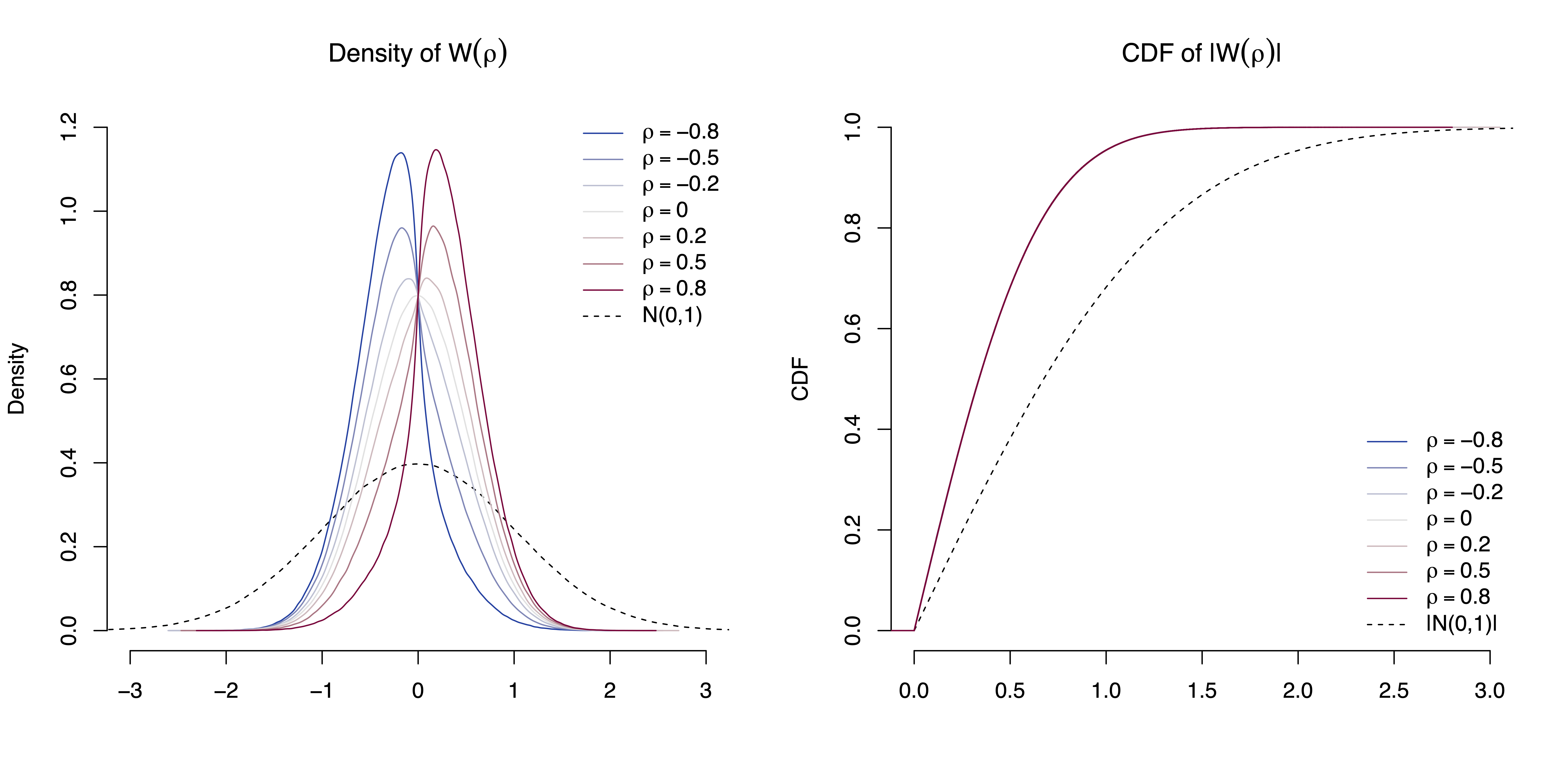}
\caption{The left subfigure compares the probability density function of \tcR{the random variable} $W(\rho)$ \tcR{(as in Theorem \ref{theorem: asymptotic distribution})} with \tcR{that of} the standard Gaussian density (dashed line). The right subfigure compares 
\tcR{the cumulative} distribution function (CDF) of $|W(\rho)|$ with the CDF of the absolute value of a standard Gaussian random variable (dashed line).}
\label{fig: asymp dist}
\vspace{-0.12in}  
\end{figure}


\section{Simulations}\label{section: simulations}
\tcr{We conducted extensive simulation studies to examine the performance of our proposed MIDA estimator and associated inferential tools under various settings, which we describe below \tcR{in Section \ref{subsection: simulation settings}}, followed by \tcR{presenting} the results \tcR{in Section \ref{results}}.}

\subsection{Simulation Settings}\label{subsection: simulation settings}



For each of the three \tcR{basic} settings given in Table \ref{table: simulation settings}, we use the \texttt{R}-package \textbf{pcalg} \citep{KalischEtAl12} to simulate $m$ random weighted DAGs $\{\mathcal{G}^{\prime(1)},\ldots,\mathcal{G}^{\prime(m)}\}$ with $p-2$ vertices $\{X_2, \ldots, X_{p-1}\}$ and $(p-2)d/2$ edges on average, where each pair of nodes in a randomly generated DAG has the probability $d/\binom{p-2}{2}$ of being adjacent (implying that the expected degree of each node is $d$). From each DAG $\mathcal{G}^{\prime(r)}$, we obtain the DAG $\mathcal{G}^{(r)},$ with $p$ vertices $\{X_1, X_2 \ldots, X_{p-1}, X_{p}\}$ by randomly adding directed edges from $X_1$ to $X_j$ with probability $0.2$ for $j = 2,\ldots, p$ and from $X_j$ to $X_{p}$ with probability $0.1$ for $j = 1,\ldots, p-1$. The edge weights are drawn independently from a uniform distribution on $[-1, -0.5]\cup[0.5, 1]$. Here we assume the absence of \tcr{any} pre-treatment covariates, 
\tcr{i.e. we set} $t=1$.


\begin{table}[ht]
\centering
\caption{Simulation settings\tcr{. Here,} \tcr{$p$} denotes the number of potential mediators \tcr{(plus two more \tcR{to account} for \tcR{the} treatment and \tcR{the} response)}, $d$ denotes the average degree in the causal DAG on the potential mediators, and $m$ denotes the number of randomly generated DAGs for a pair $(p,d)$.}
\label{table: simulation settings}
\begin{tabular}{c|rlr}
Setting &  $p$ & $d$ & $m$ \\ \hline
1 &  252 & 3 & 20 \\
2 &  502 & 3.5 & 10 \\
3 &   1002 & 4& 5\\
\hline
  \end{tabular}
\end{table}

Let $B_{\mathcal{G}^{(r)}}$ denote the weight matrix of the weighted DAG $\mathcal{G}^{(r)}$, i.e. $(B_{\mathcal{G}^{(r)}})_{ij}$ $\neq 0$ if and only if the edge $X_i \to X_j$ is present in $\mathcal{G}^{(r)}$ and it then equals the corresponding edge weight. For \tcr{$r = 1,\ldots, m$}, the weight matrix $B_{\mathcal{G}^{(r)}}$ and a random vector ${\boldsymbol\epsilon}^{(r)} = (\epsilon_{1}^{(r)},\ldots,\epsilon_{p}^{(r)})^T$ define a distribution on $\mathbf{X}^{(r)} = (X_{1}^{(r)},\ldots,X_{p}^{(r)})^T$ via the linear structural equation model
\[
\mathbf{X}^{(r)} = B_{\mathcal{G}^{(r)}}^T \mathbf{X}^{(r)} + {\boldsymbol\epsilon}^{(r)}.
\vspace{-0.07in}
\]
We choose $\epsilon_{1}^{(r)},\ldots,\epsilon_{p}^{(r)}$ to be \tcr{zero-mean} Gaussian random variables with variances independently drawn from a $\mathrm{Uniform}[0.5, 1]$ distribution. Finally, we standardize \tcr{all} the variables to have: $\Var(X_{i}^{(r)}) = 1$, for all $i = 1,\ldots,p$.

\hspace{-0.07in} For each setting, we generate $200$ random samples $\{\mathcal{D}_{n1}^{(r)}, \hdots, \mathcal{D}_{n200}^{(r)}\}$ of size $n \in \{500, 1000$, \tcR{$5000\}$} from the joint distribution of $\mathbf{X}^{(r)}$ for $r = 1,\ldots,m$.  For each $r \in \{1,\ldots, m\}$, we compute estimates and the corresponding \tcr{$p$-values} for $\eta_j^{(r)}$ based on the data $\mathcal{D}_{nk}^{(r)}$ using Algorithm \ref{algorithm: CPDAG estimation} when the graph used in line \ref{line: estimator} is \tcr{either:} (i) an estimated CPDAG obtained by applying the ARGES algorithm \citep{NandyHauserMaathuis16}, \tcr{or} (ii) the true CPDAG $C_{0}^{\prime (r)}$, \tcr{or} (iii) the true DAG $\mathcal{G}^{\prime (r)}$, or (iv) the empty graph. Note that the last case \tcr{(empty graph)} corresponds to a \emph{naive method} that assumes that the potential mediators are conditionally independent given the treatment variable.


\subsection{Results}\label{results}
As a finite sample validation of the asymptotic results obtained in Section \ref{section: asymptotic distribution}, we record (from each replication) whether the true parameter value $\eta_j^{(r)}$ lies within the $95\%$ standard Gaussian confidence interval given by \eqref{eq: confidence interval} with $\alpha = 0.05$. To present the result\tcR{s}, we split $\{\eta_j^{(r)} : r \in \{1,\ldots, m\},~ j \in \{2,\ldots, \tcr{p-1}\} \}$ 
into three equally sized groups according to the quantiles of $|\max(\theta_{tj}^{(r)}, \mathrm{aver}(\Theta_{jp}^{(r)}))|$'s distribution, and report the \tcR{median} 
empirical coverage probabilities in each group. 
Table \ref{table: coverage} shows that the\tcR{se} $95\%$ asymptotic confidence intervals exhibit an extremely high coverage in the first group 
where most of $|\max(\theta_{tj}^{(r)}, \mathrm{aver}(\Theta_{jp}^{(r)}))|$ equal zero, as well as in the second group 
where most of $|\max(\theta_{tj}^{(r)}, \mathrm{aver}(\Theta_{jp}^{(r)}))|$ are very close to zero. This is due to the fact that the correct asymptotic distribution of $\hat{T}_{n, jp}^{(r)}$ (see \tcR{Theorem \ref{theorem: asymptotic distribution} and relevant discussions in} Section \ref{subsection: inference}) in these cases is much more concentrated around zero than the distribution of a standard Gaussian random variable. The third group 
with reasonably high values of $|\max(\theta_{tj}^{(r)}, \mathrm{aver}(\Theta_{jp}^{(r)}))|$ exhibits the correct coverage when the CPDAG is known, but we do see some \tcR{minor} loss of coverage \tcR{(and/or little higher standard errors)}  when the CPDAG is estimated due to finite sample graph estimation errors.

\begin{table}[!ht]
\footnotesize
\centering
\caption{Median empirical coverage probabilities \tcR{(coverage)} and average lengths \tcr{(size)} of the 95\% confidence intervals when the CPDAG is known and when the CPDAG is estimated. For each \tcr{of the} simulation settings, the results are divided into three groups (L $\leftarrow$ Low, M $\leftarrow$ Medium and H $\leftarrow$ High) of roughly equal sizes based on $|\max(\theta_{tj}^{(r)}, \mathrm{aver}(\Theta_{jp}^{(r)}))|$\tcr{'s values}. The numbers in the brackets \tcr{denote} the corresponding standard deviations \tcR{of the coverages}.}
\label{table: coverage}
\begin{tabular}{ c @{\hspace{0.1cm}}| r @{\hspace{0.2cm}} r  @{\hspace{0.3cm}} c @{\hspace{0.3cm}}  r @{\hspace{0.3cm}} c @{\hspace{0.3cm}} l}
\toprule
&\multirow{3}{*}{$p~$} & \multirow{3}{*}{$n$}  & \multicolumn{2}{c}{Known CPDAG} &  \multicolumn{2}{c}{Estimated CPDAG}\\ 
\cmidrule(lr){4-5} \cmidrule(lr){6-7}
& &  &\hspace{0.1in}Coverage &Size    &\hspace{0.1in}Coverage  &Size\\
  \hline
&     & 500   & \hspace{0.1in}100 (0.00)  & 0.01   & \hspace{0.1in}100 (0.00)  & 0.01 \\
&252  & 1000  & \hspace{0.1in}100 (0.00)  & 0.01   & \hspace{0.1in}100 (0.00)  & 0.01 \\
&     & 5000  & \hspace{0.1in}100 (0.00)  & 0.00   & \hspace{0.1in}100 (0.00)  & 0.00 \\  \cmidrule(lr){2-7}
&     & 500   & \hspace{0.1in}100 (0.00)  & 0.01   & \hspace{0.1in}100 (0.00)  & 0.01 \\
&502  & 1000  & \hspace{0.1in}100 (0.00)  & 0.01   & \hspace{0.1in}100 (0.00)  & 0.01 \\
\rotatebox{90}{\rlap{Group L}}
&     & 5000  & \hspace{0.1in}100 (0.00)  & 0.00   & \hspace{0.1in}100 (0.00) & 0.00 \\ \cmidrule(lr){2-7}
&     & 500   & \hspace{0.1in}100 (0.00)  & 0.01   & \hspace{0.1in}100 (0.00) & 0.01 \\
&1002 & 1000  & \hspace{0.1in}100 (0.00)  & 0.01   & \hspace{0.1in}100 (0.00) & 0.01 \\
&     & 5000  & \hspace{0.1in}100 (0.00)  & 0.00   & \hspace{0.1in}100 (0.00) & 0.00 \\  \bottomrule
&     & 500   & \hspace{0.1in}99.0 (1.48) & 0.04   & \hspace{0.1in}99.0 (1.48) & 0.04 \\
&252  & 1000  & \hspace{0.1in}98.0 (2.22) & 0.03   & \hspace{0.1in}98.0 (2.97) & 0.03 \\
&     & 5000  & \hspace{0.1in}96.0 (2.22) & 0.01   & \hspace{0.1in}96.0 (2.97) & 0.01 \\ \cmidrule(lr){2-7}
&     & 500   & \hspace{0.1in}99.5 (0.74) & 0.04   & \hspace{0.1in}99.5 (0.74) & 0.04 \\
&502  & 1000  & \hspace{0.1in}98.5 (2.22) & 0.03   & \hspace{0.1in}98.5 (2.22) & 0.03 \\
\rotatebox{90}{\rlap{Group M}}
&     & 5000  & \hspace{0.1in}96.5 (2.22) & 0.01   & \hspace{0.1in}96.5 (2.97) & 0.01 \\ \cmidrule(lr){2-7}
&     & 500   & \hspace{0.1in}98.5 (2.22) & 0.05   & \hspace{0.1in}98.5 (2.22) & 0.05 \\
&1002 & 1000  & \hspace{0.1in}97.5 (2.97) & 0.04   & \hspace{0.1in}97.5 (2.97) & 0.04 \\
&     & 5000  & \hspace{0.1in}96.0 (2.22) & 0.02   & \hspace{0.1in}96.0 (2.22) & 0.02 \\   \bottomrule
&     & 500   & \hspace{0.1in}95.0 (1.48) & 0.14   & \hspace{0.1in}94.0 (2.97) & 0.14 \\
&252  & 1000  & \hspace{0.1in}95.0 (1.48) & 0.10   & \hspace{0.1in}94.0 (2.97) & 0.10 \\
&     & 5000  & \hspace{0.1in}95.0 (1.48) & 0.04   & \hspace{0.1in}94.0 (2.97) & 0.04 \\ \cmidrule(lr){2-7}
&     & 500   & \hspace{0.1in}95.5 (1.48) & 0.14   & \hspace{0.1in}94.5 (2.22) & 0.14 \\
&502  & 1000  & \hspace{0.1in}95.0 (1.48) & 0.10   & \hspace{0.1in}94.5 (2.22) & 0.10 \\
\rotatebox{90}{\rlap{Group H}}
&     & 5000  & \hspace{0.1in}95.0 (1.48) & 0.04   & \hspace{0.1in}94.5 (2.22) & 0.04 \\ \cmidrule(lr){2-7}
&     & 500   & \hspace{0.1in}95.0 (1.48) & 0.15   & \hspace{0.1in}94.5 (2.22) & 0.15 \\
&1002 & 1000  & \hspace{0.1in}95.0 (1.48) & 0.11   & \hspace{0.1in}94.5 (2.22) & 0.11 \\
&     & 5000  & \hspace{0.1in}95.0 (1.48) & 0.05   & \hspace{0.1in}94.5 (2.22) & 0.05 \\
\bottomrule
\end{tabular}

\end{table}

Next, we investigate the effect of graph estimation error in identifying the set of true mediators: $S^{(r)} := \{ X_j^{(r)} : \eta_j^{(r)} \neq 0,~ j = 2, \ldots, p-1 \}$. Figure \ref{fig: PR curves} shows the averaged (over 200 iterations) Precision-Recall curves for estimating the target set $\cup_{r=1}^{m} S^{(r)}$ based on (i) the ranking of the absolute values of estimates of $\cup_{r=1}^{m} \{\eta_1^{(r)}, \ldots,\eta_p^{(r)}\}$ and (ii) the ranking of the corresponding \tcr{$p$-values} (in the reverse order). As we would expect, the methods based on the true DAG performs the best. Although it is unrealistic to assume that the true graph is known, we include it in our results to gain insight into the loss due to estimating the true CPDAG instead of the true DAG. We note that the methods based on the estimated CPDAG and based on the true CPDAG perform equally well, and they outperform the naive method based on the empty graph. Finally, Figure \ref{fig: PR curves} \tcR{also} demonstrates that we can achieve substantial performance gain by using \tcr{$p$-values} instead of the raw estimates.
\tcR{This further exemplifies the importance and benefit of our inferential tools developed in Section \ref{section: asymptotic distribution} for IDA-type estimators that enables one to obtain these $p$-values in the first place.} 


\begin{figure}[!ht]
\centering
\includegraphics[width = \textwidth]{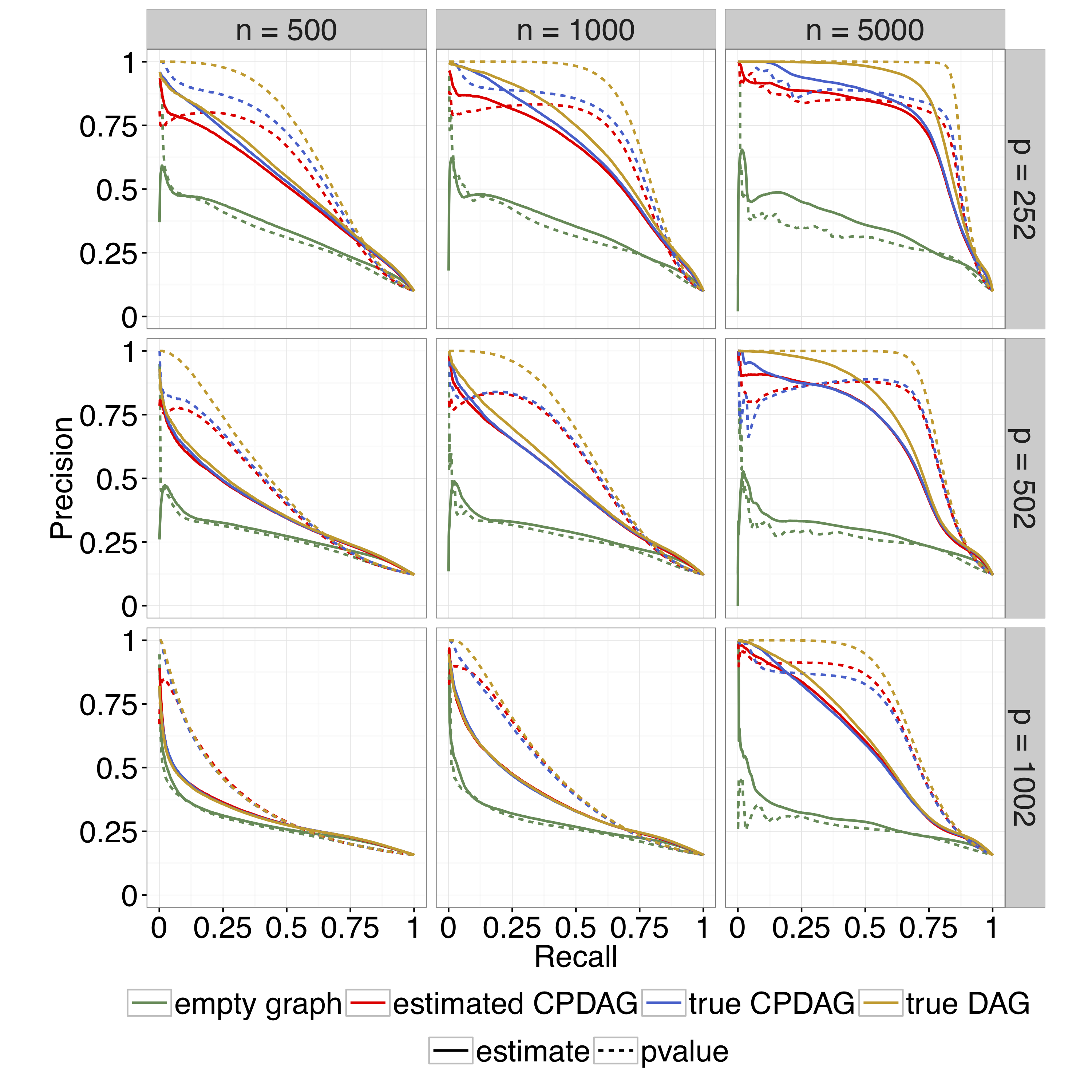}
\caption{Precision-Recall curves for the target set with top $k$ potential mediators according to: (i) large $|\hat{\eta}_j^{(r)}|$ values (estimate) (solid line), and (ii) small corresponding \tcr{$p$-values} (dashed line), for $k = 1,\hdots, \tcr{(p-2) \times m}$. Note that $\tcr{(p -2) \times m} = 5000$ for each simulation setting.}
\label{fig: PR curves}
\vspace{-0.15in}
\end{figure}


The harmonic mean of precision and recall is known as F-score, and it is a popular way \tcR{of} combining precision and recall into a single performance measure that ranges between 0 and 1. By adopting this notion of a performance measure, we aim to choose a set of top mediators that maximizes \tcr{the} F-score.
\tcR{We achieve this in practice through a heuristic \emph{$p$-value based thresholding} approach; the results are presented in Table \ref{table: F-score}.} 
Table \ref{table: F-score} demonstrates that we can achieve a nearly optimal F-score (i.e. the best achievable F-score \tcr{for} our method on a given dataset) by \tcr{a} thresholding of \tcr{$p$-values} at \tcr{a} level \tcR{of} $0.01$ for $n = 5000$, and at \tcr{a} level \tcR{of} $0.1$ for $n \in \{ 500, 1000\}$. We acknowledge that this heuristic $p$-value thresholding technique does not possess any theoretical justification, but it seems to work surprisingly well for estimating the target set $\cup_{r=1}^{m} S^{(r)}$ in our simulation settings. We also provide discussions on some alternative approaches \tcR{(including a Benjamini-Hochberg false discovery rate control procedure and its modifications)} in Section \ref{section: discussion}
and in \tcr{ \tcR{Appendix} 
\ref{sec:FDR_Control} of the \hyperref[supp_mat]{Supplement}.} 

\begin{table}[!ht]
\centering
\caption{Averaged (over 200 iterations) recall, precision and F-score for estimating the target set $\cup_{r=1}^{m}S^{(r)}$ based on \tcr{$p$-value} thresholding, where we used estimated CPDAGs for computing the value of our IDA-based estimators and the corresponding \tcr{$p$-values} \tcr{for thresholding}. The numbers in the brackets \tcr{denote} the corresponding standard deviations.}
\label{table: F-score}
\resizebox{\textwidth}{!}{
\begin{tabular}{r | r | ll | ll | ll}
  \hline
  $p$ &  $n$ & Target & Estimated & Recall & Precision & Achieved & Optimal \\
&   & size & size &  &  & F-score & F-score \\
  \hline
 & 500 &  & 568.0 (18.1) & 0.61 (0.02) & 0.54 (0.02) & 0.57 (0.01) & 0.59 (0.01)   \\
252   & 1000 & 503  & 673.6 (19.3) & 0.74 (0.02) & 0.55 (0.01) & 0.63 (0.01) & 0.69 (0.01)   \\
   & 5000 &  & 560.7 (11.4) & 0.82 (0.01) & 0.73 (0.01) & 0.77 (0.01) & 0.79 (0.01)   \\ \hline
 & 500 &  & 445.8 (17.2) & 0.38 (0.02) & 0.52 (0.02) & 0.44 (0.02) & 0.46 (0.01)   \\
  502   & 1000 & 613  & 569.6 (17.8) & 0.54 (0.02) & 0.58 (0.02) & 0.56 (0.01) & 0.56 (0.01)   \\
   & 5000 &  & 505.8 (10.5) & 0.68 (0.01) & 0.83 (0.01) & 0.75 (0.01) & 0.75 (0.01)   \\ \hline
 & 500 &  & 390.7 (17.5) & 0.24 (0.01) & 0.48 (0.02) & 0.32 (0.01) & 0.38 (0.01)   \\
  1002   & 1000 & 788  & 502.6 (18.1) & 0.35 (0.01) & 0.55 (0.02) & 0.43 (0.01) & 0.45 (0.01)   \\
   & 5000 &  & 455.0 (14.2) & 0.5 (0.01) & 0.87 (0.02) & 0.63 (0.01) & 0.67 (0.01)   \\
   \hline
\end{tabular}
}

\end{table}

\section{Application}\label{section: application}
\hspace{-0.02in}We demonstrate the applicability of \tcr{our} MIDA estimator in real data using a data set collected on 104 yeast segregants created by crossing of two genetically diverse strains, BY and RM \citep{yeast2}. The data set includes the growth yields of each segregant grown in the presence of different chemicals or small molecule drugs \citep{yeastdrug}.
These segregants have different genotypes that contribute to rich phenotypic diversity.
One key question is to understand how genetic variants contribute to the phenotypic variability. One possible path is through regulation of gene expression variations. Besides genotype data,  6189 yeast genes are profiled in rich media and in the absence of any chemical or drug using expression arrays  \citep{yeast2}.  We use the same data preprocessing steps as \cite{yeast1} to create a list of candidate gene expression features based on their potential regulatory effects, including transcription factors, signalling molecules, chromatin factors and RNA factors and  genes involved in vacuolar transport, endosome, endosome transport and vesicle-mediated transport.  We further filter out genes with standard deviation (s.d.) $ \le 0.2$ in expression level, resulting in a total of  813 genes in our analysis.

We are interested in identifying \tcr{the} genes whose expression levels  mediate the effect of genetic variants on  yeast growth yield after being treated with hydrogen peroxide. In particular, the genetic variant M2$\_$477206$\_$486640 is highly associated with the yeast growth yield  (\tcr{$p$-value =} 0.00032). Our goal is to identify the gene expressions that mediate the effect of this genetic variant. At a nominal $p$-value of 0.05,  MIDA identified six genes that may mediate the effects of the genetic variant M2$\_$477206$\_$486640 on yeast growth (see Table \ref{yeast.tbl}). Due to relatively small sample sizes, these genes are not significant after we adjust for multiple comparisons. \tcr{\tcm{However, a}lthough we cannot claim any statistically significant result \tcm{here,} we demonstrated how MIDA can be used to prioritize future biological experiments by
identifying a set of candidate genes. The candidate selection is performed by thresholding the \tcr{$p$-values,} as this turned out to be most effective method in our simulation study.}

Interestingly, the estimated signs of the mediation effects in Table \ref{yeast.tbl} agree with known biology. Among these genes, over-expression of \tcr{DBP8} \tcr{\citep{yeast-dbp1}} leads to vegetative and decreased rate of growth. In contrast, lower expression of the GPA1 gene \tcr{typically} leads to decreased resistance to chemicals and decreased sporulation efficiency \tcr{\citep{yeast-gpa2}}. 

\begin{table}[ht]
\centering
\caption{Analysis of yeast growth yield in the presence of hydrogen peroxide. The gene expression mediators \tcr{for the} genetic variant M2$\_$477206$\_$48664 \tcr{that were} identified  by MIDA with a nominal \tcr{$p$-value $<$ 0.05}. \tcr{The gene} names, their estimated mediation effects and \tcr{the corresponding} 95\% confidence intervals are presented below.}\label{yeast.tbl}

\begin{tabular}{llrrr}
  \hline
Gene ID & Gene name & Estimated effect & 95\% confidence interval & $p$-value \\
  \hline
YNR047W & YNR047W & -0.0532 (0.0233) & (-0.0988, -0.0077) & 0.022 \\
  YHR136C & SPL2 & 0.0257 (0.0116) & (0.0029, 0.0485) & 0.027 \\
  YHR184W & SSP1 & 0.0468 (0.0229) & (0.0019, 0.0916) & 0.041 \\
  YAL035W & FUN12 & 0.0757 (0.0375) & (0.0022, 0.1491) & 0.043 \\
  YHR005C & GPA1 & 0.0766 (0.0381) & (0.0020, 0.1512) & 0.044 \\
  YHR169W & DBP8 & -0.1035 (0.0526) & (-0.2067, -0.0003) & 0.049 \\
   \hline
\end{tabular}

\end{table}

\section{Discussion}\label{section: discussion}
In this paper, we have considered the problem of mediation analysis in the setting where we have high-dimensional and possibly interacting mediators. \tcr{DAGs} 
are used to characterize the possible interaction effects among the high dimensional mediators and to define the individual mediation effects based on linear structural equation models.
We have developed an IDA-based procedure, MIDA, to estimate the individual mediation effects, which takes into account the uncertainty of the estimated DAGs. We have also derived the asymptotic distributions of the estimates of both interventional effects as well as individual mediation effects, under the assumption of sub-Gaussian errors for the LSEMs, which facilitates inference based on these estimators, and to the best of our knowledge, are the first such results available in the literature. We have illustrated the methods in simulation studies with promising performance, as well as using a real data set on yeast in order to identify the possible gene expression mediators of a genetic variant for yeast growth in the presence of drugs. The methods can also be applied to the problem of identifying the important gene expression or methylation mediators that mediate the effects of genetic variants identified through genome-wide association studies on disease phenotypes.
%

\tcr{Another crucial contribution of this work lies in the results of Section \ref{section: linear regression theory} on uniform ALEs and non-asymptotic control of error terms for linear regression estimators based on varying subsets of high dimensional covariates that serve as the backbone of our results on the asymptotic distribution of the MIDA estimators. These results are applicable quite generally in several other problems, and should therefore be of independent interest.} 

\vspace{0.01in}
\tcR{\emph{FDR control for MIDA.}} Our simulations have shown that the confidence intervals based on MIDA provide correct coverage when the mediation effects are not too small. In high dimensional settings with thousands of possible mediators, large sample sizes are needed in order to accurately estimate the individual mediation effects. In order to estimate the set of true mediators $\{X_j : \eta_j \neq 0\}$, as an alternative to our heuristic \tcr{$p$-value} based thresholding approach proposed in Section \ref{results}, one can also apply the Benjamini-Hochberg (BH) false discovery rate \tcR{(FDR)} control procedure at some desired level $\alpha$ (e.g.\tcR{,} $\alpha = 0.1$). However, the theoretical guarantee of the BH procedure does \emph{not} apply in our case due to the fact that the true DAG is not identifiable. Instead, the BH procedure (asymptotically) guarantees to control FDR for estimating 
$\{ X_j : \theta_{tj}\mathrm{aver}(\Theta_{jp}) \neq 0\}$, \tcr{and only} 
when the true CPDAG is known. Furthermore, Table \ref{table: F-score} suggests that it might be unreasonable and too optimistic to enforce a high precision level such as 0.9 (equivalently, FDR control at level $\alpha = 0.1$) in the challenging problem of estimating the set of true mediators in high-dimensional settings. For these reasons, we recommend the estimation of the target set by maximizing the F-score because of its adaptive capability of automatically adjusting to the best achievable precision level (or, least FDR level) for the problem at hand. 

We also empirically observed \tcr{(see Figure \ref{fig: FDR curves_noscreen} in \tcR{Appendix \ref{sec:FDR_Control} of} the \hyperref[supp_mat]{Supplement})} that the BH procedure 
\tcr{becomes} very \tcr{\emph{conservative}} for estimating $\{ X_j : \theta_{tj}\mathrm{aver}(\Theta_{jp}) \neq 0\}$, mainly because the $p$-value corresponding to the test \tcr{$\theta_{tj}\mathrm{aver}(\Theta_{jp}) = 0$} has a non-uniform and left-skewed distribution
when both $\theta_{tj}$ and $\mathrm{aver}(\Theta_{jp})$ are 0 
 (see also the last paragraph of \tcr{Section \ref{subsection: inference}}). One way to mitigate this issue is to apply a heuristic \emph{screening procedure}. We explored this approach in our simulation studies to \tcr{first} obtain a potential set of mediators for which the total effect of the treatment on the mediator is non-zero, and then apply the BH procedure on \tcr{this} selected set. We provide further discussions on this approach in \tcR{Appendix \ref{sec:FDR_Control} of the \tcr{\hyperref[supp_mat]{Supplement}}.} 
 The simulation results presented therein show that this strategy can significantly improve the FDR controlling compared to a direct application of the BH procedure based on the $p$-values. 
%
%





\section*{\tcR{Supplement: Appendices \ref{section: proofs}--\ref{sec:FDR_Control}}}\label{supp_mat}
In the supplement, we provide: (i) proofs of all theoretical results in the main paper \tcr{(\tcR{Appendix} 
\ref{section: proofs});} (ii) additional technical tools, \tcr{including definitions and key} supporting lemmas \tcr{that are required in all the proofs (\tcR{Appendix} 
 \ref{sec:technicaltools});} (iii) \tcr{discussions on verifying the moment conditions required in the last part of Remark \ref{remark: linear functionals} (\tcR{Appendix} 
 \ref{sec:remarkverify});} and (iv) additional numerical results regarding FDR controlling \tcr{for MIDA (\tcR{Appendix} 
 \ref{sec:FDR_Control}).}



\begin{appendix}
\end{appendix}










%

\numberwithin{equation}{section}
\setcounter{equation}{0}         

\setcounter{section}{0}

\renewcommand{\thefigure}{\Alph{figure}}  
\setcounter{figure}{0}
\renewcommand{\thetable}{\Alph{table}}  
\setcounter{table}{0}

\section{Proofs of All Results}\label{section: proofs}

\subsection{Proof of Proposition \ref{proposition: product formula}}
By Theorem 3.1 of \citet{NandyMaathuisRichardson17}, we have,
\[
\theta_{tp}^{(t,j)} \;=\; \theta_{tp} - \theta_{tj}\theta_{jp}.
\]
Hence, $\eta_j = \theta_{tp} - \theta_{tp}^{(t,j)} = \theta_{tj}\theta_{jp}.$ \qed

\subsection{Proof of Lemma \ref{lemma: covariate adjustment}}
The linearity assumption implies that the total causal effect of $X_i$ on $X_k$ can be expressed as the coefficient of $X_i$ in the linear regression of $X_k$ on $X_i \cup X_S$ \citep{MaathuisKalischBuehlmann09, NandyMaathuisRichardson17} for a set of covariates $\mathbf{X}_S \subseteq \mathbf{X} \setminus \{X_i, X_k\}$ that satisfies Pearl's back-door criterion (see Definition 3.3.1 of \cite{Pearl00}). Using this, we obtain
\begin{align*}
\theta_{tj} & \;=\; \beta_{tj | \mathbf{X}_{\{1,\ldots, t-1\}}} ~~\text{and}~~ \theta_{jp} \; =\;  \beta_{jp | \mathbf{Pa}_{\mathcal{G}_0}(X_j) \cup \mathbf{X}_{\{1,\ldots,t\}}}.
\end{align*}
Recall that $\mathcal{G}_0'$ denotes the DAG on the set of potential mediators. Then the result follows from the fact that $\mathbf{Pa}_{\mathcal{G}_0}(X_j)\cup\{X_1,\ldots,X_t\} = \mathbf{Pa}_{\mathcal{G}_0'}(X_j)\cup\{X_1,\ldots,X_t\} $, \tcr{followed by an application of Proposition \ref{proposition: product formula}}. \qed

\subsection{Proof of Lemma \ref{lemma: sub-LSEM}}
The conditions $(B_{\mathcal{G}_0})_{ji} = 0$ for all $i\leq t$ imply that $\mathbf{X}_{\{1,\ldots,t\}} - \boldsymbol{\mu}_{\{1,\ldots,t\}} = (\mathrm{I} - (B_{\mathcal{G}_0})_{\{1,\ldots,t\}\{1,\ldots,t\}}^T)^{-1}  \epsilon_{\{1,\ldots,t\}}$. Therefore, it follows from the independence of the error variables that
\begin{align*}
\Exp[\boldsymbol{\epsilon} \mid \mathbf{X}_{\{1,\ldots,t\}}] & \; = \; \tcr{( \epsilon_{\{1,\ldots,t\}}, 0, \ldots, 0)^T} \\
& \; = \; \mathbf{e}_{\{1,\ldots,t\},p}^T(\mathrm{I} - (B_{\mathcal{G}_0})_{\{1,\ldots,t\}\{1,\ldots,t\}}^T)~(\mathbf{X}_{\{1,\ldots,t\}} - \boldsymbol{\mu}_{\{1,\ldots,t\}}),
\end{align*}
where $\mathbf{e}_{\{1,\ldots,t\},p}^T$ denotes the first $t$ columns of a $p \times p$ identity matrix.

\par\smallskip
Let $A$ denote the $(p-t-1) \times p$ matrix such that $\mathbf{X}' - \boldsymbol{\mu}'= A(\mathbf{X} - \boldsymbol\mu)$, \tcr{where $\boldsymbol{\mu}' := \E(\mathbf{X}')$.} Then, we have,
\begin{align*}
&\Exp[(\boldsymbol{\mathbf{X}'}- \boldsymbol{\mu}') \mid \mathbf{X}_{\{1,\ldots,t\}}] \\
& \; = \; A\Exp[(\mathbf{X} - \boldsymbol\mu) \mid \mathbf{X}_{\{1,\ldots,t\}}] \\
& \; = \; A(\mathrm{I} - B_{\mathcal{G}_0}^T)^{-1} \Exp[\boldsymbol{\epsilon} \mid \mathbf{X}_{\{1,\ldots,t\}}] \\
& \; = \; A(\mathrm{I} - B_{\mathcal{G}_0}^T)^{-1} \mathbf{e}_{\{1,\ldots,t\},p}^T(\mathrm{I} - (B_{\mathcal{G}_0})_{\{1,\ldots,t\}\{1,\ldots,t\}}^T)~(\mathbf{X}_{\{1,\ldots,t\}} - \boldsymbol{\mu}_{\{1,\ldots,t\}}).
\end{align*}
This completes the proof of the linearity property of the conditional expectation $\Exp[\boldsymbol{\mathbf{X}'} \mid \mathbf{X}_{\{1,\ldots,t\}}]$. \qed 

\par\smallskip
\noindent Next, note that for each $j \in \{t+1,\ldots,p-1\}$,
\begin{align*}
\Exp[X_j \mid \mathbf{X}_{\{1,\ldots,t\}}] \; &= \; \mu_j ~+~ \sum_{k = 1}^{p} (B_{\mathcal{G}_0})_{kj} \Exp[X_k \mid \mathbf{X}_{\{1,\ldots,t\}}] ~+~ \Exp[\epsilon_j \mid \mathbf{X}_{\{1,\ldots,t\}}] \\
&= \; \mu_j ~+~ \sum_{k = 1}^{t} (B_{\mathcal{G}_0})_{kj}X_k ~+~ \sum_{k = t+1}^{p-1} (B_{\mathcal{G}_0})_{kj} \Exp[X_k \mid \mathbf{X}_{\{1,\ldots,t\}}].
\end{align*}
The last equality follows from the fact that $(B_{\mathcal{G}_0}){pj} = 0$ and $\Exp[\epsilon_j \mid \mathbf{X}_{\{1,\ldots,t\}}] = \Exp[\epsilon_j]  = 0$. This implies
\[
X_j - \Exp[X_j \mid \mathbf{X}_{\{1,\ldots,t\}}] \; = \; \sum_{k = t+1}^{p-1} (B_{\mathcal{G}_0})_{kj} \left(X_k ~-~ \Exp[X_k \mid \mathbf{X}_{\{1,\ldots,t\}}]\right) + \epsilon_j,
\]
for all $j \in \{t+1,\ldots,p-1\}$. This completes the proof of $\mathbf{X}^{\dagger}= B_{\mathcal{G}_0'}^T \mathbf{X}^{\dagger} + \boldsymbol{\epsilon}'$. \qed

\par\medskip
\noindent Finally, we show that the faithfulness of the distribution of $\mathbf{X}^{\dagger}$ to $\mathcal{G}_0'$ follows from the faithfulness of the distribution of $\mathbf{X}$ to $\mathcal{G}_0$. Suppose $X_{i}^{\dagger}$ and $X_{k}^{\dagger}$ are conditionally independent given $\mathbf{X}_S^{\dagger}$ for some set $S \subseteq \{t+1,\ldots, p-1\} \setminus \{i,k\}$. In order to establish faithfulness of $\mathbf{X}^{\dagger}$ to $\mathcal{G}_0'$, we need to show that $X_i$ and $X_k$ are d-separated by $\mathbf{X}_S$ in $\mathcal{G}_0'$, that is $\mathbf{X}_S$ \emph{blocks} every \emph{path} between $X_i$ and $X_k$ in $\mathcal{G}_0'$. A path in a graph is a sequence of distinct nodes such that all pairs of successive nodes in the sequence are adjacent in the graph, and $S$ blocks a path in $\mathcal{G}_0'$ if the path contains a non-collider that is in $\mathbf{X}_S$, or the path contains a collider that has no descendant in $\mathbf{X}_S$, where $(X_{r},X_{s},X_{t})$ a collider in a graph $\mathcal{G}$ if $\{X_r, X_t\} \subseteq \mathbf{Pa}_{\mathcal{G}}(X_s)$.

Since there is no directed path from $X_i$ or $X_k$ to a node in $\mathbf{X}_{\{1,\ldots,t\}}$, no node in $\mathbf{X}_{\{1,\ldots,t\}}$ can be a collider on a path between $X_i$ and $X_k$. Further, since $(B_{\mathcal{G}_0})_{pj} = 0$ for all $j$, $X_p$ cannot be non-collider on any path in $\mathcal{G}_0$. These imply $X_i$ and $X_k$ are d-separated by $\mathbf{X}_S$ in $\mathcal{G}_0'$ if and only if $X_i$ and $X_k$ are d-separated by $\mathbf{X}_S\cup \mathbf{X}_{\{1,\ldots,t\}}$ in $\mathcal{G}_0$, since all paths between $X_i$ and $X_k$ in $\mathcal{G}_0$ that are not present in $\mathcal{G}_0'$ must go through $\mathbf{X}_{\{1,\ldots,t\}}$ or $X_p$. Therefore, it is sufficient to show that $X_i$ and $X_k$ are d-separated by $\mathbf{X}_S\cup \mathbf{X}_{\{1,\ldots,t\}}$ in $\mathcal{G}_0$. This is equivalent to show that the partial correlation between $X_i$ and $X_k$ given $\mathbf{X}_S\cup \mathbf{X}_{\{1,\ldots,t\}}$, denoted by $\rho_{ik|S\cup\{1,\ldots, t\}}$, is zero, as the distribution of $\mathbf{X}$ is generated from a LSEM and faithful to $\mathcal{G}_0$ (see \cite{SpirtesEtAl98, NandyHauserMaathuis16}).

Note that \tcr{the fact that} $X_{i}^{\dagger}$ and $X_{k}^{\dagger}$ are conditionally independent given $\mathbf{X}_S^{\dagger}$ implies that the partial correlation between $X_i^{\dagger}$ and $X_k^{\dagger}$ given $\mathbf{X}_S^{\dagger}$, denoted by $\rho_{ik|S}^{\dagger}$, is zero. This completes the proof, since Lemma \ref{lemma: partial correlation} 
\tcr{ensures that} $\rho_{ik|S\cup\{1,\ldots,t\}}^{2} = \rho_{ik|S}^{\dagger}$. \qed

\subsection{Proof of Lemma \ref{lemma: partial correlation}}


\tcr{Let $\Sigma^{\dagger}:=\Cov(\mathbf{X}^\dagger)$. Then,} note that the linearity of conditional expectation $\Exp[\mathbf{X}' \mid X_1]$ (Lemma \ref{lemma: sub-LSEM}) implies
$$
\Sigma^{\dagger} = \Sigma_{(t+1,\ldots,p-1)(t+1,\ldots,p-1)} - \Sigma_{(t+1,\ldots,p-1)(1,\ldots,t)}\Sigma_{(1,\ldots,t)(1,\ldots,t)}^{-1}\Sigma_{(1,\ldots,t)(t+1,\ldots,p-1)}.
$$
Recall that if $\sigma_{i|S}^{\dagger} := \Sigma^{\dagger}_{ii} - \Sigma^{\dagger}_{iS}(\Sigma^{\dagger}_{SS})^{-1}\Sigma^{\dagger}_{Si}$ is the variance of the residuals in the linear regression (based on $\Sigma^{\dagger} = \Cov(\mathbf{X}^\dagger)$) of $X_i^{\dagger}$ on $\mathbf{X}_S^{\dagger}$ and $\sigma_{i|S\cup\{k\}}^{\dagger2}$ is the variance of the residuals in the linear regression of $X_i^{\dagger}$ on $\mathbf{X}_{S\cup\{k\}}^{\dagger}$, \tcr{then it holds that} \citep{Yule1907}
\begin{align}\label{eq: yule1907}
\sigma_{i|S\cup\{k\}}^{\dagger2} \; = \; (1 - \rho_{ik|S}^{\dagger2})\sigma_{i|S}^{\dagger2}.
\end{align}

By applying the identity for expressing the Schur complement of a $(r-1) \times (r-1)$ principal submatrix of a $r \times r$ matrix as the ratio of determinants, we obtain
\begin{eqnarray*}
\sigma_{i|S}^{\dagger2} &=& \frac{|\Sigma_{(i,S)(i,S)}^{\dagger}|}{ |\Sigma_{SS}^\dagger|} \;\; = \; \frac{|\Sigma_{(i,S)(i,S)} - \Sigma_{(i,S)(1,\ldots,t)}\Sigma_{(1,\ldots,t)(1,\ldots,t)}^{-1}\Sigma_{(1,\ldots,t)(i,S)}|}{|\Sigma_{SS} - \Sigma_{S(1,\ldots,t)}\Sigma_{(1,\ldots,t)(1,\ldots,t)}^{-1}\Sigma_{(1,\ldots,t)S}|} \\
&=& \frac{|\Sigma_{(i,S,1,\ldots,t)(i,S,1,\ldots,t)}|\cdot|\Sigma_{(1,\ldots,t)(1,\ldots,t)}|}{|\Sigma_{(1,\ldots,t)(1,\ldots,t)}|\cdot|\Sigma_{(S,1,\ldots,t)(S,1,\ldots,t)}|} \;\; = \; \sigma_{i|S\cup\{1,\ldots,t\}}^2,
\end{eqnarray*}
where $\sigma_{i|S\cup\{1,\ldots,t\}}^2$ is the variance of the residuals in the linear regression (based on $\Sigma = \Cov(\mathbf{X})$) of $X_i$ on $\{X_{r} : r \in S\cup\{1,\ldots,t\}\}$. Similarly, we have $\sigma_{i|S\cup\{k\}}^{\dagger2} = \sigma_{i|S\cup\{1,\ldots,t\}\cup\{k\}}$. Hence, from \ref{eq: yule1907}, we have
\[
1 - \rho_{ik|S}^{\dagger2} \; = \; \frac{\sigma_{i|S\cup\{k\}}^{\dagger2}}{\sigma_{i|S}^{\dagger2}} \; = \; \frac{\sigma_{i|S\cup\{1,\ldots,t\}\cup\{k\}}^{2}}{\sigma_{i|S\cup\{1,\ldots,t\}}^{2}} \; = \; 1 - \rho_{ik|S\cup\{1,\ldots,t\}}^{2}. \quad \qed
\]


\subsection{Proof of Theorem \ref{theorem: high-dimensional consistency}}
\tcr{We first state and prove a lemma that will be useful in the main proof. The proof is then presented in several parts.}
\begin{lemma}\label{lemma: bounded psi2 norm}
Let $\mathbf{X}_{}$ be generated from a LSEM characterized by $(B_{\mathcal{G}_{0}},\boldsymbol{\epsilon}_{})$. Then Assumptions \ref{assumption: LSEM}, \ref{assumption: sparsity} and \ref{assumption: eigenvalue condition} imply that for any $S \subseteq \{1,\ldots, p\}$ such that $|S| \leq q + t + 2$,
$$
||\mathbf{X}_S - \boldsymbol{\mu}_S||_{\psi_2}^* \; \leq \; C_4,
$$ where $C_4 > 0$ is an absolute constant depending on $C_1$, $C_2$ and $C_3$ given by Assumptions \ref{assumption: LSEM} and \ref{assumption: eigenvalue condition}, and  $||\cdot||_{\psi_2}^*$ denotes the vector sub-Gaussian norm given by Definition \ref{def:subgaussvect}. 
\end{lemma}
\begin{proof}
Fix $S \subseteq \{1,\ldots, p\}$ such that $|S| \leq q + t + 2$. Let $A$ be the $|S|\times p$ matrix such that $\mathbf{X}_S - \boldsymbol{\mu}_S= A(\mathbf{X} - \boldsymbol{\mu})$. Therefore, we have, $\mathbf{X}_S - \boldsymbol{\mu}_S = A (I - B_{G_0}^T)^{-1}\boldsymbol{\epsilon}$. Hence,
\begin{align*}
||\mathbf{X}_S - \boldsymbol{\mu}_S||_{\psi_2}^* &\; =~ \sup_{||\mathbf{v}||_2 = 1} ||\mathbf{v}^TA (I - B_{G_0}^T)^{-1} \boldsymbol{\epsilon}||_{\psi_2} \\
&\; =~ \sup_{||\mathbf{v}||_2 = 1} ||\mathbf{v}^TA (I - B_{G_0}^T)^{-1}||_2 ~ \left|\left|\frac{\mathbf{v}^TA (I - B_{G_0}^T)^{-1}\boldsymbol{\epsilon}}{||\mathbf{v}^TA (I - B_{G_0}^T)^{-1}||_2}\right|\right|_{\psi_2} \\
&\; \leq ~ ||A (I - B_{G_0}^T)^{-1}||_2~ ||\boldsymbol{\epsilon}||_{\psi_2}^*,
\end{align*}
where the last inequality follows from the definitions of spectral norm and $||\cdot||_{\psi_2}^*$ norm and the fact that $\underset{\mathbf{x}}{\sup}~ f(\mathbf{x})~ g(\mathbf{x}) ~\leq~ \underset{\mathbf{x}}{\sup}~ f(\mathbf{x}) ~\underset{\mathbf{x}}{\sup} ~ g(\mathbf{x})$.

Since $\epsilon_1, \epsilon_2,\ldots,\epsilon_p$ are independent zero-mean sub-Gaussian random variables satisfying $\underset{1 \leq i \leq p}{\max}||\epsilon_i||_{\psi_2} \leq C_1$, it follows from Lemma 5.24 of \cite{Vershynin12} that $||\boldsymbol{\epsilon}||_{\psi_2}^* \leq C_0C_1$, for some absolute constant $C_0$. Furthermore, since $(\Sigma_0)_S = A (I - B_{G_0}^T)^{-1} D (I - B_{G_0}^T)^{-T} A^T$ for $D := \Cov(\boldsymbol{\epsilon})$, it follows from the sub-multiplicity property of the spectral norm that
\[
||A (I - B_{G_0}^T)^{-1}||_2 \; \leq \; ||(\Sigma_0)_S^{1/2}||_2 ~||D^{-1/2}||_2 \; \leq \; \sqrt{C_3} ~||D^{-1/2}||_2,
\]
where the last inequality follows from Assumption \ref{assumption: eigenvalue condition}.

Thus it remains to show that $||D^{-1/2}||_2$ is bounded. To this end, note that $$||D^{-1/2}||_2 \; = \; \frac{1}{\sqrt{\underset{1 \leq i \leq p}{\min} \Var(\epsilon_i)}}.$$
Finally, from the interlacing property of
eigenvalues of a Hermitian matrix A and the eigenvalues of the Schur complement
of any principal submatrix of A (see, e.g., Corollary 2.3 of \cite{Zhang05}), it follows that
\begin{align*}
\Var(\epsilon_i) &\; =~ (\Sigma_0)_i - (\Sigma_0)_{i\mathbf{Pa}_{\mathcal{G}_0}(i)}(\Sigma_0)_{\mathbf{Pa}_{\mathcal{G}_0}(i)\mathbf{Pa}_{\mathcal{G}_0}(i)}^{-1}(\Sigma_0)_{\mathbf{Pa}_{\mathcal{G}_0(i)} i} \\
& \; \geq~  \lambda_{\min}\left((\Sigma_0)_{(i, \mathbf{Pa}_{\mathcal{G}_0}(i))(i, \mathbf{Pa}_{\mathcal{G}_0}(i))}\right) \\
& \; \geq~ C_2,
\end{align*}
where $\mathbf{Pa}_{\mathcal{G}_0}(i) = \{r : X_r \in \mathbf{Pa}_{\mathcal{G}_0}(X_i)$. The last inequality follows from Assumption \ref{assumption: eigenvalue condition}, since from Assumption \ref{assumption: sparsity}, we have $|\mathbf{Pa}_{\mathcal{G}_0}(i)| \leq q$.

\tcr{Combining the bounds for $\Var(\epsilon_i), ||D^{-1/2}||_2$ and $||A (I - B_{G_0}^T)^{-1}||_2$ above and applying them in the original bound for $||\mathbf{X}_S - \boldsymbol{\mu}_S||_{\psi_2}^*$ yields the desired result. This completes the proof of Lemma \ref{lemma: bounded psi2 norm}.}
\end{proof}

\begin{proofof}[the first part of Theorem \ref{theorem: high-dimensional consistency}]

Let $A_n := \{\hat{\mathcal{C}}'(\lambda_n) = \mathcal{C}_{0}\}$. 
Since $\Prob(A_n) \longrightarrow 0$ (by Assumption \ref{assumption: structure learning consistency}) and $\mathrm{aver}(\hat{\Theta}_{jp}(\hat{\mathcal{C}}'))$ $= \mathrm{aver}(\hat{\Theta}_{jp}(\mathrm{C}_{0}'))$ on the set $A_n$, it is sufficient to show that for any $\delta>0$,
\begin{align*}
\Prob\left(\underset{t< j < p}{\max}| \mathrm{aver}(\hat{\Theta}_{jp}(\mathcal{C}_{0}')) - \mathrm{aver}(\Theta_{jp})|>\delta,~A_n\right) \; \longrightarrow \; 0.
\end{align*}
For $j \in \{t+1,\ldots,p-1\}$, we denote the distinct elements in the multi-set $\{\mathbf{Pa}_{\mathcal{G}}(X_j)\cup\{\mathbf{X}_{\{1,\ldots,t\}}\} : \mathcal{G} \in \MEC(\mathcal{C}_0')\}$ by $\{\mathbf{X}_{S_{j1}},\ldots,\mathbf{X}_{S_{jm_j}}\}$. By Assumption \ref{assumption: sparsity}, we have $m_{j} \leq 2^{q}$ for all $j \in \{t+1,\ldots,p-1\}$. Therefore,
\begin{align*}
&~~\Prob\left(\underset{t< j < p}{\max}| \mathrm{aver}(\hat{\Theta}_{jp}(\mathcal{C}_{0}')) - \mathrm{aver}(\Theta_{jp})|>\delta,~A_n\right) \\
\leq&~~\Prob\left(\underset{t< j < p}{\max}~\underset{1\leq r \leq m_{j}}{\max} | \hat{\beta}_{jp|\mathbf{X}_{S_{jr}}} - \beta_{jp|\mathbf{X}_{S_{jr}}} | ~>\delta,~A_n\right) \\
\leq& ~~(p-2)2^{q}\underset{t< j < p}{\max}~\underset{1\leq r \leq m_{j}}{\max} \Prob\left( | \hat{\beta}_{jp|\mathbf{X}_{S_{jr}}} - \beta_{jp|\mathbf{X}_{S_{jr}}} |~>\delta \right).
\end{align*}

We complete the proof by showing that for all $j \in \{t+1,\ldots,p-1\}$ and $r \in \{1,\ldots,m_{j}\}$, $\Prob(| \hat{\beta}_{jp|\mathbf{X}_{S_{jr}}} - \beta_{jp|\mathbf{X}_{S_{jr}}} |>\delta) \leq 2\exp(-C_7n\delta^2)$ for some absolute constant $C_7>0$. Note that this implies
\begin{eqnarray*}
&& (p-2)2^{q}\underset{t< j < p}{\max}~\underset{1\leq r \leq m_{nj}}{\max} \Prob(| \hat{\beta}_{jp|\mathbf{X}_{S_{jr}}} - \beta_{jp|\mathbf{X}_{S_{jr}}} |>\delta) \\
&& \qquad \leq \; O(\exp((\log(2)+\log(p))q - C_7 n\delta^2)) \; \longrightarrow \; 0,
\end{eqnarray*}
since from Assumptions \ref{assumption: high-dim} and \ref{assumption: sparsity}, we have $p = O(n^a)$ and $q = O(n^{1-b_1})$ for some $0\leq a < \infty$ and $0 < b_1 \leq 1$.

Fix $j \in \{t+1,\ldots,p-1\}$ and $r \in \{1,\ldots,m_{j}\}$. Let $\Sigma$ and $\Sigma'$ denote the submatrices of $\Sigma_{0} = \Cov(\mathbf{X})$ that corresponds to $(X_{j},\mathbf{X}_{S_{jr}},X_{p})$ and $(X_{nj},\mathbf{X}_{S_{jr}})$ respectively. Then $\beta_{jp|\mathbf{X}_{S_{jr}}} = \mathbf{e}_1^{T}\Sigma^{\prime-1}\boldsymbol\sigma_{p}$, where $\mathbf{e}_1$ denote the first column of an identity matrix of appropriate order and $\boldsymbol\sigma_{p}$ denote the last column of $\Sigma$. Similarly, we define the corresponding sample covariance matrices $\hat{\Sigma}$ and $\hat{\Sigma}'$ to obtain
$\hat{\beta}_{jp|\mathbf{X}_{S_{jr}}} = \mathbf{e}_1^{T}\hat{\Sigma}^{\prime-1}\hat{\boldsymbol\sigma}_{p}$, where $\hat{\boldsymbol\sigma}_{p}$ denote the last column of $\hat{\Sigma}$. We show below that
\begin{align}\label{eq: trivial inequality}
|  \mathbf{e}_1^{T}\hat{\Sigma}^{\prime-1}\hat{\boldsymbol\sigma}_{p} - \mathbf{e}_1^{T}\Sigma^{\prime-1}\boldsymbol\sigma_{p} | \;\; & \leq \; \frac{1}{C_2} ||\hat{\Sigma} - \Sigma||_2 + C_3 ||\hat{\Sigma}^{\prime-1} - \Sigma^{\prime-1} ||_2 \nonumber \\
& \;\;\; + \; ||\hat{\Sigma} - \Sigma||_2 ||\hat{\Sigma}^{\prime-1} - \Sigma^{\prime-1} ||_2,
\end{align}
where for a matrix $A$, $||A||_2$ denote its spectral norm, and $C_3$ and $C_2$ are given by Assumption \ref{assumption: eigenvalue condition}. To this end, we first apply the inequality
$$
|\mathbf{a}_1^T\mathbf{a}_2 - \mathbf{b}_1\mathbf{b}_2| ~\leq~ ||\mathbf{b}_1||_2 ||\mathbf{a}_2 - \mathbf{b}_2 ||_2 \hspace{0.03in}+\hspace{0.03in} ||\mathbf{b}_2||_2 ||\mathbf{a}_1 - \mathbf{b}_1 ||_2 \hspace{0.03in}+\hspace{0.03in} ||\mathbf{a}_1 - \mathbf{b}_1 ||_2 ||\mathbf{a}_2 - \mathbf{b}_2 ||_2,
$$
with $\mathbf{a}_1= \hat{\Sigma}_n^{\prime-1}\mathbf{e}_1$, $\mathbf{a}_2 = \hat{\boldsymbol\sigma}_{p}$, $\mathbf{b}_1= \Sigma^{\prime-1}\mathbf{e}_1$ and $\mathbf{b}_2 = \boldsymbol\sigma_{p}$, where for a vector $\mathbf{a}$, $||\mathbf{a}||_2$ denote its $\ell_2$ norm. Next, note that $||\Sigma^{\prime-1}\mathbf{e}_1||_2 \leq ||\Sigma^{\prime-1}||_2 \leq 1/C_2$, where the last inequality follows from Assumption \ref{assumption: eigenvalue condition} and Cauchy's interlacing theorem for eigenvalues of positive definite matrices, since $|S_{jr}| \leq q +t+2$. Similarly, we have $|| \boldsymbol\sigma_{p}||_2 \leq C_3$. This completes the proof of \eqref{eq: trivial inequality}.

From Lemma \ref{lemma: bounded psi2 norm}, we have $||(X_{j}, \mathbf{X}_{S_{jr}}^T,X_{p})^T||_{\psi_2}$ $< C_4$ for some constant $C_4>0$. Therefore, for any $\delta \in (0,1)$ and sufficiently large $n$, we have
\vspace{-0.05in}
\begin{align}\label{eq: concentration 1}
\Prob(||\hat{\Sigma} - \Sigma||_2 > \delta) \;\leq\; 2\exp(-C_5n\delta^2),
\end{align}
for some absolute constant $C_5>0$ depending on $C_4$ (see Corollary 5.50 of \cite{Vershynin12}). Similarly, for any $\delta \in (0,1)$ and sufficiently large $n$, we have
\vspace{-0.05in}
\begin{align}\label{eq: concentration 2}
\Prob(||\hat{\Sigma}' - \Sigma'||_2 > \delta) \;\leq\; 2\exp(-C_5n\delta^2).
\end{align}

We show below that a similar result holds for $||\hat{\Sigma}^{\prime-1} - \Sigma^{\prime-1}||_2$. To this end, we consider $\delta \leq C_2/2$ and $||\hat{\Sigma} - \Sigma||_2 \leq \delta$. Using the sub-multiplicity property of the spectral norm, we obtain
$$
||(\hat{\Sigma}' - \Sigma')\Sigma^{\prime-1}||_2 \;\leq\; ||\hat{\Sigma}' - \Sigma'||_2 ||\Sigma^{\prime-1}||_2  \;\leq\; \frac{\delta}{C_2} \; \leq \; 1/2  < 1.
$$  This implies $(\hat{\Sigma}' - \Sigma')\Sigma^{\prime-1} + \mathrm{I}$ is invertible and the following inequality holds (see, for example, Section 5.8 of \citet{HornJohnson90}):
\begin{eqnarray}
 ||((\hat{\Sigma}' - \Sigma')\Sigma^{\prime-1} + \mathrm{I})^{-1} - \mathrm{I}||_2 &\leq& \frac{||(\hat{\Sigma}' - \Sigma')\Sigma^{\prime-1}||_2}{1-||(\hat{\Sigma}' - \Sigma')\Sigma^{\prime-1}||_2} \label{eq: matrix inversion} \\
&\leq & 2~||(\hat{\Sigma}' - \Sigma')\Sigma^{\prime-1}||_2 \quad \leq \; \frac{2\delta}{C_2}, \nonumber
\end{eqnarray}
where the second inequality follows from $||(\hat{\Sigma}' - \Sigma')\Sigma^{\prime-1}||_2 \leq 1/2$ and the third inequality follows from the sub-multiplicity property of the spectral norm, the assumption that $||\hat{\Sigma}' - \Sigma'||_2 \leq \delta$ and Assumption \ref{assumption: eigenvalue condition}.

Therefore, $||\hat{\Sigma}' - \Sigma'||_2 \leq \delta < C_2/2$ implies
\begin{align}\label{eq: matrix inversion 2}
||\hat{\Sigma}^{\prime-1} - \Sigma^{\prime-1}||_2 & \; = \; ||\Sigma^{\prime-1}\{((\hat{\Sigma}' - \Sigma')\Sigma^{\prime-1} + \mathrm{I})^{-1} - \mathrm{I}\}||_2 \nonumber \\
& \; \leq \; ||\Sigma^{\prime-1}||_2 \frac{2\delta}{C_2} \; \leq \; \frac{2\delta}{C_2^2},
\end{align}
where the second inequality follows from the sub-multiplicity property of the spectral norm and \eqref{eq: matrix inversion}, and the third inequality follows from Assumption \ref{assumption: eigenvalue condition}. By combining \eqref{eq: concentration 2} and \eqref{eq: matrix inversion 2}, we have, for any $\delta \in (0, C_2 / 2)$,
\vspace{-0.05in}
\begin{align}\label{eq: concentration 3}
\Prob(||\hat{\Sigma}_n^{\prime-1} - \Sigma_n^{\prime-1}||_2 > \delta) \; \leq \; 2\exp(-C_6n\delta^2),
\end{align}
where $C_6>0$ is an absolute constant depending on $C_2$ and $C_5$.

Finally, by combining, \eqref{eq: trivial inequality}, \eqref{eq: concentration 1} and \eqref{eq: concentration 2}, we obtain
\begin{align}\label{eq: concentration 4}
\Prob\left(| \hat{\beta}_{jp|\mathbf{X}_{S_{jr}}} - \beta_{jp|\mathbf{X}_{S_{jr}}} |>\delta \right) \; \leq \; 2\exp(-C_7n\delta^2),
\end{align}
for some absolute constant $C_7>0$ depending on $C_3$, $C_2$ and $C_4$. This completes the proof of the first part of Theorem \ref{theorem: high-dimensional consistency}. \qed
\end{proofof}

\begin{proofof}[the second part of Theorem \ref{theorem: high-dimensional consistency}]
\tcr{First we recall} that \tcr{$\hat{\eta}_{j}(\lambda_n)$} $= \hat{\beta}_{tj|\mathbf{X}_{\{1,\ldots,t\}}} \tcr{\times} \mathrm{aver}(\hat{\Theta}_{jp}(\lambda_n))$, and $\theta_{tj} = \beta_{tj | \mathbf{X}_{\{1,\ldots,t\}}}$. Therefore, we have
\begin{align}\label{eq: decomposition}
 | \hat{\eta}_{j}(\lambda_n) - \theta_{tj}\cdot\mathrm{aver}(\Theta_{jp}) | \;\; & \leq \; |\theta_{tj}|\cdot |\mathrm{aver}(\hat{\Theta}_{jp}(\lambda_n) - \mathrm{aver}(\Theta_{jp})| \nonumber \\
& \quad + \; |\mathrm{aver}(\Theta_{jp})| \cdot |\hat{\beta}_{tj | \mathbf{X}_{\{1,\ldots,t\}}} -  \beta_{tj | \mathbf{X}_{\{1,\ldots,t\}}}|.
\end{align}
From the \tcr{first} 
part of Theorem \ref{theorem: high-dimensional consistency}, we have
\begin{align}\label{eq: concentration 5}
\underset{t< j < p}{\max}\left| \mathrm{aver}(\hat{\Theta}_{jp}(\lambda_n)) - \mathrm{aver}(\Theta_{jp})\right| \; \overset{\Prob}{\longrightarrow} \; 0.
\end{align}
Further, by similar argument as given in the proof of the second part of Theorem \ref{theorem: high-dimensional consistency}, we can show that
\begin{align}\label{eq: concentration 6}
\Prob\left( \underset{t< j < p}{\max}|\hat{\beta}_{tj | \mathbf{X}_{\{1,\ldots,t\}}} -  \beta_{tj | \mathbf{X}_{\{1,\ldots,t\}}}|>\delta \right) \; \leq \; 2(p-2)\exp(-C_7n\delta^2) \; \longrightarrow \; 0,
\end{align}
where $C_7$ is as in \eqref{eq: concentration 4}.

Finally, note that for any $i,k \in \{t+1,\ldots,p\}$ and $S \subseteq \{1,\ldots,p\}\setminus\{i,k\}$ such that $|S| \leq q+t+2$, we have
\begin{align}\label{eq: upper bound 1}
|\beta_{ik|\mathbf{X}_S}| \; = \; |\mathbf{e}_1^{T}\Sigma^{\prime-1}\boldsymbol\sigma_{k}| \; \leq \; ||\Sigma^{\prime-1}\mathbf{e}_1||_2 ||\boldsymbol\sigma_{k}||_2 \; \leq \; \frac{C_3}{C_2},
\end{align}
where $\Sigma$ and $\Sigma'$ denote the submatrices of $\Sigma_{0} = \Cov(\mathbf{X})$ corresponding to $(X_{i},\mathbf{X}_S,X_{k})$ and $(X_{i},\mathbf{X}_S)$ respectively, $\boldsymbol\sigma_{k}$ is the last column of $\Sigma$ Note that the first inequality in \eqref{eq: upper bound 1} is Cauchy-Schwarz and the second inequality in \eqref{eq: upper bound 1} follows from Assumption \ref{assumption: eigenvalue condition} and Cauchy's interlacing theorem for eigenvalues of positive definite matrices (we used similar arguments in the proof of the second part of Theorem \ref{theorem: high-dimensional consistency}).

Since \eqref{eq: upper bound 1} implies that $|\theta_{tj}|$ and $|\mathrm{aver}(\Theta_{njp})|$ are bounded above by $C_3/C_2$ for all $j \in \{t+1,\ldots,p\}$, \tcr{the second} 
 part of Theorem \ref{theorem: high-dimensional consistency} \tcr{now} follows from \eqref{eq: decomposition}, \eqref{eq: concentration 5} and \eqref{eq: concentration 6}. \tcr{This completes the proof of Theorem \ref{theorem: high-dimensional consistency}.} \qed
\end{proofof}


\subsection{Proof of Theorem \ref{linreg:main_thm1}}\label{proof:main:thm1}
We first state below two supporting lemmas that serve as essential ingredients in our proof of Theorem \ref{linreg:main_thm1} and may also be of independent interest. Their proofs are given in \tcR{Appendices} 
\ref{proof:main:lem1} and \ref{proof:main:lem2}.

\par\medskip
\tcr{Recall the notations $\bSigmahatS, \bGammahatS, \bSigmahatSY,\bGammahatSY$ from Section \ref{section: linear regression theory} and further define:}
\begin{align*}
 & \tcr{\;\;\bSigmatilS  \;\; := \; \frac{1}{n} \sum_{i=1}^n (\bXSi - \bmuS)(\bXSi - \bmuS)^T \; \equiv \; \bSigmahatS + \bGammahatS, \quad \mbox{and}} \nonumber \\
 & \tcr{\bSigmatilSY  \; := \; \frac{1}{n} \sum_{i=1}^n (Y_i - \muY)(\bXSi - \bmuS) \;\; \equiv \; \bSigmahatSY + \bGammahatSY.} \nonumber 
\end{align*}

\begin{lemma}\label{linreg:main_lem1}
Suppose Assumption \ref{linreg:main_assmpn}(i) holds for a given $\Ssc \subseteq \Omega_{\Jsc}$ with $|\Ssc| := \Ln \equiv \LnS $ and $\sup_{S \in \Ssc} |S| \leq \qn \equiv \qnS \leq \min(n,p_n)$, and let $\rn := (\qn + \log \Ln)$. Then, for any $c > 0$, the following bounds hold:
\begin{align}
&(i) \;\; \P\left\{ \underset{S \in \Ssc}{\sup} \MLtwonorm{\bSigmatilS - \bSigmaS} > (c+1)\KSsc \left( \sqrt{\frac{\rn}{n}} + \frac{\rn}{n} \right) \right\} \leq \; 2 \exp\left( -c\rn \right), \label{eq:main_lem1:res1} \\ 
&(ii) \; \P\left\{ \underset{S \in \Ssc}{\sup} \MLtwonorm{\bGammahatS} > 16 (c+1)\KSsc \left(\frac{\sqrt{\rn}}{n} + \frac{\rn}{n} \right) + \frac{\lambdasupSsc}{n} \right\} \leq \; 2 \exp\left( -c\rn \right), \nonumber \\ 
&(iii) \; \P\left\{ \underset{S \in \Ssc}{\sup} \MLtwonorm{\bSigmahatS - \bSigmaS} > (c+1)\KSsc \left( \sqrt{\frac{\rn}{n}} + \frac{33 \rn}{n} \right) 
+ \frac{\lambdasupSsc}{n} \right\} \nonumber \\ 
&\qquad\qquad \leq \; 4 \exp\left( -c\rn \right). \nonumber 
\end{align}
Further, let $\rntil := (\rn +1)$, and suppose Assumption \ref{linreg:main_assmpn}(ii) also holds. Then, for any $c > 0$, the following bounds hold:
\begin{align}
& (i) \;\; \P\left\{ \underset{S \in \Ssc}{\sup} \Ltwonorm{\bSigmatilSY - \bSigmaSY} > (c+1)\KtilSsc \left( \sqrt{\frac{\rntil}{n}} + \frac{\rntil}{n} \right) \right\}  \leq \; 2 \exp\left(-c\rntil \right), \label{eq:main_lem1:res2} \\ 
& (ii) \; \P\left\{ \underset{S \in \Ssc}{\sup} \Ltwonorm{\bGammahatSY} > 16 (c+1)\KtilSsc \left(\frac{\sqrt{\rntil}}{n} + \frac{\rntil}{n} \right) + \frac{\lambdatilsupSsc}{n} \right\} \nonumber \\
& \qquad\qquad \leq \; 2 \exp(-c\rntil ). \nonumber
\end{align}
Lastly, the constants $\lambdasupSsc$ and $\lambdatilsupSsc$ may be chosen such that $\lambdasupSsc \leq 2 \sigmaXSsc^2$ and $\lambdatilsupSsc \leq 2 \sigmatilZSsc^2  \equiv 2 (\sigmaY + \sigmaXSsc)^2$. Moreover, $\sup_{s\in \Ssc} \Ltwonorm{\bbetaS}^2 \; \leq \; \Var(Y) \lambdainfSsc^{-1} \; \leq \; 2 \sigmaY^2 \lambdainfSsc^{-1} $.
\end{lemma}

\begin{lemma}\label{linreg:main_lem2}
Let $\Ssc$ and $\rn$ be as in Lemma \ref{linreg:main_lem1}, and suppose Assumption \ref{linreg:main_assmpn}(i) holds. Let $c^* > 0$ be any constant satisfying:
\begin{equation}
(c^*+1)\KSsc \left( \sqrt{\frac{\rn}{n}} + \frac{33 \rn}{n} \right)  + \frac{\lambdasupSsc}{n} \; \leq \; \frac{1}{2} \lambdainfSsc,  \label{eq:cstarcond}
\end{equation}
and let $\KSsc^* := 2 \lambdainfSsc^{-2} \KSsc$. Then, for any $c^* > 0$ as in (\ref{eq:cstarcond}) and for any $c> 0$, and defining $a_n(c,c^*,\rn) := \exp\left( -c\rn \right) + \exp\left( -c^*\rn \right)$, we have
\begin{eqnarray*}
&& \P\left\{ \underset{S \in \Ssc}{\sup} \MLtwonorm{\bSigmahatS^{-1} - \bSigmaS^{-1}} > (c+1)\KSsc^* \left( \sqrt{\frac{\rn}{n}} + \frac{33 \rn}{n} \right) + \frac{2}{n}\frac{\lambdasupSsc}{\lambdainfSsc^2} \right\} \nonumber \\
&& \qquad\;\; \leq \; 4 a_n(c,c^*,\rn) \; \equiv \; 4\exp\left( -c\rn \right) + 4\exp\left( -c^*\rn \right). \label{eq:main_lem2:res}
\end{eqnarray*}
\end{lemma}

The proof of Theorem \ref{linreg:main_thm1} essentially follows from carefully combining all the results established in Lemmas \ref{linreg:main_lem1} and \ref{linreg:main_lem2}. To this end, first note that under Assumption \ref{linreg:main_assmpn} and using Lemmas \ref{linreg:main_lem1} and \ref{linreg:main_lem2}, we have: for any $c > 0$ and any $c^* > 0$ satisfying condition (\ref{eq:cstarcond}),
\begin{equation}
\P\left\{ \underset{S \in \Ssc}{\sup} \MLtwonorm{\bSigmatilS - \bSigmaS}> \epsilon_{n,1}(c,\rn) \right\}  \leq  2 \exp(-c\rn), \;\; \underset{S \in \Ssc}{\sup} \Ltwonorm{\bbetaS}  \leq  \CSsc, \label{eq1:pf:mainthm1}
\end{equation}
\begin{equation}
 \P\left\{ \underset{S \in \Ssc}{\sup} \Ltwonorm{\bSigmatilSY - \bSigmaSY}> \epsilon_{n,2}(c,\rn) \right\}  \leq  2 \exp(-c\rntil)  \leq  2 \exp(-c \rn), \label{eq2:pf:mainthm1}
\end{equation}
\begin{equation}
\P\left\{ \underset{S \in \Ssc}{\sup} \MLtwonorm{\bSigmahatS^{-1} - \bSigmaS^{-1}} > \delta_{n}(c,\rn) \right\} \; \leq \; 4 \exp(-c \rn) + 4 \exp(-c^* \rn),\label{eq3:pf:mainthm1}
\end{equation}
\begin{equation}
\P\left\{ \underset{S \in \Ssc}{\sup} \Ltwonorm{\bGammahatSY}> \eta_{n,2}(c,\rn) \right\} \; \leq \; 2 \exp(-c\rntil) \; \leq \; 2 \exp(-c \rn), \;\; \mbox{and} \label{eq4:pf:mainthm1}
\end{equation}
\begin{equation}
\P\left\{ \underset{S \in \Ssc}{\sup} \MLtwonorm{\bGammahatS}> \eta_{n,1}(c,\rn) \right\}  \; \leq \;  2 \exp(-c\rn),  \;\; \underset{S \in \Ssc}{\sup} \MLtwonorm{\bSigmaS^{-1}}  \leq  \lambdainfSsc^{-1}, \label{eq5:pf:mainthm1}
\end{equation}
where $\{\epsilon_{n,j}(c,\rn), \eta_{n,j}(c,\rn)\}_{j=1}^2$, $\delta_n(c,\rn)$ and $\CSsc$ are all as defined in Theorem \ref{linreg:main_thm1}. Note that for (\ref{eq4:pf:mainthm1}) and (\ref{eq5:pf:mainthm1}), we also used $\sqrt{\rn} \leq \rn$, and for (\ref{eq2:pf:mainthm1}) and (\ref{eq4:pf:mainthm1}), we used $\rn \leq \rntil$.
\par\medskip
Next, noting that $(\bSigmaSY - \bSigmaS\bbetaS) = \bzero$ for any $S \in \Ssc$, due to (\ref{eq:targetdef}), we have:
\begin{eqnarray}
\underset{S \in \Ssc}{\sup} \Ltwonorm{\bTnS} &\equiv& \underset{S \in \Ssc}{\sup} \Ltwonorm{(\bSigmahatS^{-1} - \bSigmaS^{-1}) \{(\bSigmatilSY - \bSigmaSY) - (\bSigmatilS - \bSigmaS)\bbetaS\} } \nonumber \\
&\leq & \underset{S \in \Ssc}{\sup} \MLtwonorm{\bSigmahatS^{-1} - \bSigmaS^{-1}}  \underset{S \in \Ssc}{\sup}  \Ltwonorm{\bSigmatilSY - \bSigmaSY} \label{eq6:pf:mainthm1} \\
&& + \;\; \underset{S \in \Ssc}{\sup} \MLtwonorm{\bSigmahatS^{-1} - \bSigmaS^{-1}} \underset{S \in \Ssc}{\sup}  \MLtwonorm{\bSigmatilS - \bSigmaS} \underset{S \in \Ssc}{\sup} \Ltwonorm{\bbetaS}, \nonumber
\end{eqnarray}
where the inequality in (\ref{eq6:pf:mainthm1}) follows from multiple applications of Lemma \ref{lem:5:matbasic} (i). Using (\ref{eq1:pf:mainthm1}), (\ref{eq2:pf:mainthm1}) and (\ref{eq3:pf:mainthm1}) in (\ref{eq6:pf:mainthm1}), along with the union bound, we have: for any $c> 0$,
\begin{align*}
& \P \left[ \underset{S \in \Ssc}{\sup} \Ltwonorm{\bTnS} >  \delta_{n}(c,\rn) \{\epsilon_{n,1}(c,\rn) \CSsc + \epsilon_{n,2}(c,\rn) \} \right]  \\
& \qquad\qquad \leq \; 8 \exp(-c \rn) + 4 \exp(-c^* \rn). 
\end{align*}
This establishes the first of the two claims in result (i) of Theorem \ref{linreg:main_thm1}. \qed
\par\smallskip
%
Next, recall \tcr{$\bRnS \equiv \bSigmahatS^{-1} (\bGammahatS \bbetaS - \bGammahatSY)$,} 
and hence using Lemma \ref{lem:5:matbasic} (i),
\begin{eqnarray}
&& \;\; \supS \Ltwonorm{\bRnS}  \leq    \supS \Ltwonorm{\bSigmahatS^{-1}} \left( \supS  \MLtwonorm{\bGammahatS} \supS  \Ltwonorm{\bbetaS} + \supS \Ltwonorm{\bGammahatSY} \right). \label{eq7:pf:mainthm1}
\end{eqnarray}
Consequently, using (\ref{eq3:pf:mainthm1}), (\ref{eq4:pf:mainthm1}) and (\ref{eq5:pf:mainthm1}) in (\ref{eq7:pf:mainthm1}), along with the union bound, we have:
\begin{align*}
& \P \left[ \underset{S \in \Ssc}{\sup} \Ltwonorm{\bRnS} >  \{\delta_{n}(c,\rn) + \lambdainfSsc^{-1}\}\{\epsilon_{n,1}(c,\rn) \CSsc + \epsilon_{n,2}(c,\rn) \} \right] \\
& \qquad\qquad \leq \; 8 \exp(-c \rn) + 4 \exp(-c^* \rn) \quad \forall \; c > 0.
\end{align*}
This establishes the second and final claim in result (i) of Theorem \ref{linreg:main_thm1}. \qed
\par\smallskip
Finally, recall that $\bPsiS(\bZ) = \bSigmaS^{-1} \bpsiS(\bZ)$ and $n^{-1} \sum_{i=1}^n \bpsiS(\bZ_i) \equiv \bSigmatilSY - \bSigmatilS \bbetaS$. Further $(\bSigmaSY - \bSigmaS\bbetaS) = \bzero$ $\forall \; S \in \Ssc$, due to (\ref{eq:targetdef}). Hence, using  Lemma \ref{lem:5:matbasic} (i), we have:
\begin{align}
& \supS \Ltwonorm{ \frac{1}{n} \sum_{i=1}^n \bPsiS(\bZ_i)} \; \equiv \; \supS \Ltwonorm{\bSigmaS^{-1}\{(\bSigmatilSY - \bSigmaSY) - (\bSigmatilS - \bSigmaS) \bbetaS\}} \nonumber \\
&  \quad \leq \; \supS\MLtwonorm{\bSigmaS^{-1}} \left(\supS\Ltwonorm{\bSigmatilSY - \bSigmaSY} + \supS \MLtwonorm{\bSigmatilS - \bSigmaS} \supS \Ltwonorm{ \bbetaS}\right).\nonumber 
\end{align}
Hence, we have: for any $c > 0$,
\begin{align}
& \P \left[ \supS \Ltwonorm{ \frac{1}{n} \sum_{i=1}^n \bPsiS(\bZ_i)} >   \lambdainfSsc^{-1}\{\epsilon_{n,1}(c,\rn) \CSsc + \epsilon_{n,2}(c,\rn) \} \right] \nonumber \\
& \qquad \qquad \leq \; 4 \exp(-c\rn) \quad \forall \; c > 0, \nonumber
\end{align}
where the final probability bound follows from applying (\ref{eq1:pf:mainthm1}), (\ref{eq2:pf:mainthm1}) and (\ref{eq5:pf:mainthm1}), along with the union bound, to the preceding bound. This establishes the result (ii) of Theorem \ref{linreg:main_thm1}. \qed
\par\smallskip
Finally, all the `$\lesssim$' type bounds claimed in results (i) and (ii) are quite straightforward and follow trivially from the definitions of $\{\epsilon_{n,j}(c,\rn)\}_{j=1}^2\}$,  $\{\eta_{n,j}(c,\rn)\}_{j=1}^2$, $\delta_n(c,\rn)$ and $\CSsc$. The details are thus skipped here for brevity. The proof of Theorem \ref{linreg:main_thm1} is now complete. \qed

\subsection{Proof of Lemma \ref{linreg:main_lem1}}\label{proof:main:lem1}

Applying Lemma \ref{lem:7:matconc}, under Assumption \ref{linreg:main_assmpn} (i), to the random vectors $\{\bXSi - \bmuS\}_{i=1}^n$ for any $S \in \Ssc$, and recalling the definition of the constant $\KSsc > 0$ in (\ref{linreg:main_assmpn:eq1}) along with the fact that  $s \leq \qn$ $\forall \; S \in \Ssc$, it follows that for any $\epsilon \geq 0$ and for each $S \in \Ssc$,
\begin{eqnarray*}
&& \P \left\{ \MLtwonorm {\bSigmatilS - \bSigmaS} \; > \; \KSsc \left( \sqrt{\frac{\qn + \epsilon}{n}} + \frac{\qn + \epsilon}{n} \right) \right\} \; \leq \; 2 \exp\left( - \epsilon \right), \nonumber 
\end{eqnarray*}
or equivalently, for any $\epsilon \geq 0$ and $ S\in \Ssc$,
\begin{eqnarray}
&& \P \left\{ \MLtwonorm {\bSigmatilS - \bSigmaS} \; > \; \KSsc \left( \sqrt{\epsilon} + \epsilon \right) \right\} \; \leq \; 2 \exp\left( - n \epsilon +  \qn \right). \label{eq1:pf:mainlem1}
\end{eqnarray}
Consequently, using (\ref{eq1:pf:mainlem1}) along with the union bound, we then have:
\begin{eqnarray}
&& \P \left\{ \underset{S \in \Ssc} {\sup} \MLtwonorm {\bSigmatilS - \bSigmaS} \; > \; \KSsc \left( \sqrt{\epsilon} + \epsilon \right) \right\} \nonumber \\
&& \qquad \leq \; \sum_{s \in \Ssc} \P \left\{ \MLtwonorm {\bSigmatilS - \bSigmaS} \; > \; \KSsc \left( \sqrt{\epsilon} + \epsilon \right) \right\} \nonumber \\
&& \qquad \leq \; 2 \Ln \exp\left( - n \epsilon +  \qn \right)  \; \equiv \; 2\exp\left( - n \epsilon +  \qn + \log \Ln\right) \quad \forall \; \epsilon \geq 0. \label{eq2:pf:mainlem1}
\end{eqnarray}
Substituting $\epsilon$ in (\ref{eq2:pf:mainlem1}) above as: $\epsilon = (c + 1)(\qn + \log \Ln)/n \equiv (c+1) \rn/n$ for any $c \geq 0$, and noting that $\sqrt{c+1} \leq (c+1)$, we then have: $\forall \; c \geq 0$,
\begin{equation}
\P \left\{ \underset{S \in \Ssc} {\sup} \MLtwonorm {\bSigmatilS - \bSigmaS} \; > \; (c+1) \KSsc \left( \sqrt{\frac{\rn}{n}} + \frac{\rn}{n} \right) \right\} \; \leq \; 2 \exp( - c\rn). \label{eqfinal1:pf:mainlem1}
\end{equation}
This therefore establishes the first claim (i) in (\ref{eq:main_lem1:res1}). \qed
\par\medskip
Next, using Lemma \ref{lem:3:sgconc}, along with Lemma \ref{lem:2:sgvectprop} (i) and the definition of $\psitwonormstar{\cdot}$ in \ref{def:subgaussvect}, it follows, under Assumption \ref{linreg:main_assmpn} (i), that for any $S \in \Ssc$,
\begin{eqnarray*}
&& \psitwonorm{\bv^T(\bXbarS -\bmuS)} \; \leq \; (4 \sigmaXSsc/\sqrt{n}) \Ltwonorm{\bv} \;\;\; \mbox{for any} \; \bv \in \R^s, \quad \mbox{and thus}, \nonumber \\
&& \sup_{S \in \Ssc} \psitwonormstar{\bXbarS - \bmuS} \;\; \leq \; (4 \sigmaXSsc/\sqrt{n}).
\end{eqnarray*}
Further, $\overline{\bSigma}_{n,S} := \E\{ \bXbarS - \bmuS) (\bXbarS - \bmuS)^T \equiv \Cov(\bXbarS - \bmuS)  = n^{-1}\bSigmaS$, so that $\MLtwonorm{\bar{\bSigma}_{n,S}} \equiv n^{-1}\lambdamax(\bSigmaS) \leq n^{-1} \lambdasupSsc$.
Hence, using Lemma \ref{lem:7:matconc} again, this time applied to (a single observation of) $\bXbarS - \bmuS$ for any $S \in \Ssc$, we have: for any $\epsilon \geq 0$ and any $S \in \Ssc$,
\begin{eqnarray*}
&& \P \left\{ \MLtwonorm{\bGammahatS}  \; > \; \frac{\lambdasupSsc}{n} + \frac{16}{n} \KSsc \left( \sqrt{\qn + \epsilon} + \qn + \epsilon \right) \right\} \; \leq \; 2 \exp(-\epsilon),
\nonumber
\end{eqnarray*}
or equivalently, for any $\epsilon \geq 0$ and $S \in \Ssc$,
\begin{eqnarray}
&& \P \left\{ \MLtwonorm{\bGammahatS}  \; > \; \frac{\lambdasupSsc}{n} + \frac{16}{n} \KSsc \left( \sqrt{\epsilon} + \epsilon \right) \right\} \; \leq \; 2 \exp(-\epsilon + \qn).
\label{eq3:pf:mainlem1}
\end{eqnarray}
Consequently, using (\ref{eq3:pf:mainlem1}) along with the union bound, similar to the arguments used earlier for obtaining (\ref{eq2:pf:mainlem1}), we have: for any $\epsilon \geq 0$,
\begin{equation}
\P \left\{ \underset{S \in \Ssc}{\sup} \MLtwonorm{\bGammahatS}  \; > \; \frac{\lambdasupSsc}{n} + \frac{16}{n} \KSsc \left( \sqrt{\epsilon} + \epsilon \right) \right\} \; \leq 2\exp\left( - \epsilon +  \qn + \log \Ln\right). \label{eq4:pf:mainlem1} \\
\end{equation}
Substituting $\epsilon$ in (\ref{eq4:pf:mainlem1}) above as: $\epsilon = (c + 1)(\qn + \log \Ln) \equiv (c+1) \rn$ for any $c \geq 0$, and noting that $\sqrt{c+1} \leq (c+1)$, we then have: $\forall \; c \geq 0$,
\begin{equation}
\P \left\{ \underset{S \in \Ssc} {\sup} \MLtwonorm {\bGammahatS} \; > \; 16 (c+1) \KSsc \left( \frac{\sqrt{\rn}}{n} + \frac{\rn}{n} \right) + \frac{\lambdasupSsc}{n} \right\} \; \leq \; 2 \exp( - c\rn). 
\label{eqfinal2:pf:mainlem1}
\end{equation}
This establishes the second claim (ii) in (\ref{eq:main_lem1:res1}). \qed
\par\medskip
Finally, the third claim in (\ref{eq:main_lem1:res1}) follows from a simple application of the triangle inequality, along with combination (via the union bound) of the bounds in (\ref{eqfinal1:pf:mainlem1}) and (\ref{eqfinal2:pf:mainlem1}) with a slight adjustment applied to (\ref{eqfinal2:pf:mainlem1}). Specifically, since $\bSigmahatS \equiv \bSigmatilS - \bGammahatS \; \forall \; S \in \Ssc$, we have:
\begin{eqnarray*}
\underset{S \in \Ssc}{\sup} \MLtwonorm{\bSigmahatS - \bSigmaS} &\equiv& \underset{S \in \Ssc}{\sup}  \MLtwonorm{(\bSigmatilS - \bSigmaS) + \bGammahatS}  \nonumber \\
 &\leq&  \underset{S \in \Ssc}{\sup} \MLtwonorm{\bSigmatilS - \bSigmaS}  + \underset{S \in \Ssc}{\sup} \MLtwonorm{\bGammahatS}.
\end{eqnarray*}
Hence, combining (\ref{eqfinal1:pf:mainlem1}) and (\ref{eqfinal2:pf:mainlem1}) through the union bound, and simplifying the resulting bound further by noting that $\sqrt{\rn}/n \leq \rn/n$, we then have: 
\begin{eqnarray}
&& \P \left\{ \underset{S \in \Ssc} {\sup} \MLtwonorm{\bSigmahatS - \bSigmaS} \; > \; (c+1) \KSsc \left( \sqrt{\frac{\rn}{n}} + \frac{33 \rn}{n} \right) + \frac{\lambdasupSsc}{n} \right\} \label{eqfinal3:pf:mainlem1} \\
&& \qquad \leq \; 4 \exp( - c\rn) \quad \mbox{for any} \; c \geq 0. \nonumber
\end{eqnarray}
This now establishes the third and final claim (iii) in (\ref{eq:main_lem1:res1}). \qed
\par\medskip
To establish the claims (i) and (ii) in (\ref{eq:main_lem1:res2}), we first recall the definitions of $\bZS$, $\bnuS$ and $\bXiS$ from Assumption \ref{linreg:main_assmpn} (ii) and further, with $\{\bZSi\}_{i=1}^n := \{(Y_i,\bXSi)\}_{i=1}^n$, we define: $\forall \; S \in \Ssc$,
\begin{alignat*}{2}
& \bXitilS \; := \; \frac{1}{n}\sum_{i=1}^n (\bZSi - \bnuS)(\bZSi - \bnuS)^T, \quad &\bZbarS \; := \; \frac{1}{n} \sum_{i=1}^n \bZSi \quad \mbox{and} \nonumber \\
& \bUpshatS \; := \; (\bZbarS - \bnuS)(\bZbarS - \bnuS)^T. &
\end{alignat*}
Then, note that the vectors $(\bSigmatilSY - \bSigmaSY)$ and $\bGammahatSY$ are simply $s \times 1$ sub-matrices of the matrices $\bXitilS$ and $\bUpshatS$, respectively. Hence, using Lemma \ref{lem:5:matbasic} (iii), we deterministically have: for any $S \in \Ssc$,
\begin{equation}
\Ltwonorm{\bSigmatilSY - \bSigmaSY} \; \leq \; \MLtwonorm{\bXitilS - \bXiS} \;\; \mbox{and} \;\; \Ltwonorm{\bGammahatSY} \; \leq \; \MLtwonorm{\bUpshatS}. 
\label{eq5:pf:mainlem1}
\end{equation}

Next, under Assumption \ref{linreg:main_assmpn} (i), note that for any $\bu \equiv (a,\bv) \in \R^{s+1}$ with $a \in \R, \bv \in \R^s$, and for any $S \in \Ssc$, we have: $\psitwonorm{\bu^T(\bZS - \bnuS)} \equiv \psitwonorm{a \Ytil + \bv^T\bXtilS} $ $\leq |a| \sigmaY + \Ltwonorm{\bv}\|\bXtilS\|_{\psitwo}^*$ $\leq \Ltwonorm{\bu}(\sigmaY + \sigmaXSsc)$, where the steps follow through repeated use of Lemma \ref{lem:1:genprop} (i), along with use of Lemma \ref{lem:2:sgvectprop} (i), Assumption \ref{linreg:main_assmpn} (i) and the definition of $\psitwonormstar{\cdot}$ in \ref{def:subgaussvect}. Using Definition \ref{def:subgaussvect} again, we therefore have: $\sup_{S \in \Ssc}\psitwonormstar{\bZS - \bnuS} \leq \sigmatilZSsc \equiv (\sigmaY + \sigmaXSsc)$. \par\medskip
Hence, similar to (\ref{eq1:pf:mainlem1}), applying Lemma \ref{lem:7:matconc} to the random vectors $\{\bZSi - \bnuS\}_{i=1}^n$ for any $S \in \Ssc$, and recalling the definition of the constant $\KtilSsc > 0$ in Assumption \ref{linreg:main_assmpn} (ii) along with the fact that  $\mbox{dim} (\bZS) \leq \qntil := \qn +1$ $\forall \; S \in \Ssc$, it follows that for any $\epsilon \geq 0$,
\begin{equation*}
 \P \left\{ \MLtwonorm {\bXitilS - \bXiS} \; > \; \KtilSsc \left( \sqrt{\epsilon} + \epsilon \right) \right\} \; \leq \; 2 \exp\left( - n \epsilon +  \qntil \right) \;\; \forall \; S \in \Ssc, 
\end{equation*}
and therefore, for any $\epsilon \geq 0$,
\begin{equation}
 \P \left\{ \underset{S \in \Ssc} {\sup} \MLtwonorm {\bXitilS - \bXiS} \; > \; \KtilSsc \left( \sqrt{\epsilon} + \epsilon \right) \right\} \; \leq \; 2\exp\left( - n \epsilon +  \qntil + \log \Ln\right), \label{eq6:pf:mainlem1}
\end{equation}
where the last bound follows from using the union bound, similar to (\ref{eq2:pf:mainlem1}). Consequently, for any $\epsilon \geq 0$,
\begin{equation*}
\P \left\{ \underset{S \in \Ssc} {\sup} \Ltwonorm {\bSigmatilSY - \bSigmaSY} \; > \; \KtilSsc \left( \sqrt{\epsilon} + \epsilon \right) \right\} \; \leq \; 2\exp\left( - n \epsilon +  \qntil + \log \Ln\right), 
\end{equation*}
and hence, for any $c \geq 0$,
\begin{equation}
\P \left\{ \underset{S \in \Ssc} {\sup} \Ltwonorm {\bSigmatilSY - \bSigmaSY} \; > \; (c+1) \KtilSsc \left( \sqrt{\frac{\rntil}{n}} + \frac{\rntil}{n} \right) \right\} \; \leq \; 2 \exp( - c\rntil), 
\label{eqfinal4:pf:mainlem1}
\end{equation}
where the first bound follows from using (\ref{eq5:pf:mainlem1}) and (\ref{eq6:pf:mainlem1}), and the second bound follows from substituting $\epsilon$ as: $\epsilon = (c + 1)(\qntil + \log \Ln)/n \equiv (c+1) \rntil/n$ for any $c \geq 0$, and noting that $\sqrt{c+1} \leq (c+1)$. This therefore establishes the first claim (i) in (\ref{eq:main_lem1:res2}). \qed
\par\medskip
Next, similar to arguments used to prove claim (ii) in (\ref{eq:main_lem1:res1}), it follows using Lemma \ref{lem:3:sgconc}, Lemma \ref{lem:2:sgvectprop} (i)  and the definition of $\psitwonormstar{\cdot}$ in \ref{def:subgaussvect}, that
\begin{equation*}
\psitwonorm{\bu^T(\bZbarS -\bnuS)} \; \leq \; (4 \sigmatilZSsc/\sqrt{n}) \Ltwonorm{\bu}, \;\; \mbox{for any} \; S \in \Ssc \; \mbox{and any} \; \bu \in \R^{s+1},
\end{equation*}
and thus,
\begin{equation*}
\sup_{S \in \Ssc} \psitwonormstar{\bZbarS - \bnuS} \; \leq \; (4 \sigmatilZSsc/\sqrt{n}).
\end{equation*}
Further, $\overline{\bXi}_{n,S} := \E\{ \bZbarS - \bnuS) (\bZbarS - \bnuS)^T \equiv \Cov(\bZbarS - \bnuS)  = n^{-1}\bXiS$, so that $\MLtwonorm{\bar{\bXi}_{n,S}} \equiv n^{-1}\lambdamax(\bXiS) \leq n^{-1} \lambdatilsupSsc$.
Hence, similar to (\ref{eq3:pf:mainlem1}), applying Lemma \ref{lem:7:matconc} to (a single observation of) $\bZbarS - \bmuS$ for any $S \in \Ssc$, we have: for any $\epsilon \geq 0$ and any $S \in \Ssc$,
\begin{equation*}
\P \left\{ \MLtwonorm{\bUpshatS}  \; > \; \frac{\lambdatilsupSsc}{n} + \frac{16}{n} \KtilSsc \left( \sqrt{\epsilon} + \epsilon \right) \right\} \; \leq \; 2 \exp(-\epsilon + \qntil), 
\end{equation*}
and therefore, for any $\epsilon \geq 0$,
\begin{equation}
\P \left\{ \underset{S \in \Ssc}{\sup} \MLtwonorm{\bUpshatS}  \; > \; \frac{\lambdatilsupSsc}{n} + \frac{16}{n} \KtilSsc \left( \sqrt{\epsilon} + \epsilon \right) \right\} \; \leq \; 2\exp\left( - \epsilon +  \qntil + \log \Ln\right), \label{eq7:pf:mainlem1}
\end{equation}
where the last bound follows from using the union bound, similar to (\ref{eq4:pf:mainlem1}). Consequently, for any $\epsilon \geq 0$,
\begin{align*}
& \P \left\{ \underset{S \in \Ssc}{\sup} \Ltwonorm{\bGammahatSY}  \; > \; \frac{\lambdatilsupSsc}{n} + \frac{16}{n} \KtilSsc \left( \sqrt{\epsilon} + \epsilon \right) \right\} \; \leq \; 2\exp\left( - \epsilon +  \qntil + \log \Ln\right), \nonumber 
\end{align*}
and hence, for any $c \geq 0$,
\begin{align}
& \P \left\{ \underset{S \in \Ssc} {\sup} \Ltwonorm {\bGammahatSY} \; > \; 16 (c+1) \KtilSsc \left( \frac{\sqrt{\rntil}}{n} + \frac{\rntil}{n} \right) + \frac{\lambdatilsupSsc}{n} \right\} \; \leq \; 2 \exp( - c\rntil), 
\label{eqfinal5:pf:mainlem1}
\end{align}
where the bounds follow from using (\ref{eq5:pf:mainlem1}) and (\ref{eq7:pf:mainlem1}) and the second bound follows from substituting $\epsilon$ as: $\epsilon = (c + 1)(\qntil + \log \Ln) \equiv (c+1) \rntil$ for any $c \geq 0$, and noting that $\sqrt{c+1} \leq (c+1)$. This therefore establishes the second claim (ii) in (\ref{eq:main_lem1:res2}). \qed
\par\medskip
The remaining claims at the end of Lemma \ref{linreg:main_lem1} are quite straightforward. Using Lemma \ref{lem:2:sgvectprop} (i) and Lemma \ref{lem:1:genprop} (iii), we first note that under Assumption \ref{linreg:main_assmpn}, for any $S \in \Ssc$, $\bv \in \R^s$ and $\bu \in \R^{s+1}$, $\E\{(\bv^T \bXtilS)^2\} \leq 2 \sigmaXSsc^2 \Ltwonorm{\bv}^2$ and $\E\{(\bu^T\bZtilS)^2 \} \leq 2 \sigmatilZSsc^2$. Further, $\lambdamax(\bSigmaS) \equiv \sup_{\Ltwonorm{\bv} \leq 1} \E\{(\bv^T\bXtilS)^2\}$ and $\lambdamax(\bXiS) \equiv \sup_{\Ltwonorm{\bu} \leq 1} \E\{(\bu^T\bZtilS)^2\}$ for each $S \in \Ssc$. Hence, we have: $\sup_{S \in \Ssc} \lambdamax(\bSigmaS) \leq 2 \sigmaXSsc^2$, and $\sup_{S \in \Ssc} \lambdamax(\bXiS) \leq 2 \sigmatilZSsc^2$. This justifies the claimed choices for the constants $\lambdasupSsc$ and $\lambdatilsupSsc$ in Assumption \ref{linreg:main_assmpn}. \qed
\par\medskip
Lastly, owing to the very definition of $\bbetaS$ in (\ref{eq:targetdef}) and the estimating equation satisfied by $\bbetaS$ therein, we have: $\forall \; S \in \Ssc$, $\E\{ \bXtilS (\Ytil - \bXtilS^T\bbetaS)\} = \bzero$ and $\E\{(\Ytil - \bXtilS^T\bbetaS) (\bXtilS^T\bbetaS)\} = 0$, so that $\E(\Ytil^2) = \E\{(\Ytil - \bXtilS^T \bbetaS)^2\} + \E\{ (\bXtilS^T\bbetaS)^2\}$ and hence, for each $S \in \Ssc$,
$$
\E(\Ytil^2) \; \equiv \; \Var(Y) \; \geq \; \E\{ (\bXtilS^T \bbetaS)^2\} \; \equiv \; \bbetaS^T \bSigmaS \bbetaS \; \geq \; \Ltwonorm{\bbetaS}^2 \lambdamin(\bSigmaS). 
$$
Using (\ref{linreg:main_assmpn:eq1}), we therefore have: $\sup_{S \in \Ssc} \Ltwonorm{\bbetaS}^2 \leq \lambdainfSsc^{-1} \Var(Y)$. Further, due to Lemma \ref{lem:1:genprop} (iii), $\Var(Y) \leq 2 \sigmaY^2$ and thus $\sup_{S \in \Ssc} \Ltwonorm{\bbetaS}^2 \leq \lambdainfSsc^{-1} \Var(Y) \leq 2 \lambdainfSsc^{-1} \sigmaY^2$. This establishes the final claim in Lemma \ref{linreg:main_lem1}. The proof of Lemma \ref{linreg:main_lem1} is now complete. \qed

\subsection{Proof of Lemma \ref{linreg:main_lem2}}\label{proof:main:lem2}

For for any $c > 0$ and any constant $c^* > 0$ satisfying (\ref{eq:cstarcond}), let us define the events:

\begin{align}
& \AscnS(c)   \;\; :=  \; \left\{ \sup_{S\in \Ssc} \MLtwonorm{\bSigmahatS - \bSigmaS}  > \; (c+1)\KSsc \left( \sqrt{\frac{\rn}{n}} + \frac{33 \rn}{n} \right) + \frac{\lambdasupSsc}{n} \right\}, \label{eq1:pf:mainlem2} \\
& \AscnS(c^*)  := \;  \left\{ \sup_{S\in \Ssc} \MLtwonorm{\bSigmahatS - \bSigmaS}  > \; (c^*+1)\KSsc \left( \sqrt{\frac{\rn}{n}} + \frac{33 \rn}{n} \right) + \frac{\lambdasupSsc}{n} \right\}, \nonumber
\end{align}
\begin{align*}
& \BscnS(c)  :=  \left\{ \sup_{S\in \Ssc} \MLtwonorm{\bSigmahatS^{-1} - \bSigmaS^{-1}} >  (c+1)\KSsc^* \left( \sqrt{\frac{\rn}{n}} + \frac{33 \rn}{n} \right) + \frac{2}{n}\frac{\lambdasupSsc}{\lambdainfSsc^2} \right\} \nonumber
\end{align*}
and let $\AscnS^c(c^*)$ denote the complement event of $\AscnS(c^*)$. Then, for any $c > 0$ and for any $c^* > 0$ satisfying (\ref{eq:cstarcond}), we first note that
\begin{equation}
\P\{ \AscnS(c)\} \; \leq \;  4 \exp(-c \rn), \;\; \mbox{and} \;\; \P\{ \AscnS(c^*) \} \; \leq \; 4 \exp(-c^* \rn), \label{eq2:pf:mainlem2}
\end{equation}
where both bounds are direct consequences of Lemma \ref{linreg:main_lem1} (iii), which applies under Assumption \ref{linreg:main_assmpn} (i). Further, for the events defined in (\ref{eq1:pf:mainlem2}), the following inclusions hold:
\begin{eqnarray}
&& \BscnS(c) \cap \AscnS^c(c^*) \; \subseteq_{(a)}\; \BscnS(c) \cap \left\{ \sup_{S\in \Ssc} \MLtwonorm{\bSigmahatS - \bSigmaS} \; \leq \; \frac{\lambdainfSsc}{2} \right\}  \nonumber \\
&& \qquad\quad \subseteq_{(b)} \; \AscnS(c) \cap \left\{ \sup_{S\in \Ssc} \MLtwonorm{\bSigmahatS - \bSigmaS} \; \leq \; \frac{\lambdainfSsc}{2} \right\}   \quad \subseteq \; \AscnS(c), \label{eq3:pf:mainlem2}
\end{eqnarray}
and hence, $\P\{\BscnS(c) \cap \AscnS^c(c^*) \} \leq \P\{ \AscnS(c)\}$. The inclusion (a) in (\ref{eq3:pf:mainlem2}) above follows since $c^*$ satisfies the condition (\ref{eq:cstarcond}) in Lemma \ref{linreg:main_lem2}, while the inclusion (b) follows from an application of Lemma \ref{lem:6:matinv} and from noting the definitions of the constants $\KSsc^*$ and $\lambdainfSsc$.
\par\medskip
Hence, for any $c > 0$ and for any $c^* > 0$ satisfying (\ref{eq:cstarcond}), we then have:
\begin{eqnarray*}
&& \P \left\{ \sup_{S\in \Ssc} \MLtwonorm{\bSigmahatS^{-1} - \bSigmaS^{-1}} \; > \; (c+1)\KSsc^* \left( \sqrt{\frac{\rn}{n}} + \frac{33 \rn}{n} \right) + \frac{2}{n}\frac{\lambdasupSsc}{\lambdainfSsc^2} \right\} \nonumber \\
&& \quad \equiv \; \P\left\{ \BscnS(c)\right\} \; = \; \P\left\{ \BscnS(c) \cap \AscnS^c(c^*)\right\} + \P\left\{ \BscnS(c) \cap \AscnS(c^*)\right\} \nonumber \\
&& \quad \leq_{(a)} \; \P\left\{\AscnS(c)\right\}  + \P\left\{\AscnS(c^*) \right\} \nonumber \\
&& \quad \leq_{(b)} \; 4 \exp(-c\rn) + 4 \exp(-c^* \rn) \;\; \equiv \; 4 a_n(c,c^*, \rn), \nonumber
\end{eqnarray*}
where the inequalities (a) and (b) follow from using (\ref{eq3:pf:mainlem2}) and (\ref{eq2:pf:mainlem2}), respectively. This establishes the claim in Lemma \ref{linreg:main_lem2} and completes the proof. \qed

\subsection{Proof of Theorem \ref{theorem: asymptotic linearity}}

\begin{proofof}[the first part of Theorem \ref{theorem: asymptotic linearity}]
First, note that
\[
\mathrm{aver}(\hat{\Theta}_{jp}(\hat{\mathcal{C}}_{0}')) - \mathrm{aver}(\Theta_{jp}(\mathcal{C}_{0}')) - E_{n,jp}(\hat{\mathcal{C}}_{0}', \mathcal{C}_{0}')
~=~  \mathrm{aver}(\hat{\Theta}_{jp}(\mathcal{C}_{0}')) - \mathrm{aver}(\Theta_{jp}(\mathcal{C}_{0}')).
\]
Next, let $\mathbf{X}_{S_{j1}'},\ldots,\mathbf{X}_{S_{jL_{\mathrm{{distinct}, j}}}'} \subseteq \mathbf{Adj}_{\mathcal{C}_0'}(X_j)$ be the distinct parent sets of $X_j$ in $\MEC(\mathcal{C}_0')$ with multiplicities $m_{j1},\ldots,m_{jL_{\mathrm{distinct}, j}}$ respectively, so that
\begin{align*}
\mathrm{aver}(\hat{\Theta}_{jp}(\mathcal{C}_{0}')) & \; = \; \frac{1}{L} \sum_{r = 1}^{L_{\mathrm{{distinct}, j}}} m_r~ \hat{\beta}_{jp|\mathbf{X}_{S_{jr}'} \cup \mathbf{X}_{\{1,\ldots,t\}}}\quad \text{and} \\
\mathrm{aver}(\Theta_{jp}(\mathcal{C}_{0}')) & \; = \; \frac{1}{L} \sum_{r = 1}^{L_{\mathrm{{distinct}, j}}} m_r~ \beta_{jp|\mathbf{X}_{S_{jr}'} \cup \mathbf{X}_{\{1,\ldots,t\}}}, \;\;\; \mbox{where} \;\; L \; := \sum_{r = 1}^{L_{\mathrm{{distinct}, j}}} m_r.
\end{align*}
These now enable us to apply the results obtained in Section \ref{subsec:linreg}. To see this, note that the assumptions of Theorem \ref{linreg:main_thm1} for
\[
\text{$Y := X_p$ \;\; and \;\; $\mathcal{S} \; := \; \{S_{jr} = (j, 1, \ldots, t, S_{jr}') : r \in \{1,\ldots,L_{\mathrm{distinct},j}\}\}$}
\] follow from Assumptions \ref{assumption: LSEM}, \ref{assumption: high-dim}, \ref{assumption: sparsity} and \ref{assumption: eigenvalue condition} and Lemma \ref{lemma: bounded psi2 norm}. Thus, it follows directly from the first part of Remark \ref{remark: linear functionals} with $\mathbf{a}_{S_{jr}}^T =  ~(m_r/L, 0, \ldots, 0)$ that
\[
\mathrm{aver}(\hat{\Theta}_{jp}(\mathcal{C}_{0}')) - \mathrm{aver}(\Theta_{jp}(\mathcal{C}_{0}')) \; = \; \frac{1}{n}\sum_{r = 1}^{n} Z_{jp}^{(r)} + O_{\Prob}\left(\frac{q_{j} + \log(L_{\mathrm{distinct}, j})}{n}\right),
\]
since
$$\sum_{r = 1}^{L_{\mathrm{{distinct}, j}}} ||\mathbf{a}_{S_{jr}}||_2 \; = \; \sum_{r=1}^{L_{\mathrm{{distinct}, j}}} \frac{m_r}{L} \; = \; 1. \qed$$
\end{proofof}

\begin{proofof}[the second part of Theorem \ref{theorem: asymptotic linearity}]
Recall that
\[
\hat{\eta}(\hat{\mathcal{C}}_0') - \eta(\mathcal{C}_0') ~=~ \hat{\theta}_{ij}~ \mathrm{aver}(\hat{\Theta}_{jp}(\hat{\mathcal{C}}_{0}')) - \theta_{tj}~ \mathrm{aver}(\Theta_{jp}(\mathcal{C}_{0}')).
\]
Therefore, it is straightforward to obtain the result by applying the identity $a_nb_n - ab = a(b_n - b) + b(a_n - a) + (a_n - a)(b_n - b)$ with $a_n = \hat{\theta}_{ij}$, $b_n = \mathrm{aver}(\hat{\Theta}_{jp}(\hat{\mathcal{C}}_{0}'))$, $a = \theta_{tj}$, and $b = \mathrm{aver}(\Theta_{jp}(\mathcal{C}_{0}'))$, and using the first part of Theorem \ref{theorem: asymptotic linearity} and the following well-known result from the asymptotic theory \tcr{of} multiple linear regression \tcr{(e.g., see \citet{VanderVaart98})}
\begin{align*}
\hat{\theta}_{tj} - \theta_{tj} ~=~ \hat{\beta}_{tj | \mathbf{X}_{\{1,\ldots,t-1\}}} - \beta_{tj | \mathbf{X}_{\{1,\ldots,t-1\}}} ~=~ \frac{1}{n} \sum_{r = 1}^{n}Z_{tj}^{(r)} ~ + ~O_{\Prob}\left(\frac{1}{n}\right). \qed
\end{align*}
\end{proofof}

\subsection{Proof of Corollary \ref{corollary: asymptotic normality of IDA}}
Recall that
\begin{align*}
&~\tcr{\mathrm{aver}(\hat{\Theta}_{jp}(\hat{\mathcal{C}}_{0}')) ~-~ \mathrm{aver}(\Theta_{jp}(\mathcal{C}_{0}')) 
 ~= 
 ~E_{n,jp}(\hat{\mathcal{C}}_{0}', \mathcal{C}_{0}') ~+~
 \mathrm{aver}(\hat{\Theta}_{jp}(\mathcal{C}_{0}')) - \mathrm{aver}(\Theta_{jp}(\mathcal{C}_{0}')).}
\end{align*}
Since from the discussion before Theorem \ref{theorem: asymptotic linearity} (see Section \ref{section: asymptotic distribution}) we have
$$E_{n,jp}(\hat{\mathcal{C}}_{0}', \mathcal{C}_{0}')  ~=~ o_{\Prob}(1/ \sqrt{n}),$$ it is sufficient to show $\Exp[Z_{jp}^2] = \Omega(1)$ and
\begin{align}\label{eq: modified asymp normality}
\frac{\sqrt{n}\left\{ \mathrm{aver}(\hat{\Theta}_{jp}(\mathcal{C}_{0}')) - \mathrm{aver}(\Theta_{jp}(\mathcal{C}_{0}')) \right\}}{\sqrt{\Exp[Z_{jp}^2]}}~ \; \overset{d}{\longrightarrow} ~\;\mathcal{N}(0,1).
\end{align}

As discussed in the proof of Theorem \ref{theorem: asymptotic linearity}, $ \mathrm{aver}(\hat{\Theta}_{jp}(\mathcal{C}_{0}')) - \mathrm{aver}(\Theta_{jp}(\mathcal{C}_{0}'))$ can be written as $\sum_{S \in \mathcal{S}} \mathbf{a}_S^T (\hat{\boldsymbol{\beta}}_S - \boldsymbol{\beta}_S)$, \tcr{following the notations and setup of Remark \ref{remark: linear functionals}}, for some set \tcr{of vectors} $\{\mathbf{a}_S : S \in \mathcal{S} \}$ satisfying $\sum_S ||\mathbf{a}_S||_2 = O(1)$, where $\boldsymbol{\beta}_S$ denotes the vector of regression coefficients in the regression of $Y := X_p$ on $\mathbf{X}_S$ \tcr{and} $\hat{\boldsymbol{\beta}}_S$ denotes its sample version (i.e. the corresponding OLS estimator) respectively, and $\mathcal{S} := \{S_{jr} = (j, 1,\ldots, t, S_{jr}') : r \in \{1,\ldots,L_{\mathrm{distinct},j}\}\}$ is as in the proof of Theorem \ref{theorem: asymptotic linearity}.

\par\smallskip
Therefore, given the stronger sparsity assumption \tcr{(from Assumption \ref{assumption: conditional variance})}
$$
n^{-1/2}\{q_{j}+ \log(L_{\mathrm{distinct}, j})\} ~\longrightarrow~ 0,
$$ \eqref{eq: modified asymp normality} \tcr{now} follows from the second part of Remark \ref{remark: linear functionals} as long as the second moment of the influence function in the asymptotic linear expansion of $\mathrm{aver}(\hat{\Theta}_{jp}(\mathcal{C}_{0}')) - \mathrm{aver}(\Theta_{jp}(\mathcal{C}_{0}'))$ is bounded below, i.e.\ $\Exp[Z_{jp}^2] = \Omega(1)$. \qed

\par\medskip
We prove $\Exp[Z_{jp}^2] = \Omega(1)$ \tcr{by verifying the sufficient conditions given in the last paragraph of \tcR{Appendix} 
\ref{sec:remarkverify} (regarding the two moment conditions required in Remark \ref{remark: linear functionals}), namely:} 
there exist constants $c_1 > 0$ and $c_2 > 0$ such that
\begin{align}\label{eq: sufficient lower bound condition}
\Var(Y \mid \cup_{S \in \mathcal{S}} \mathbf{X}_{S}) ~>~ c_1 ~~ \text{and} ~~ \Exp\left[\left(\sum_{S \in \mathcal{S}} \mathbf{a}_S^T\left(\Sigma_{SS}\right)^{-1}(\mathbf{X}_S - \boldsymbol{\mu}_S) \right)^2 \right] ~>~ c_2.
\end{align}


The first part of Assumption \ref{assumption: conditional variance} and the first part of \eqref{eq: sufficient lower bound condition} with $c_1 = v$ are identical, since $\cup_{S \in \mathcal{S}}~ \mathbf{X}_S = \mathbf{Adj}_{\mathcal{C}_0'}(X_j) \cup \{X_1,\ldots,X_t, X_j\}$. Thus we complete the proof by showing that the second part of \eqref{eq: sufficient lower bound condition} follows from Assumptions \ref{assumption: eigenvalue condition} and \ref{assumption: conditional variance}.

\par\smallskip
To this end, following the notation in the proof of Theorem \ref{theorem: asymptotic linearity}, we write
\begin{align*}
&~\sum_{S \in \mathcal{S}} \mathbf{a}_S^T(\Sigma_{SS})^{-1}(\mathbf{X}_S - \boldsymbol{\mu}_S) \\
= &~ \sum_{r = 1}^{L_{\mathrm{distinct},j}} \frac{m_r}{L}~ e_{1, |S_{jr}'| + t + 1}^T \left((\Sigma_0)_{(j, 1, \ldots, t, S_{jr}')(j, 1, \ldots, t, S_{jr}')}\right)^{-1} \\
& \qquad \qquad \qquad \qquad \qquad \qquad \qquad \qquad  \left(\mathbf{X}_{(j, 1, \ldots, t, S_{jr}')} - \boldsymbol{\mu}_{(j, 1, \ldots, t, S_{jr}')} \right),
\end{align*}
where $e_{1, |S_{jr}'| + t + 1}^T$ denote the first row of an $(|S_{jr}'| + t + 1) \times (|S_{jr}'| + t + 1)$ identity matrix.

By partitioning $(\Sigma_0)_{(j, 1, \ldots, t, S_{jr}')(j, 1, \ldots, t, S_{jr}')}$ as
\[
(\Sigma_0)_{(j, 1, \ldots, t, S_{jr}')(j, 1, \ldots, t, S_{jr}')}
~=~ \left( \begin{array}{ll} (\Sigma_0)_{jj} & (\Sigma_0)_{j(1, \ldots, t, S_{jr}')} \\ (\Sigma_0)_{(1, \ldots, t, S_{jr}')j} & (\Sigma_0)_{(1, \ldots, t, S_{jr}')} \end{array} \right),
\]
and applying the well-known formula for the inverse of a partitioned matrix, we obtain
\begin{align*}
&~e_{1, |S_{jr}'| + t + 1}^T \left((\Sigma_0)_{(j, 1, \ldots, t, S_{jr}') (1, \ldots, t, S_{jr}')}\right)^{-1} \left(\mathbf{X}_{(j, 1, \ldots, t, S_{jr}')} - \boldsymbol{\mu}_{(j, 1, \ldots, t, S_{jr}')} \right) \\
= &~ \frac{(X_j - \mu_j) - (\Sigma_0)_{j(1, \ldots, t, S_{jr}')} \left((\Sigma_0)_{(1, \ldots, t, S_{jr}') (1, \ldots, t, S_{jr}')}\right)^{-1} \left(\mathbf{X}_{(1, \ldots, t, S_{rj}')} - \boldsymbol{\mu}_{(1, \ldots, t, S_{rj}')}\right)}{(\Sigma_0)_{jj} - (\Sigma_0)_{j(1, \ldots, t, S_{jr}')} \left((\Sigma_0)_{(1, \ldots, t, S_{jr}') (1, \ldots, t, S_{jr}')}\right)^{-1} (\Sigma_0)_{(1, \ldots, t, S_{jr}')j}}.
\end{align*}
We define
\begin{align*}
\boldsymbol{\beta}_{j|(1, \ldots, t, S_{jr}')}^T &~:=~ (\Sigma_0)_{j(1, \ldots, t, S_{jr}')} \left((\Sigma_0)_{(1, \ldots, t, S_{jr}')\tcm{(1, \ldots, t, S_{jr}')}}\right)^{-1}, \quad \text{and} \\
\sigma_{j | (1, \ldots, t, S_{jr}')}^2 &~:=~ (\Sigma_0)_{jj} - (\Sigma_0)_{j(1, \ldots, t, S_{jr}')} \left((\Sigma_0)_{(1, \ldots, t, S_{jr}')\tcm{(1, \ldots, t, S_{jr}')}}\right)^{-1} (\Sigma_0)_{(1, \ldots, t, S_{jr}')j}.
\end{align*}
Therefore,
\begin{align*}
&~\Exp\left[\left(\sum_{S \in \mathcal{S}} \mathbf{a}_S^T\left(\Sigma_{SS}\right)^{-1}(\mathbf{X}_S - \boldsymbol{\mu}_S) \right)^2 \right] \\
=&~ \Exp\left[\left( \sum_{r = 1}^{L_{\mathrm{distinct},j}} \frac{m_r}{L}~ e_{1, |S_{jr}'| + t + 1}^T \left((\Sigma_0)_{(j, 1, \ldots, t, S_{jr}')(j, 1, \ldots, t, S_{jr}')}\right)^{-1} \left(\mathbf{X}_{(j, 1, \ldots, t, S_{jr}')} - \boldsymbol{\mu}_{(j, 1, \ldots, t, S_{jr}')} \right) \right)^2 \right] \\
=&~ \Exp\left[\left( \sum_{r = 1}^{L_{\mathrm{distinct},j}} \frac{m_r}{L~\sigma_{j | (1, \ldots, t, S_{jr}')}^2}~ \left( (X_j - \mu_j) - \boldsymbol{\beta}_{j|(1, \ldots, t, S_{jr}')}^T (\mathbf{X}_{(1, \ldots, t, S_{rj}')} - \boldsymbol{\mu}_{(1, \ldots, t, S_{rj}')})\right) \right)^2 \right] 
\end{align*}    
\begin{align*}
=&~ \Var\left( \sum_{r = 1}^{L_{\mathrm{distinct},j}} \frac{m_r~ \left(X_j - \boldsymbol{\beta}_{j|(1, \ldots, t, S_{jr}')}^T\mathbf{X}_{(1, \ldots, t, S_{rj}')}\right)}{L~\sigma_{j | (1, \ldots, t, S_{jr}')}^2} \right) \\
\geq &~ \Exp\left[ \Var\left( \left\{\sum_{r = 1}^{L_{\mathrm{distinct},j}} \frac{m_r~ \left(X_j - \boldsymbol{\beta}_{j|(1, \ldots, t, S_{jr}')}^T\mathbf{X}_{(1, \ldots, t, S_{rj}')}\right)}{L~\sigma_{j | (1, \ldots, t, S_{jr}')}^2} \right\} \mid \mathbf{Adj}_{\mathcal{C}_0'}(X_j) \cup \mathbf{X}_{\{1,\ldots,t\}} \right) \right] \\
\;\; = & ~ \Exp\left[ \Var\left( \left\{\sum_{r = 1}^{L_{\mathrm{distinct},j}} \frac{m_r X_j}{L~\sigma_{j | (1, \ldots, t, S_{jr}')}^2} \right\} \mid \mathbf{Adj}_{\mathcal{C}_0'}(X_j) \cup \mathbf{X}_{\{1,\ldots,t\}} \right) \right] \\
\;\; = & ~ \left\{\sum_{r = 1}^{L_{\mathrm{distinct},j}} \frac{m_r}{L~\sigma_{j | (1, \ldots, t, S_{jr}')}^2} \right\}^2 \Exp\left[ \Var\left(X_j \mid  \mathbf{Adj}_{\mathcal{C}_0'}(X_j) \cup \mathbf{X}_{\{1,\ldots,t\}} \right) \right],
\end{align*}
where the second last equality follows from the fact that $\mathbf{X}_{S_{jr}'} \subseteq \mathbf{Adj}_{\mathcal{C}_0'}(X_j)$ for all $r \in \{1,\ldots, L_{\mathrm{distinct}, j}\}$.

Now since $\sum_r m_r = L$ and $\sigma_{j | (1, \ldots, t, S_{jr}')}^2 \leq (\Sigma_0)_{j\tcr{j}}$ for all $r \in \{1,\ldots, L_{\mathrm{distinct}, j}\}$, by Assumption \ref{assumption: eigenvalue condition}, we have
\[
\left\{\sum_{r = 1}^{L_{\mathrm{distinct},j}} \frac{m_r}{L~\sigma_{j | (1, \ldots, t, S_{jr}')}^2} \right\}^2 ~\geq~ \frac{1}{C_3^2}.
\]
Finally, by Assumption \ref{assumption: conditional variance}, we have
\[
\Exp\left[ \Var\left(X_j \mid  \mathbf{Adj}_{\mathcal{C}_0'}(X_j) \cup \mathbf{X}_{\{1,\ldots,t\}} \right) \right] ~\geq~ v.
\]
This \tcr{therefore establishes \eqref{eq: sufficient lower bound condition} and} completes \tcr{the proof of Corollary \ref{corollary: asymptotic normality of IDA}.} \qed

\subsection{Proof of Corollary \ref{corollary: asymptotic distribution of MIDA}}
It is easy to see that Theorem \ref{theorem: asymptotic linearity}, Corollary \ref{corollary: asymptotic normality of IDA} and Assumption \ref{assumption: stronger sparsity} imply
\begin{align*}
&~ \hat{\eta}(\hat{\mathcal{C}}_0') - \eta(\mathcal{C}_0')  ~=  ~\hat{\theta}_{tj}~E_{n,jp}(\hat{\mathcal{C}}_{0}', \mathcal{C}_{0}') ~+~ \\
  & \qquad \qquad \qquad \quad \;\; ~+~ \frac{1}{n}\sum_{r = 1}^{n} \left\{ \theta_{tj}Z_{jp}^{(r)} ~+~  \mathrm{aver}(\Theta_{jp}(\mathcal{C}_{0}')) Z_{tj}^{(r)} \right\} ~+~ o_{\Prob}\left( \frac{1}{\sqrt{n}}\right).
\end{align*}
Thus the result follows from the same arguments given in the proof of Corollary \ref{corollary: asymptotic normality of IDA}. \tcr{This completes the proof of Corollary \ref{corollary: asymptotic distribution of MIDA}.} \qed

\subsection{Proof of Theorem \ref{theorem: asymptotic distribution}}
%
First, consider the case $\theta_{jt} = \mathrm{aver}(\Theta_{jp}(\mathcal{C}_{0}')) = 0$. Then $\tilde{T}_{n,jp}$ can be written as
\begin{align*}
\tilde{T}_{n,jp} \; := \; \frac{W_{n,1j}~ W_{n,jp}}{\sqrt{W_{n,1j}^2 + W_{n,jp}^2 + 2~ \rho~W_{n,1j}~ W_{n, jp}}},
\end{align*}
where $ \rho$ is the correlation coefficient between $Z_{jp}$ and $Z_{tj}$, and
\[
(W_{n, 1j}, W_{n, jp})^T \; := \; \left(\frac{\sqrt{n}~ \hat{\theta}_{tj} }{\sqrt{\Exp[Z_{tj}^2]}}, \frac{\sqrt{n} ~\mathrm{aver}(\hat{\Theta}_{jp}(\mathcal{C}_{0}'))}{\sqrt{\Exp[Z_{jp}^2]}} \right)^T.
\]

By following similar arguments as in Theorem \ref{theorem: asymptotic linearity} and in Corollary \ref{corollary: asymptotic distribution of MIDA} it can be shown that
\[
(W_{n, 1j}, W_{n, jp})^T ~ \overset{d}{\longrightarrow} ~ (W_1, W_2)^T ~ \sim ~\mathcal{N}\left( \mathbf{0},~ \left(\begin{array}{cc} 1 & \rho\\ \rho & 1 \end{array} \right) \right).
\]

Therefore, by the continuous mapping theorem,  $\tilde{T}_{n,jp}$ is asymptotically distributed as
\begin{align}\label{eq: 0-0 case}
W \; := \; \frac{W_1W_2}{\sqrt{W_1^2 + W_2^2 + 2 \rho W_1 W_2}}.
\end{align}

Next, suppose at least one of $\theta_{jt}$ or $\mathrm{aver}(\Theta_{jp}(\mathcal{C}_{0}'))$ is non-zero. It follows from Theorem \ref{theorem: high-dimensional consistency} that the denominator of $\tilde{T}_{n,jp}$ converges in probability to the denominator of $T_{n,jp}$. Therefore, $\tilde{T}_{n,jp} \overset{d}{\longrightarrow} \mathcal{N}(0, 1)$ follows from Corollary \ref{corollary: asymptotic distribution of MIDA} and \tcr{an application of} Slutsky's theorem. \qed

\section{Technical Tools - \tcr{Definitions and} Supporting Lemmas}\label{sec:technicaltools}   


In this section, we collect \tcr{some definitions referred in the main paper} and some key technical lemmas that \tcr{will} be useful \tcr{in the proofs of all our main results.}

\subsection{\tcr{Sub-Gaussians and Sub-Exponentials}}\label{subsec:orlicznorms}


\tcr{Here,} we formally define sub-Gaussian and sub-exponential random variables (and vectors), \tcr{used in some of our assumptions,} based on the \tcr{concept} of (exponential) Orlicz norms. 


\begin{definition}[The $\psialpha$-Orlicz norm, and sub-Gaussian and sub-exponential variables]\label{def:psialnorm}
For any $\alpha > 0$, define the function $\psialpha(u)  := \exp(u^\alpha) - 1$ $\forall \; u \geq 0$. For any random variable $X$ and any $\alpha > 0$, the \emph{$\psialpha$-Orlicz norm} (the exponential Orlicz norm of order $\alpha > 0$) of $X$ is then defined as:
\begin{equation*}
\psialphanorm{X} \; := \; \inf \{ c > 0: \; \E\{ \psialpha(|X|/c) \} \; \leq \; 1 \}, 
\end{equation*}
where $\psialphanorm{X}$ is understood \tcr{to} be $\infty$ if the infimum above is over an empty set. The special cases of
$\alpha = 2$ and $\alpha = 1$ correspond to \tcr{the} 
sub-Gaussian and sub-exponential random variables, respectively. $X$ is said to be 
sub-Gaussian if $\psitwonorm{X} < \infty$  (and $\psitwonorm{X}$ is 
\tcr{called} the `sub-Gaussian norm' of $X$), and $X$ is 
sub-exponential if $\psionenorm{X} < \infty$ \tcr{(and $\psionenorm{X}$ is its `sub-exponential norm').}
\end{definition}

\begin{definition}[Sub-Gaussian norm(s) for random vectors]\label{def:subgaussvect}
A random vector \tcr{$\bX = (X_j)_{j=1}^d \in \R^d$} $(d \geq 1)$ is defined to be sub-Gaussian if and only if \tcr{for all} $\bv \in \R^d$, $\bv^T\bX$ is sub-Gaussian, as in Definition \ref{def:psialnorm}. For such random vectors, we define two sub-Gaussian norms as follows: 
\begin{equation*}
\psitwonorm{\bX} \; := \; \underset{1 \leq j \leq d}{\max} \psitwonorm{\tcr{X_j}} \;\; \mbox{and} \;\; \psitwonormstar{\bX} \; := \; \underset{\|\bv \|_2 \leq 1}{\sup} \psitwonorm{\bv^T\bX}.
\end{equation*}
For a general $\alpha > 0$, we also define, analogous to $\psitwonorm{\bX}$, the $\psialpha$-Orlicz norm of a random vector $\bX \in \R^d$ as: $\psialphanorm{\bX} \; := \; \underset{1 \leq j \leq d}{\max} \psialphanorm{\tcr{X_j}}$. 
\end{definition}
For most of our analyses, we use $\psitwonormstar{\cdot}$ as the vector sub-Gaussian norm which has usually been the accepted definition \citep[e.g.]{Vershynin12, Vershynin18}.
The corresponding extension of the $\psitwonormstar{\bX}$ norm to a $\psialphanorm{\cdot}^*$ norm for a general $\alpha > 0$ is however not immediate, and certainly not standard in the literature.

\subsection{Properties of Orlicz Norms and Concentration Bounds}\label{subsec:probbound}

We next enlist, through a sequence of lemmas, some useful general properties of Orlicz norms, as well as a few specific ones for sub-Gaussians and sub-exponentials. These are all quite well known and routinely used. Their statements (possibly with slightly different constants) and proofs can be found in several relevant references, including \citet{VanderVaartWellner96:EmpiricalProcessesBook, Pollard15, Vershynin12, Vershynin18, Wainwright_Book_2019} and \citet{Rigollet17}, among others. The proofs are \tcr{therefore skipped here} for brevity.

\begin{lemma}[General properties of Orlicz norms, sub-Gaussians and sub-exponentials]\label{lem:1:genprop}
In the following, $X, Y \in \R$ denote generic random variables and $\mu$ denotes $\E(X) \in \R$.
\begin{enumerate}[(i)]
\item \emph{(Basic properties).} For $\alpha \geq 1$, $\psialphanorm{\cdot}$ is a norm  (and a quasinorm if $\alpha < 1$) satisfying: (a) $\psialphanorm{X} \geq 0$ and $\psialphanorm{X} = 0 \Leftrightarrow X = 0$ almost surely (a.s.), (b) $\psialphanorm{c X} = |c| \psitwonorm{X}$ $\forall \; c \in \R$ and $\psialphanorm{|X|} = \psialphanorm{X}$, and (c) $\psialphanorm{X + Y} \leq \psialphanorm{X} + \psialphanorm{Y}$.
\par\smallskip

\item \emph{(Tail bounds and equivalences).} (a) If $\psialphanorm{X} \leq \sigma$ for some $(\alpha, \sigma) > 0$, then $\forall \; \epsilon \geq 0$, $\P(|X| > \epsilon) \leq 2 \exp(- \epsilon^{\alpha}/ \sigma^{\alpha})$. (b) Conversely, if $\P(|X| > \epsilon) \leq C \exp(- \epsilon^\alpha/ \sigma^\alpha) \; \forall \; \epsilon \geq 0$, for some $(C, \sigma, \alpha) > 0$, then $\psialphanorm{X} \leq \sigma (1 + C/2)^{1/\alpha}$.
\par\smallskip

\item \emph{(Moment bounds).} If $\psialphanorm{X} \leq \sigma$ for some $(\alpha, \sigma) > 0$, then $\E(|X|^m) \leq C_{\alpha}^m \sigma^m m^{m/\alpha}$ $\forall \; m \geq 1$, for some constant $C_{\alpha}$ depending only on $\alpha$. (A converse of this result also holds, although not explicitly presented here). For $\alpha = 1$ and $2$ in particular, we have:
    \begin{enumerate}[(a)]
    \item If $\psionenorm{X} \leq \sigma$, then for each $m \geq 1$, $\E(|X|^m) \; \leq \; \sigma^m  m! \; \leq \; \sigma^m m^m $.
    \item If $\psitwonorm{X} \leq \sigma$, then $\E(|X|^m) \leq 2 \sigma^m \Gamma(m/2 + 1)$ for each $m \geq 1$, where $\Gamma(a) := \int_0^{\infty} x^{a - 1} exp(-x) dx $ $\forall \; a > 0$ denotes the Gamma function. Hence, $\E(|X|) \leq \sigma \sqrt{\pi}$ and $\E(|X|^m) \leq 2 \sigma^m (m/2)^{m/2}$ for any $m \geq 2$.
    \end{enumerate}

\item \emph{(H\"{o}lder-type inequality for the Orlicz norm of products).} For any $\alpha, \beta > 0$, let $\gamma := (\alpha^{-1} + \beta^{-1})^{-1}$. Then, for any two random variables $X$ and $Y$ with $\psialphanorm{X} < \infty$ and $\psibetanorm{Y} < \infty$, $\psigammanorm{XY} < \infty$ and $\psigammanorm{XY} \leq  \psialphanorm{X} \psibetanorm{Y}$. In particular, for any two sub-Gaussians $X$ and $Y$, $XY$ is sub-exponential and $\psionenorm{XY} \leq \psitwonorm{X} \psitwonorm{Y}$. Moreover, if $Y \leq M$ a.s. and $\psialphanorm{X} < \infty$, then $\psialphanorm{XY} \leq M \psialphanorm{X}$.

\par\smallskip

\item \emph{(MGF related properties of sub-Gaussians).} Let $\E[\exp\{t(X-\mu)\}]$ denote the moment generating function (MGF) of $X - \mu$ at $t \in \R$. Then:
\begin{enumerate}[(a)]
\item If $\psitwonorm{X - \mu} < \sigma$ for some $\sigma > 0$, then $\E[\exp\{t(X-\mu)\}] \leq \exp(2\sigma^2 t^2)$ $\forall \; t \in \R$.
\item Conversely, if $\E[\exp\{t(X-\mu)\}] \leq \exp(\sigma^2 t^2)$  $\forall \; t \in \R$  for some $\sigma \geq 0$, then for any $\epsilon \geq 0$, $\P(| X - \mu | > \epsilon) \leq 2 \exp(- \epsilon^2/ 4 \sigma^2)$ and hence, $\psitwonorm{X-\mu} \leq 2\sqrt{2} \sigma$.
    \end{enumerate}
\end{enumerate}
\end{lemma}

\begin{lemma}[Properties of sub-Gaussian random vectors]\label{lem:2:sgvectprop}
Let \tcr{$\bX = (X_j)_{j=1}^d \in \R^d$} be any random vector, 
and let $\bv \in \R^d$ and $\bM \in \R^{d \times d}$ denote any generic (fixed) vectors and matrices, \tcr{for any $d \geq 1$.} Then,
\begin{enumerate}[(i)]
\item For any $\bv \in \R^d$, $\psitwonorm{\bv^T\bX} \leq \| \bv \|_2 \psitwonormstar{\bX}$ and $\psitwonorm{\bv^T\bX} \leq \| \bv \|_1 \psitwonorm{\bX}$ $\leq \sqrt{d} \Ltwonorm{\bv} \psitwonorm{\bX}$. Hence, $\psitwonorm{\bX} \leq \psitwonormstar{\bX} \leq \sqrt{d} \psitwonorm{\bX}$.
    Further, for any \tcr{matrix} 
    $\bM \in \R^{d \times d}$, $\psitwonorm{\bM \bX} \leq \MLinfnorm{\bM} \psitwonorm{\bX} \leq \sqrt{d} \MLtwonorm{\bM} \psitwonorm{\bX}$ and $\psitwonormstar{\bM \bX} \leq \MLtwonorm{\bM} \psitwonormstar{\bX} \leq \sqrt{d} \MLtwonorm{\bM} \psitwonorm{\bX}$.
\par\smallskip

\item Suppose $\E(\bX) = \mathbf{0}$, $\psitwonorm{\bX} \leq \sigma$ and assume further that the coordinates \tcr{$\{X_j\}_{j=1}^d$} 
of $\bX$ are independent. Then for any $\bv \in \R^d$, $\psitwonorm{\bv^T\bX} \leq 2\sqrt{2} \sigma \Ltwonorm{\bv}$. Thus, under these additional assumptions on $\bX$, it holds that $\psitwonorm{\bX} \leq \psitwonormstar{\bX} \leq 2\sqrt{2} \psitwonorm{\bX}$. Further, for any $\bM \in \R^{d \times d}$, $\psitwonorm{\bM \bX} \leq \psitwonormstar{\bM \bX} \leq \MLtwonorm{\bM} \psitwonormstar{\bX} \leq 2 \sqrt{2} \MLtwonorm{\bM} \psitwonorm{\bX}$.
\end{enumerate}
\end{lemma}

\begin{lemma}[Concentration bounds for sums of independent sub-Gaussian variables]\label{lem:3:sgconc}
For any $n \geq 1$, let $\{X_i\}_{i=1}^n$ be independent (not necessarily i.i.d.) random variables with means $\{\mu_i\}_{i=1}^n$ and $\underset{1 \leq i \leq n} {\max} \psitwonorm{X_i - \mu_i} \leq \sigma $ for some constant $\sigma \geq 0$. Then, for any collection of real numbers $\{a_i\}_{i=1}^n$ and letting $\mathbf{a} = (a_1,\hdots,a_n) \in \R^n$, we have:
    \begin{eqnarray*}
    &&\E \left[ \exp\left\{t \sum_{i=1}^n a_i (X_i -\mu_i) \right\}\right]  \;\leq \; \exp\left(2 \sigma^2 t^2 \| \mathbf{a} \|_2^2\right) \quad \forall \; t \in \R, \quad \mbox{and} \\
    &&\P\left\{ \left| \sum_{i=1}^n a_i (X_i - \mu_i) \right| \; > \; \epsilon \right\} \;\; \leq \; 2 \exp\left\{ - \epsilon^2/ \left(8 \sigma^2 \|\mathbf{a} \|_2^2 \right)\right\} \quad \forall \; \epsilon \geq 0.
    \end{eqnarray*}
    In particular, when $a_i = 1/n$, we have: $\psitwonorm{\frac{1}{n} \sum_{i=1}^n (X_i - \mu_i)} \leq (4\sigma)/\sqrt{n}$, \tcr{and for any $\epsilon \geq 0$,
    $ \P \left\{ \left|\frac{1}{n} \sum_{i=1}^n (X_i - \mu_i)\right| >  \epsilon \right\} \leq  2 \exp\left\{- n \epsilon^2 / \left(8 \sigma^2 \right)\right\}$.} 
\end{lemma}

\subsection{Basic Matrix Inequalities and Deviation Bounds for Random Matrices under the Spectral Norm}\label{subsec:matbound}

We provide here a sequence of lemmas collecting some useful and fairly well known inequalities regarding matrix norms and spectral properties of matrices and their submatrices and inverses. The lemmas also include some important results such as deterministic inequalities relating spectral distance between inverses of two p.d. matrices to that between the original matrices, as well as exact concentration bounds for deviations (under the spectral norm) of covariance-type random matrices defined by sub-gaussian random vectors.

\begin{lemma}[Basic inequalities on matrix norms and spectral properties of submatrices]\label{lem:5:matbasic}
Let $\bM \in \R^{d \times d}$ $(d \geq 1)$ denote any generic square matrix. Then, $\MLinfnorm{\bM} \leq \sqrt{d} \MLtwonorm{\bM} \leq d \MLinfnorm{\bM}$ and $\Mmaxnorm{\bM} \leq \MLtwonorm{\bM} \leq d \Mmaxnorm{\bM}$. Further, the following results hold.
\begin{enumerate}[(i)]
\item $\Ltwonorm{\bM \bv} \leq \MLtwonorm{\bM} \Ltwonorm{\bv}$ and $\Linfnorm{\bM \bv} \leq \MLinfnorm{\bM} \Linfnorm{\bv}$ for any $\bv \in \R^d$. Further, for any $\bM_1, \bM_2 \in \R^{d \times d}$, $\MLtwonorm{\bM_1 \bM_2} \leq \MLtwonorm{\bM_1} \MLtwonorm{\bM_2}$ and $\MLinfnorm{\bM_1 \bM_2} \leq \MLinfnorm{\bM_1} \MLinfnorm{\bM_2}$.
\par\smallskip
\item Let $\bM $ be symmetric and let $\bM_k$ denote any principal submatrix of $\bM$ of order $k \leq d$. Let $\lambda_1 \geq \hdots \geq \lambda_d$ and $\mu_1 \geq \hdots \geq \mu_k$ respectively denote the ordered eigenvalues of $\bM$ and $\bM_k$. Then, these are `interlaced' as: $\lambda_{d - k + 1} \leq \mu_i \leq \lambda_i$ $\forall \; 1 \leq i \leq k$. 

\par\smallskip
\item For any $\bM \in \R^{d \times d}$ (not necessarily symmetric) and any square submatrix (not necessarily principal) $\bM_k$ of $\bM$ of order $k \leq d$, let $\lambda_1 \geq \hdots \geq \lambda_d$ and $\mu_1 \geq \hdots \geq \mu_k$ respectively denote the ordered singular values of $\bM$ and $\bM_k$. Then, we have the `upper' and `lower' interlacing(s): $\mu_i \leq \lambda_i$ $\forall \; 1 \leq i \leq k$, and $\mu_i \geq \lambda_{2d - 2k + i} $ $\forall \; 1 \leq i \leq (2k - d)$.
\end{enumerate}
\end{lemma}
A few remarks regarding Lemma \ref{lem:5:matbasic} (ii)\tcR{--}(iii) are in order. The interlacing inequalities in (ii) are special cases of the well known Poincare Separation Theorem (and more generally, the Courant-Fisher Min-Max Theorem). The particular case of $k = d-1$ is also known as the Cauchy Interlacing Theorem (see \citet{Thompson72} for further details). Note that these inequalities are only for the eigenvalues (not singular values) of symmetric matrices and their principal submatrices (for n.n.d. matrices however, these two coincide).  The inequalities in (iii) are adopted from \citet{Thompson72} (they also apply more generally to non-square matrices). Notably, they apply directly to singular values (not eigenvalues) of matrices and submatrices of arbitrary nature and order. Among other implications, they also establish that $\MLtwonorm{\bM^*} \leq \MLtwonorm{\bM}$ for arbitrary matrices $\bM$ and submatrices $\bM^*$ of $\bM$.

\begin{lemma}[Inequalities relating spectral deviations of p.d. matrices and their inverses]\label{lem:6:matinv}
Let $\bM_0 \in \R^{d \times d}$ be any symmetric positive definite matrix with inverse $\bM_0^{-1}$ and minimal eigenvalue (also singular value) $\lambda_{\min}(\bM_0) \equiv \MLtwonorm{\bM_0^{-1}}^{-1} > 0$. Let $\bM \in \R^{d \times d}$ be any matrix such that $\MLtwonorm{\bM - \bM_0} \leq \lambda_{\min}(\bM_0)$. Then, $\MLtwonorm{(\bM - \bM_0)\bM_0^{-1}} < 1$, and $\{I + (\bM - \bM_0)\bM_0^{-1}\}$ and $\bM$ are both invertible. Further,
\begin{eqnarray*}
&& \MLtwonorm{\bM^{-1} - \bM_0^{-1}} 
 \;\; \leq \;\; \frac{\lambda_{\min}^{-2}(\bM_0)}{1 - \MLtwonorm{\bM - \bM_0}\lambda_{\min}^{-1}(\bM_0)}  \MLtwonorm{\bM - \bM_0} \\
 && \qquad\qquad \leq \; 2\lambda_{\min}^{-2}(\bM_0) \MLtwonorm{\bM - \bM_0} \quad \mbox{if} \;\; \MLtwonorm{\bM - \bM_0} \; \leq \; \frac{1}{2}\lambda_{\min}(\bM_0).
\end{eqnarray*}
\end{lemma}

\begin{lemma}[Deviation bounds under the spectral norm for covariance-type matrices]\label{lem:7:matconc}
Let $\bX \in \R^d$ be any random vector with $\E(\bX) = \mathbf{0}$ and $\psitwonormstar{\bX} \leq \sigma_*$ for some $\sigma_* \geq 0$. Let $\bSigma := \E(\bX\bX^T)$ which is assumed to be positive definite with minimum and maximum eigenvalues $\lambdamin(\bSigma) > 0$ and $\lambdamax(\bSigma) \equiv \MLtwonorm{\bSigma} \geq \lambdamin(\bSigma) > 0$ respectively. Consider a collection $\{\bX_i \}_{i=1}^n$ of $n \geq 1$ independent realizations of $\bX$. Then, for any $\epsilon \geq 0$, we have:
\begin{eqnarray*}
&& \P \left\{ \MLtwonorm {\frac{1}{n}\sum_{i=1}^n\bX_i \bX_i^T - \bSigma}  > \; C_1 K_{\bX}^2 \left( \sqrt{\frac{d + \epsilon}{n}} + \frac{d + \epsilon}{n} \right) \right\} \; \leq \; 2 \exp\left( - \epsilon \right)\\ 
&& \mbox{and} \;\; \E\left( \MLtwonorm {\frac{1}{n}\sum_{i=1}^n\bX_i \bX_i^T - \bSigma} \right) \; \leq \; C_2 K_{\bX}^2 \left( \sqrt{\frac{d}{n}} + \frac{d}{n} \right), \quad \mbox{where} 
\end{eqnarray*}
$K_{\bX}^2 := \frac{\sigma_*^2 \lambdamax(\bSigma)}{\lambdamin(\bSigma)}$, and $C_1, C_2 > 0$ are absolute constants that do not depend on any other quantities introduced above. Specifically, choosing $\epsilon = cd$ for any $c > 0$ and noting that $\sqrt{c+1} \leq c+1$, we have: for any $c > 0$,
\begin{equation*}
\P \left\{ \MLtwonorm {\frac{1}{n}\sum_{i=1}^n\bX_i \bX_i^T - \bSigma} \; > \; C_1 K_{\bX}^2 (c+1)\left( \sqrt{\frac{d}{n}} + \frac{d}{n} \right) \right\} \; \leq \; 2 \exp\left( - cd\right). 
\end{equation*}
\end{lemma}

Lemma \ref{lem:6:matinv} is adopted from (the proof of) Lemma 5 in \citet{HarrisDrton13}. Lemma \ref{lem:7:matconc} is obtained using Theorem 4.7.1 (more fundamentally, Theorem 4.6.1) of \citet{Vershynin18}, in conjunction with Exercise 4.7.3 therein, along with appropriate modifications of his notations and assumptions to adapt to our setting. \tcr{Similar} results, though slightly more involved and with less explicit constants, may also be obtained using Theorem 5.39 of \citet{Vershynin12}, along with equation (5.26) in Remark 5.40 therein.


\section{\tcr{Verifying the Moment Conditions in Remark \ref{remark: linear functionals}}}\label{sec:remarkverify}
{We provide here some discussions regarding verification of the moment conditions: $\rhoxiAS = O(1)$ and $\sigmaxiAS = \Omega(1)$ \tcr{required in the last part of} 
Remark \ref{remark: linear functionals}.
The first condition \tcr{(and more)} can \tcr{indeed be} verified generally under our basic assumptions \tcr{in Section \ref{subsec:linreg}.} 
To this end, note that under Assumption \ref{linreg:main_assmpn} (i), and through multiple uses of Lemma  \ref{lem:2:sgvectprop} (i), Lemma \ref{lem:1:genprop} (i) and Lemma \ref{lem:1:genprop} (iv), as well as the last claim in Lemma \ref{linreg:main_lem1}, we have: for each $S \in \Ssc$,
\begin{eqnarray*}
&& \psionenorm{\baS^T \bPsiS(\bZ)} \;\; \leq \;\; \psitwonorm{\baS^T \bSigmaS^{-1} \bXtilS} \psitwonorm{\Ytil - \bXtilS^T \bbetaS} \nonumber  \\ 
&& \quad \;\; \leq \;\; \Ltwonorm{\baS} \{\lambdamin(\bSigmaS)\}^{-1} \psitwonormstar{\bXS} (\sigmaY + \psitwonormstar{\bXS}\Ltwonorm{\bbetaS})  \nonumber \\
&& \quad \;\; \leq \;\; \Ltwonorm{\baS} \left\{\lambdainfSsc^{-1} \sigmaXSsc (\sigmaY + \sqrt{2} \sigmaXSsc \lambdainfSsc^{-1/2} \sigmaY) \right\} \;\; \equiv \; \Ltwonorm{\baS} D_{\Ssc} \; \mbox{(say)}, \nonumber
\end{eqnarray*}
where $D_{\Ssc} := \lambdainfSsc^{-1} \sigmaXSsc \sigmaY (1 + \sqrt{2} \sigmaXSsc \lambdainfSsc^{-1/2}) $ depends only on the constants in Assumption \ref{linreg:main_assmpn} (i). Thus, as long as $D_{\Ssc} = O(1)$ and $\sum_{S \in \Ssc}\Ltwonorm{\baS} = O(1)$, as assumed before, we have:
\begin{equation*}
 \psionenorm{\xi_{\ASsc}(\bZ)} \; \equiv \; \psionenorm{\sum_{S \in \Ssc} \baS^T \bPsiS(\bZ)} \; \leq \; D_{\Ssc} \left( \sum_{S \in \Ssc} \Ltwonorm{\baS} \right) \; = \; O(1). \qed 
\end{equation*}
Consequently, using Lemma \ref{lem:1:genprop} (iii), we have $0 \leq \sigmaxiAS \leq \rhoxiAS \leq O(1)$. Among other implications, \tcr{this verifies the first condition: $\rhoxiAS = O(1)$. \qed} 
}

\par\smallskip
\tcr{Next,} we provide some sufficient conditions for verifying the other moment condition: $\sigmaxiAS = \Omega(1)$. To this end, suppose that for some positive constant $\eta_{\Ssc} = \Omega(1)$, we have $\Var(Y \mid \cup_{S \in \Ssc} \bXS) \geq \eta_{\Ssc} > 0$. Then,
\begin{align*}
& \sigmaxiAS^2 \; \equiv \; \Var \{ \xi_{\ASsc}(\bZ)\} \;\; \geq \; \E \left[ \Var \left\{ \xi_{\ASsc}(\bZ) \mid  \cup_{S \in \Ssc} \bXS \right\} \right] \nonumber \\
& \qquad \; \equiv \; \E \left[ \Var \left\{ \sum_{S \in \Ssc} \baS^T \bSigmaS^{-1} \bXtilS (\Ytil - \bXtilS^T\bbetaS) \mid  \cup_{S \in \Ssc} \bXS \right\} \right] \nonumber \\
&= \; \E \left\{ \Var\left(Y \mid \cup_{S \in \Ssc} \bXS\right) \left( \sum_{S \in \Ssc} \baS^T \bSigmaS^{-1} \bXtilS \right)^2 \right\} \; \geq \; \eta_{\Ssc} \E \left( \sum_{S \in \Ssc} \baS^T \bSigmaS^{-1} \bXtilS \right)^2.  \nonumber 
\end{align*}
Hence, as long as $\E \{ ( \sum_{S \in \Ssc} \baS^T \bSigmaS^{-1} \bXtilS )^2 \} = \Omega(1)$, and $\eta_{\Ssc} = \Omega(1)$ as assumed, we have $\sigmaxiAS^2 = \Omega(1)$, \tcr{thereby verifying the second condition.} \qed 

\section{False Discovery Rate (FDR) Control for MIDA}\label{sec:FDR_Control}



As an additional validation to our asymptotic results \tcr{on the theoretical properties and inferential tools for MIDA,} we discuss \tcr{here} some \tcr{numerical results on FDR control for MIDA, based on the setting used for our simulation studies in Section \ref{section: simulations}, for estimating the set of significant mediators:} 
\tcm{$\cup_{r=1}^{m}S^{(r)} := \cup_{r=1}^{m}\{ X_j^{(r)} : \eta_j^{(r)} \neq 0,~ j = 2, \ldots, p-1\}$} (\texttt{Target}) and
\tcm{$\cup_{r=1}^{m}S^{*(r)} := \cup_{r=1}^{m}\{ X_j^{(r)} : \theta_{1j}^{(r)}\mathrm{aver}(\Theta_{jp}^{(r)}) \neq 0,~ j = 2, \ldots, p-1\}$} (\texttt{Target\_CPDAG}),
when the true CPDAG is known as well as when the CPDAG is estimated. The BH procedure at \tcr{a} level $\alpha$ (asymptotically) guarantees to control \tcr{the} FDR at 
level 
\tcm{$\alpha m_0/M$} for estimating \texttt{Target\_CPDAG}, where for each simulation setting, \tcm{$m_0$} denotes the total number of true hypotheses 
\tcm{$|\cup_{r=1}^{m}S^{*(r)}|$ among the $M = m \times (p-2) = 5000$ hypotheses.}
Since 
\tcm{$\cup_{r=1}^{m}S^{*(r)} \subseteq \cup_{r=1}^{m}S^{(r)}$}, it is expected that the empirical FDR level would be higher when it is measured with respect to \texttt{Target}.

\begin{figure}[!ht]   
\centering
\includegraphics[width = \textwidth]{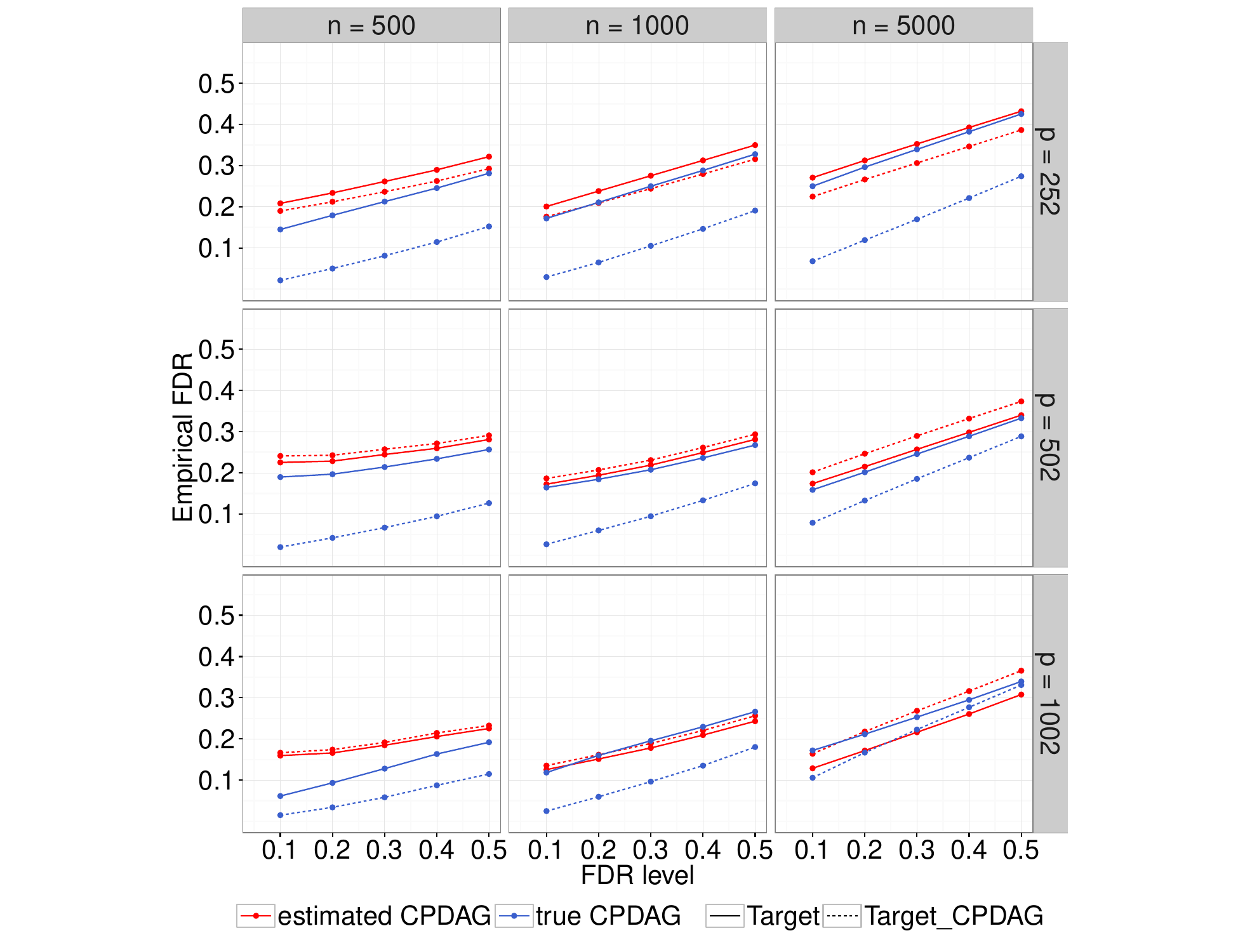}
\caption{Empirical FDR of the estimated sets based on the BH procedure \tcr{\emph{without} any $p$-value} screening for estimating \texttt{Target} and \texttt{Target\_CPDAG} when the true CPDAG is known as well as when the CPDAG is estimated.}
\label{fig: FDR curves_noscreen}
\end{figure}

\par\smallskip
Figure \ref{fig: FDR curves_noscreen} shows that the BH procedure becomes \tcr{quite} \emph{conservative} for estimating
\tcm{$\cup_{r=1}^{m}S^{*(r)}$,} 
\tcr{though} we ignore the additional adjustment suggested by \cite{BenjaminiYekutieli01} in order to correct for possible dependencies among \tcr{the} hypotheses \tcr{here}. The conservativeness of \tcr{the} BH procedure can be attributed to the fact
\tcr{(a consequence of Theorem \ref{theorem: asymptotic distribution}, as was discussed in Section \ref{subsection: inference})} that the \tcr{$p$-value} corresponding to the test: 
\tcm{$\theta_{1j}^{(r)} \mathrm{aver}(\Theta_{jp}^{(r)}) = 0$}
has a stochastically larger distribution than $\mathrm{Uniform}[0,1]$ when both 
\tcm{$\mathrm{aver}(\Theta_{jp}^{(r)})$ and $\theta_{1j}^{(r)}$} are zero.
In order to mitigate this issue, we apply a heuristic screening, where\tcr{by} we \tcr{first} select the potential mediators for which the total effect of the treatment $X_1$ on the mediator is non-zero, by testing 
\tcm{$\theta_{1j}^{(r)} = 0$} at the significance level $0.01$. Then, we apply the BH procedure on \tcr{this} selected set. Figure \ref{fig: FDR curves} shows \tcr{the empirical FDR of the estimated sets based on the BH procedure after this screening, and demonstrates} that the heuristic screening method above is indeed effective in reducing the conservativeness of \tcr{the} BH procedure \tcr{in controlling the FDR level for multiple testing using MIDA.}



\begin{figure}[!ht]
\centering
\includegraphics[width = \textwidth]{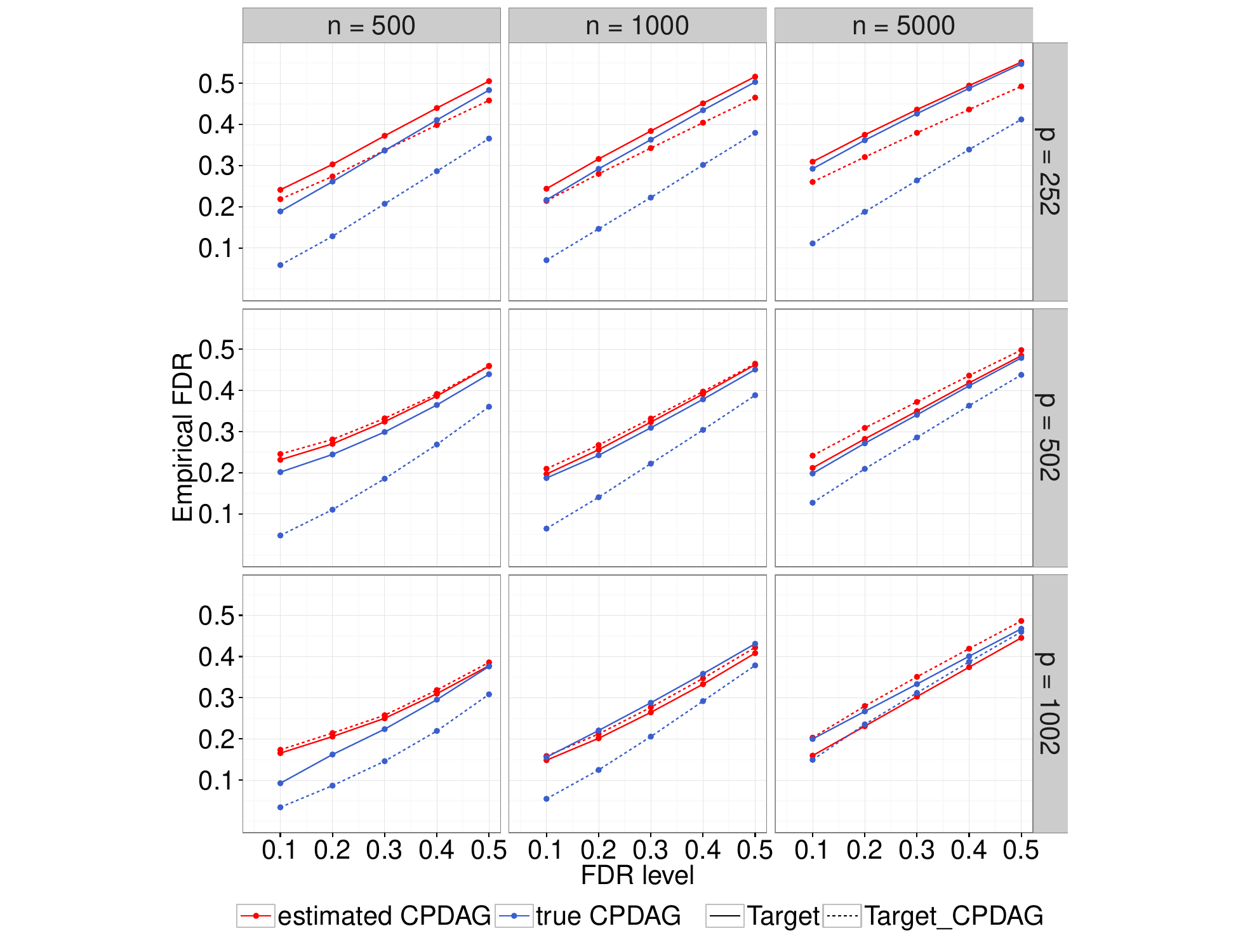}
\caption{Empirical FDR of the estimated sets based on the BH procedure \tcr{\emph{with} $p$-value} screening for estimating \texttt{Target} and \texttt{Target\_CPDAG} when the true CPDAG is known as well as when the CPDAG is estimated.}
\label{fig: FDR curves}
\end{figure}

\appendix   


\bibliographystyle{apalike}
\bibliography{bib/Mybibliography}

\end{document}